\documentclass[letterpaper,11pt]{article}

\usepackage[letterpaper, left=2.5cm, right=2.5cm, top=2.5cm, bottom=2.5cm]{geometry}
\usepackage{amsmath}
\usepackage{amsthm}
\usepackage{mathtools} 
\usepackage{booktabs} 
\usepackage{amsmath}
\usepackage{amssymb}
\usepackage{enumerate}
\usepackage{bm}
\usepackage{enumitem}
\newcommand{\ignore}[1]{{}}

\DeclareMathOperator*{\argmin}{arg\,min}

\newtheorem{theorem}{Theorem}
\newtheorem{lemma}{Lemma}

\usepackage{url}
\usepackage{hyperref}
\usepackage{stackrel}
\usepackage{mathtools}
\usepackage{longtable}
\usepackage{tikz}
\usetikzlibrary{calc}

\setlistdepth{5}
\newlist{myEnumerate}{enumerate}{5}
\setlist[myEnumerate,1]{label=\Roman*., leftmargin=0.5cm}
\setlist[myEnumerate,2]{label=(\Alph*), leftmargin=0.5cm}
\setlist[myEnumerate,3]{label=(\alph*), leftmargin=0.5cm}
\setlist[myEnumerate,4]{label=(\roman*), leftmargin=0.5cm}
\setlist[myEnumerate,5]{label=\textbf{Case \arabic*:}, leftmargin=1.25cm, itemsep=0.15cm}

\newcommand\numberthis{\addtocounter{equation}{1}\tag{\theequation}} 
\newcommand{\LCP}{\ensuremath{\text{LCP}}}
\newcommand{\alignstack}[2]{\stackrel{\mathmakebox[\widthof{\ensuremath{#2}}]{#1}}{#2}}
\newcommand{\alignstacktinytext}[2]{\stackrel{\mathmakebox[\widthof{\ensuremath{#2}}]{\text{\tiny #1}}}{#2}}

\DeclareMathOperator{\fpart}{frac}
  
\begin{document}

\title{Optimal Algorithms for Right-Sizing Data Centers ---\\ Extended Version\footnote{Work supported by the European Research Council, Grant Agreement No.\ 691672.}}

\author{Susanne Albers \\
	Technical University of Munich \\
	albers@in.tum.de \\
	\and
	Jens Quedenfeld\footnote{Contact author} \\
	Technical University of Munich \\
	jens.quedenfeld@in.tum.de \\
}

\maketitle

\begin{abstract}
	Electricity cost is a dominant and rapidly growing expense in data centers. 
	Unfortunately, much of the consumed energy is wasted because servers are idle 
	for extended periods of time. We study a capacity management problem that
	dynamically right-sizes a data center, matching the number of active servers
	with the varying demand for computing capacity. We resort to a data-center
	optimization problem introduced by Lin, Wierman, Andrew and Thereska~\cite{LinWierman2011infocom,LinWierman2013}
	that, over a time horizon, minimizes a combined objective function consisting 
	of operating cost, modeled by a sequence of convex functions, and server switching
	cost. All prior work addresses a continuous setting in which the number of active
	servers, at any time, may take a fractional value.
	
	In this paper, we investigate for the first time the discrete data-center optimization 
	problem where the number of active servers, at any time, must be integer valued. Thereby
	we seek truly feasible solutions. First, we show that the offline problem can be solved 
	in polynomial time. Our algorithm relies on a new, yet intuitive graph theoretic model
	of the optimization problem and performs binary search in a layered graph. Second,
	we study the online problem and extend the algorithm {\em Lazy Capacity Provisioning\/} (LCP) 
	by Lin et al.~\cite{LinWierman2011infocom,LinWierman2013} to the discrete setting.  We prove that LCP is 3-competitive. 
	Moreover, we show that no deterministic online algorithm can achieve a competitive ratio 
	smaller than~3. Hence, while LCP does not attain an optimal competitiveness in the continuous
	setting, it does so in the discrete problem examined here. We prove that the lower
	bound of~3 also holds in a problem variant with more restricted operating cost functions,
	introduced by Lin et al.~\cite{LinWierman2011infocom}.  
	
	In addition, we develop a randomized online algorithm that is 2-competitive against an oblivious adversary. It is based on the algorithm of Bansal et al.~\cite{Bansal2015} (a deterministic, 2-competitive algorithm for the continuous setting) and uses randomized rounding to obtain an integral solution. Moreover, we prove that 2 is a lower bound for the competitive ratio of randomized online algorithms, so our algorithm is optimal. We prove that the lower bound still holds for the more restricted model. 
	
	Finally, we address the continuous setting and give a lower bound of~2 on the best competitiveness 
	of online algorithms. This matches an upper bound by Bansal et al.~\cite{Bansal2015}. A lower bound 
	of~2 was also shown by Antoniadis and Schewior~\cite{Antoniadis2017}. We develop an independent proof 
	that extends to the scenario with more restricted operating cost.
\end{abstract}

\section{Introduction}
Energy conservation in data centers is a major concern for both operators and the environment. 
In the U.S.,\ about 1.8\% of the total electricity consumption is attributed to data centers~\cite{Shehabi2016}.
In 2015, more than 416 TWh (terawatt hours) were used by the world's data centers, which exceeds the total 
power consumption in the UK~\cite{Bawden2016}. Electricity cost is a significant expense in data centers~\cite{Dayarathna2016};
about 18--28\% of their budget is invested in power \cite{Hamilton2008, Brill2007}.
Remarkably, the servers of a data center are only utilized 12--40\% of the time on average~\cite{Delforge2014,Armbrust2009,Barroso2007}. 
Even worse, when idle and  in active mode, they consume about half of their peak power \cite{Schmid2009power}. Hence, a promising 
approach for energy conservation and capacity management is to transition idle servers into low-power sleep 
states. However, state transitions, and in particular power-up operations, also incur energy/cost. Therefore, 
dynamically matching the number of active servers with the varying demand for computing capacity is a 
challenging optimization problem. In essence, the goal is to right-size a data center over time so as 
to minimize energy and operation costs. 

\vspace*{0.1cm}

{\bf Problem Formulation.} We investigate a basic algorithmic problem with the objective of dynamically 
resizing a data center. Specifically, we resort to a framework that was introduced by Lin, Wierman, Andrew 
and Thereska~\cite{LinWierman2011infocom,LinWierman2013} and further explored, for instance, in~\cite{Antoniadis2016,Antoniadis2017,Bansal2015,Andrew2013,Wang2015,LinWierman2011,LiuLinWierman2015,Zhang2018,Antoniadis2020}.

Consider a data center with $m$ homogeneous servers, each of which has an active state and a sleep state. 
An optimization is performed over a discrete finite time horizon consisting of time steps $t=1,\ldots, T$. At 
any time $t$, $1\leq t\leq T$, a non-negative convex cost function $f_t(\cdot)$ models the operating cost of the
data center. More precisely, $f_t(x_t)$ is the incurred cost if $x_t$ servers are in the active state at time $t$,
where $0\leq x_t\leq m$. 
This operating cost captures, e.g., energy cost and service delay, for an incoming
workload, depending on the number of active servers. Furthermore, at any time $t$ there is a
switching cost, taking into account that the data center may be resized by changing the number of active
servers. This switching cost is equal to $\beta(x_t-x_{t-1})^+$, where $\beta$ is a positive real constant
and $(x)^+=\max(0,x)$. Here we assume that transition cost is incurred when servers are powered up from the
sleep state to the active state. A cost of powering down servers may be folded into this cost. 
The constant 
$\beta$ incorporates, e.g., the energy needed to transition a server from the sleep state to the active state, 
as well as delays resulting from a migration of data and connections. We assume that at the beginning 
and at the end of the time horizon all servers are in the sleep state, i.e., $x_0=x_{T+1}=0$. The goal
is to determine a vector $X=(x_1,\ldots,x_T)$ called \emph{schedule}, specifying at any time the number of active servers, that
minimizes
\begin{equation}
\sum_{t=1}^T f_t(x_t) + \beta \sum_{t=1}^T (x_t-x_{t-1})^+.
\label{eqn:model:cost}
\end{equation}
In the offline version of this data-center optimization problem, the convex functions $f_t$, $1\leq t \leq T$, 
are known in advance. In the online version, the $f_t$ arrive over time. At time $t$, function $f_t$ is presented. 
Recall that the operating cost at time $t$ depends for instance on the incoming workload, which becomes known
only at time $t$. 

All previous work on the data-center optimization problem assumes that the server numbers~$x_t$, $1\leq t\leq T$,
may take fractional values. That is, $x_t$ may be an arbitrary real number in the range~$[0,m]$. From a practical
point of view this is acceptable because a data center has a large number of machines. Nonetheless, from an algorithmic
and optimization perspective, the proposed algorithms do not compute feasible solutions. Important questions remain
if the $x_t$ are indeed integer valued: (1)~Can optimal solutions be computed in polynomial time? (2)~What is the
best competitive ratio achievable by online algorithms? In this paper, we present the first study of the data-center
optimization problem assuming that the $x_t$ take integer values and, in particular, settle questions~(1) and (2). 

\vspace*{0.1cm}

{\bf Previous Work.} As indicated above, all prior work on the data-center optimization problem assumes that the 
$x_t$, $1\leq t \leq T$, may take fractional values in $[0,m]$. First, Lin et al.~\cite{LinWierman2013} consider the offline problem.
They develop an algorithm based on a convex program that computes optimal solutions. Second, 
Lin et al.~\cite{LinWierman2013} study the online problem. They devise a deterministic algorithm called {\em Lazy Capacity 
Provisioning (LCP)\/} and prove that it achieves a competitive ratio of exactly~3. Algorithm LCP, at any time $t$, computes
a lower bound and an upper bound on the number of active servers by considering two scenarios in which the switching
cost $\beta$ is charged, either when a server is powered up or when it is powered down. The LCP algorithm lazily stays within
these two bounds. The tight bound of~3 on the competitiveness of LCP also holds if the algorithm has a finite prediction
window $w$, i.e., at time $t$ it knows the current as well as the next $w$ arriving functions $f_t, \dots, f_{t+w}$.
Furthermore, Lin et al.~\cite{LinWierman2013} perform an experimental study with two real-world traces evaluating
the savings resulting from right-sizing in data centers.

Bansal et al.~\cite{Bansal2015} presented a 2-competitive online algorithm and showed that no deterministic or randomized
online strategy can attain a competitiveness smaller than 1.86. Recently, Antoniadis and Schewior~\cite{Antoniadis2017} improved
the lower bound to~2. Bansal et al.~\cite{Bansal2015} also gave a 3-competitive memoryless 
algorithm and showed that this is the best competitive factor achievable by a deterministic memoryless algorithm.
The data-center optimization problem is an online convex optimization problem with switching costs. 
Andrew et al.~\cite{Andrew2013} showed that there is an algorithm with sublinear regret but that $\mathcal{O}(1)$-competitiveness
and sublinear regret cannot be achieved simultaneously. 

The continuous data-center optimization problem on \emph{heterogeneous} data centers (that contain different server types) is special case of convex function chasing where the values $x_t$ are points in a metric space. Sellke~\cite{Sellke2020} presented a $(d+1)$-competitive online algorithm for convex function chasing. A similar result was found by Argue et al.~\cite{Argue2020}. Goel and Wierman~\cite{GoelWierman2018} developed an algorithm called Online Balanced Descent (OBD) that achieves a competitive ratio of $3 + \mathcal{O}(1/\mu)$ if the arriving operating cost functions are $\mu$-strongly convex. Chen et al. ~\cite{ChenGoelWierman2018} showed that OBD is $(3 + \mathcal{O}(1/\alpha))$-competitive if the functions are locally $\alpha$-polyhedral. Other publications handling convex function chasing and related problems are \cite{Antoniadis2016,BubeckSellke2020nested}.

Further work on energy conservation in data center includes, for instance, \cite{Khuller2010,LiKhuller2011}. Khuller et al.~\cite{Khuller2010} 
introduce a machine activation problem. 
There exists an activation cost budget and jobs have to be scheduled on the 
selected, activated machines so as to minimize the makespan. They present algorithms that simultaneously approximate 
the budget and the makespan. A second paper by Li and Khuller~\cite{LiKhuller2011} considers a generalization where the 
activation cost of a machine is a non-decreasing function of the load. In the more applied computer science 
literature, power management strategies and the value of sleep states have been studied extensively. The papers
 focus mostly on experimental evaluations. Articles that also present analytic results include~\cite{Gandhi2011,Gandhi2010,Haas2015}.

\vspace*{0.1cm}

{\bf Our Contribution.} We conduct the first investigation of the {\em discrete\/} data-center optimization problem,
where the values $x_t$, specifying the number of active servers at any time $t\in\{1,\ldots,T\}$, must be integer 
valued. Thereby, we seek truly feasible solutions. 

First, in Section~\ref{sec:poly} we study the offline algorithm. We show that optimal solutions can be computed in polynomial
time. Our algorithm is different from the convex optimization approach by Lin et al.~\cite{LinWierman2013}. We propose a new, yet
natural graph-based representation of the discrete data-center optimization problem. We construct a grid-structured graph 
containing a vertex $v_{t,j}$, for each $t\in \{1,\ldots, T\}$ and $j\in\{0,\ldots,m\}$. Edges represent right-sizing operations,
i.e., changes in the number of active servers, and are labeled with operating and switching costs. An optimal solution
could be determined by a shortest path computation. However, the resulting algorithm would have a pseudo-polynomial running
time. Instead, we devise an algorithm that improves solutions iteratively using binary search.
In each iteration the algorithm uses only a constant number of graph layers. The resulting running time is $\mathcal{O}(T\log m)$.

The remaining paper focuses on the online problem and develops tight bounds on the competitiveness. 
In Section~\ref{sec:lcp}, we adapt the LCP algorithm by Lin et al.~\cite{LinWierman2013} to the discrete data-center optimization
problem. We prove that LCP is 3-competitive, as in the continuous setting. We remark that our analysis is different from
that by Lin et al.~\cite{LinWierman2013}. Specifically, our analysis resorts to the discrete structure of the problem and identifies respective
properties. The analysis by Lin et al.~\cite{LinWierman2013} relates to their convex optimization approach that characterizes
optimal solutions in the continuous setting. 

In Section~\ref{sec:random}, we develop a randomized online algorithm which is 2-competitive against an oblivious adversary. It is based on the algorithm of Bansal et al.~\cite{Bansal2015} that achieves a competitive ratio of 2 for the continuous setting. Our algorithm works as follows. First, it extends the given discrete problem instance to the continuous setting. Then, it calculates a 2-competitive fractional schedule by using the algorithm of Bansal et al. Finally, we round the fractional schedule randomly to obtain an integral schedule. By using the right rounding technique it can be shown that the resulting schedule is 2-competitive according to the original discrete problem instance.

In Section~\ref{sec:lower}, we devise lower bounds. We prove that no deterministic online algorithm can obtain a competitive
ratio smaller than~3. Hence, LCP achieves an optimal competitive factor. Interestingly, while LCP does not attain an
optimal competitiveness in the continuous data-center optimization problem (where the $x_t$ may take fractional values), 
it does so in the discrete problem (according to deterministic algorithms).  
We prove that the lower bound of~3 on the best possible competitive ratio also holds
for a more restricted setting, originally introduced by Lin et al.~\cite{LinWierman2011infocom} in the conference publication of their paper.
Specifically, the problem is to find a vector $X=(x_1,\ldots,x_T)$ that minimizes
\begin{equation}
\sum_{t=1}^T x_t f(\lambda_t/x_t) + \beta \sum_{t=1}^T (x_t-x_{t-1})^+,
\label{eqn:model:lin}
\end{equation}
subject to $x_t\geq \lambda_t$, for $t\in \{1,\ldots,T\}$. Here $\lambda_t$ is the incoming workload at time $t$ and $f(z)$
is a non-negative convex function representing the operating cost of a single server running with load $z \in [0,1]$. Since $f$ is convex, 
it is optimal to distribute the jobs equally to all active servers, so that the operating cost at time $t$ is $x_t f(\lambda_t / x_t)$. 
This problem setting is more restricted in that there is only a single function $f$ modeling operating cost over the time horizon. Nonetheless,
it is well motivated by real data center environments.  

Furthermore, in Section~\ref{sec:lower}, we address the continuous data-center optimization problem and prove that no deterministic online 
algorithm can achieve a competitive ratio smaller than 2. The same result was shown by Antoniadis and Schewior~\cite{Antoniadis2017}. We develop an
independent proof that can again be extended to the more restricted optimization problem stated in~(\ref{eqn:model:lin}), i.e.,
the lower bound of~2 on the best competitiveness holds in this setting as well.

In addition, we show that there is no randomized online algorithm with a competitive ratio smaller than~2, so our randomized online algorithm presented in Section~\ref{sec:random} is optimal. The construction of the lower bound uses some results of the lower bound proof for the continuous setting. Again, we show that the lower bound holds for the more restricted model.

Finally, in Section~\ref{sec:lower}, we analyze online algorithms with a finite prediction window, i.e., at time~$t$ an online algorithm knows the 
current as well as the next $w$ arriving functions $f_t, \dots, f_{t+w}$. We show that all our lower bounds, for both settings (continuous and discrete) 
and both models (general and restricted), still hold.

\section{An optimal offline algorithm} 
\label{sec:poly}

In this section we study the offline version of the discrete data-center optimization problem. We develop
an algorithm that computes optimal solutions in $\mathcal{O}(T\log m)$ time. 

\setlength{\textfloatsep}{12pt plus 1.0pt minus 2.0pt}
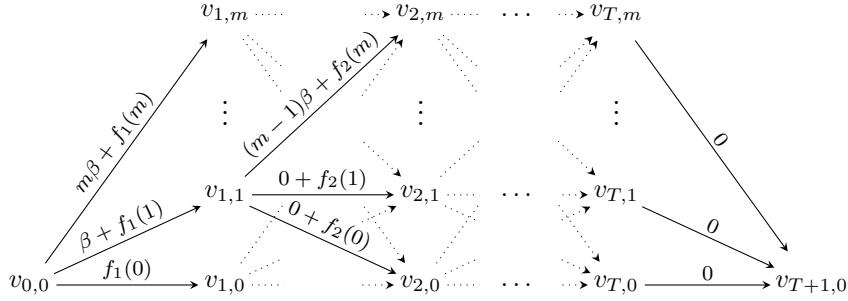
\begin{figure} 
	\setlength{\abovecaptionskip}{0pt plus 0pt minus 0pt}
	\setlength{\belowcaptionskip}{0pt plus 0pt minus 0pt}
	
	\centering
	\begin{tikzpicture}
	\pgfmathsetmacro{\xBegin}{0}
	\pgfmathsetmacro{\xDist}{2.6}
	\pgfmathsetmacro{\yBegin}{0}
	\pgfmathsetmacro{\yDist}{1.2}
	\pgfmathsetmacro{\yEnd}{\yBegin + \yDist * 3}
	\pgfmathsetmacro{\ratioDottedEdge}{0.30}
	\pgfmathsetmacro{\ratioDottedEdgeInv}{1 - \ratioDottedEdge}
	
	\tikzstyle{vertex}=[circle, inner sep=1pt, font=\small]
	\tikzstyle{edge}=[-stealth]
	\tikzstyle{edgetext}=[font=\scriptsize, , above=-2pt, sloped]

	\node[vertex] (v00) at (\xBegin, \yBegin) {$v_{0,0}$};
	
	\node[vertex] (v10) at (\xBegin + \xDist, \yBegin) {$v_{1,0}$};
	\node[vertex] (v11) at (\xBegin + \xDist, \yBegin + \yDist) {$v_{1,1}$};
	\node[] (v1x) at (\xBegin + \xDist, \yEnd/2 + \yBegin/2 + \yDist/2 ) {$\vdots$};
	\node[vertex] (v1m) at (\xBegin + \xDist, \yEnd) {$v_{1,m}$};
	
	\draw[edge] (v00) edge node[edgetext] {$ f_1(0)$} (v10);
	\draw[edge] (v00) edge node[edgetext] {$\beta + f_1(1)$} (v11);
	\draw[edge] (v00) edge node[edgetext] {$m\beta + f_1(m)$} (v1m);
	
	\node[vertex] (v20) at (\xBegin + \xDist * 2, \yBegin) {$v_{2,0}$};
	\node[vertex] (v21) at (\xBegin + \xDist * 2, \yBegin + \yDist) {$v_{2,1}$};
	\node[] (v2x) at (\xBegin + \xDist * 2, \yEnd/2 + \yBegin/2 + \yDist/2 ) {$\vdots$};
	\node[vertex] (v2m) at (\xBegin + \xDist * 2, \yEnd) {$v_{2,m}$};
	
	\draw[] (v10) edge [dotted] ($(v10)!\ratioDottedEdge!(v20)$);
	\draw[] (v10) edge [dotted] ($(v10)!\ratioDottedEdge!(v21)$);
	\draw[] (v10) edge [dotted] ($(v10)!\ratioDottedEdge!(v2m)$);
	
	\draw[] (v1m) edge [dotted] ($(v1m)!\ratioDottedEdge!(v20)$);
	\draw[] (v1m) edge [dotted] ($(v1m)!\ratioDottedEdge!(v21)$);
	\draw[] (v1m) edge [dotted] ($(v1m)!\ratioDottedEdge!(v2m)$);
	
	\draw[] (v20) edge [stealth-, dotted] ($(v10)!\ratioDottedEdgeInv!(v20)$);
	\draw[] (v21) edge [stealth-,dotted] ($(v10)!\ratioDottedEdgeInv!(v21)$);
	\draw[] (v2m) edge [stealth-,dotted] ($(v10)!\ratioDottedEdgeInv!(v2m)$);
	
	\draw[] (v20) edge [stealth-, dotted] ($(v1m)!\ratioDottedEdgeInv!(v20)$);
	\draw[] (v21) edge [stealth-,dotted] ($(v1m)!\ratioDottedEdgeInv!(v21)$);
	\draw[] (v2m) edge [stealth-,dotted] ($(v1m)!\ratioDottedEdgeInv!(v2m)$);
	
	
	\draw[edge] (v11) edge node[edgetext] {$0 + f_2(0)$} (v20);
	\draw[edge] (v11) edge node[edgetext] {$0 + f_2(1)$} (v21);
	\draw[edge] (v11) edge node[edgetext] {$(m-1)\beta + f_2(m)$} (v2m);
	
	
	\node[vertex] (vx0) at (\xBegin + \xDist * 2.5, \yBegin) {$\dots$};
	\node[vertex] (vx1) at (\xBegin + \xDist * 2.5, \yBegin + \yDist) {$\dots$};
	\node[vertex] (vxm) at (\xBegin + \xDist * 2.5, \yEnd) {$\dots$};
	
	\node[vertex] (vT0) at (\xBegin + \xDist * 3, \yBegin) {$v_{T,0}$};
	\node[vertex] (vT1) at (\xBegin + \xDist * 3, \yBegin + \yDist) {$v_{T,1}$};
	\node[] (vTx) at (\xBegin + \xDist * 3, \yEnd/2 + \yBegin/2 + \yDist/2 ) {$\vdots$};
	\node[vertex] (vTm) at (\xBegin + \xDist * 3, \yEnd) {$v_{T,m}$};
	
	\draw[] (v20) edge [dotted] ($(v20)!\ratioDottedEdge!(vT0)$);
	\draw[] (v20) edge [dotted] ($(v20)!\ratioDottedEdge!(vT1)$);
	\draw[] (v20) edge [dotted] ($(v20)!\ratioDottedEdge!(vTm)$);
	
	\draw[] (v21) edge [dotted] ($(v21)!\ratioDottedEdge!(vT0)$);
	\draw[] (v21) edge [dotted] ($(v21)!\ratioDottedEdge!(vT1)$);
	\draw[] (v21) edge [dotted] ($(v21)!\ratioDottedEdge!(vTm)$);
	
	\draw[] (v2m) edge [dotted] ($(v2m)!\ratioDottedEdge!(vT0)$);
	\draw[] (v2m) edge [dotted] ($(v2m)!\ratioDottedEdge!(vT1)$);
	\draw[] (v2m) edge [dotted] ($(v2m)!\ratioDottedEdge!(vTm)$);
	
	\draw[] (vT0) edge [stealth-, dotted] ($(v20)!\ratioDottedEdgeInv!(vT0)$);
	\draw[] (vT1) edge [stealth-,dotted] ($(v20)!\ratioDottedEdgeInv!(vT1)$);
	\draw[] (vTm) edge [stealth-,dotted] ($(v20)!\ratioDottedEdgeInv!(vTm)$);
	
	\draw[] (vT0) edge [stealth-, dotted] ($(v21)!\ratioDottedEdgeInv!(vT0)$);
	\draw[] (vT1) edge [stealth-,dotted] ($(v21)!\ratioDottedEdgeInv!(vT1)$);
	\draw[] (vTm) edge [stealth-,dotted] ($(v21)!\ratioDottedEdgeInv!(vTm)$);
	
	\draw[] (vT0) edge [stealth-, dotted] ($(v2m)!\ratioDottedEdgeInv!(vT0)$);
	\draw[] (vT1) edge [stealth-,dotted] ($(v2m)!\ratioDottedEdgeInv!(vT1)$);
	\draw[] (vTm) edge [stealth-,dotted] ($(v2m)!\ratioDottedEdgeInv!(vTm)$);
	
	\node[vertex] (vEnd) at (\xBegin + \xDist * 4, \yBegin) {$v_{T+1,0}$};
	
	\draw[edge] (vT0) edge node[edgetext] {$0$} (vEnd);
	\draw[edge] (vT1) edge node[edgetext] {$0$} (vEnd);
	\draw[edge] (vTm) edge node[edgetext] {$0$} (vEnd);
	
	\end{tikzpicture}
	\caption{Construction of the graph.}
	\label{fig:poly:graph:example}
\end{figure}

\subsection{Graph-based approach}
\label{sec:poly:graph}

Our algorithm works with an underlying directed, weighted graph $G=(V,E)$ that we describe first. 
Let $[k] \coloneqq \{1, 2, \dots, k\}$ and $[k]_0 \coloneqq \{0, 1, \dots, k\}$ with $k \in \mathbb{N}$. 
For each $t\in [T]$ and each $j \in [m]_0$, there is a vertex $v_{t,j}$, representing the state that
exactly $j$ servers are active at time $t$. Furthermore, there are two vertices $v_{0,0}$ and $v_{T+1,0}$ for
the initial and final states $x_0=0$ and $x_{T+1}=0$. For each $t\in \{2,\ldots,T\}$ and each pair
$j,j'\in [m]_0$, there is a directed edge from $v_{t-1,j}$ to $v_{t,j'}$ having weight $\beta (j'-j)^+ +f_t(j')$.
This edge weight corresponds to the switching cost when changing the number of servers between time $t-1$ and $t$ 
and to the operating cost incurred at time $t$.
Similarly, for $t=1$ and each $j'\in [m]_0$, there
is a directed edge from $v_{0,0}$ to $v_{1,j'}$ with weight $f_1(j')+\beta (j')^+$. Finally, for $t=T$ and each 
$j\in [m]_0$, there is a directed edge from $v_{T,j}$ to $v_{T+1,0}$ of weight~0. The structure of $G$ is depicted in Figure~\ref{fig:poly:graph:example}.

In the following, for each $j\in [m]_0$, vertex set $\mathcal{R}_j =\{ v_{t,j} \mid t \in [T]\}$ is called 
\emph{row}~$j$. For each $t\in [T]$, vertex set $\{ v_{t,j} \mid j \in [m]_0\}$ is called 
\emph{column}~$t$. 

A path between $v_{0,0}$ and $v_{T+1,0}$ represents a schedule. If the path visits $v_{t,j}$, then 
$x_t=j$ servers are active at time $t$. Note that a path visits exactly one vertex in each column, because the directed edges connect adjacent columns. The total length (weight) of a path
is equal to the cost of the corresponding schedule. 
An optimal schedule can be determined using a shortest
path computation, which takes $\mathcal{O}(Tm)$ time in the particular graph $G$.
However, this running time is not polynomial
because the encoding length of an input instance is linear in $T$ and $\log m$, in addition to the encoding
of the functions $f_t$. 

In the following, we present a polynomial time algorithm that improves an initial schedule iteratively using binary search.
In each iteration the algorithm constructs and uses only a constant number of rows of $G$. 

\subsection{Polynomial time algorithm}
\label{sec:poly:polytimealg}


An instance of the data-center optimization problem is defined by the tuple $\mathcal{P} = (T, m, \beta, F)$ with $F = (f_1, \dots, f_T)$. 
We assume that $m$ is a power of two. If this is not the case we can transform the given problem instance $\mathcal{P} = (T, m, \beta, F)$ to $\mathcal{P}' = (T, m', \beta, F')$ with $m' = 2^{\lceil \log m \rceil }$ and 
\begin{equation*}
f'_t(x) = \begin{cases}
f_t(x) & x \leq m \\ 
x \cdot (f_t(m) +\epsilon) & \text{otherwise}
\end{cases}
\end{equation*}
with $\epsilon > 0$. The term $x \cdot f_t(m)$ ensures that $f'_t(x)$ is a convex function, since the greatest slope of $f_t$ is $f_t(m) - f_t(m-1) \leq f_t(m)$. The inequality holds because $f_t(x) \geq 0$ for all $x \in [m]_0$. The additional term $x \cdot \epsilon$ ensures that it is adverse to use a state $x > m$, because the cost of $f_t(m)$ is always smaller.

Our algorithm uses $\log m - 1$ iterations denoted reversely by $k = K \coloneqq \log m - 2$ for the first iteration and $k = 0$ for the last iteration. 
The states used in iteration $k$ are always multiples of~$2^k$. 
For the first iteration we use the rows $\mathcal{R}_0, \mathcal{R}_{m/4}, \mathcal{R}_{m/2}, \mathcal{R}_{3m/4}, \mathcal{R}_m$, so that the graph of the first iteration contains the vertices
\begin{equation*}
V^K \coloneqq \{v_{0,0}, v_{T+1, 0}\} \cup \left\{v_{t, \xi \cdot m/4} \mid t \in [T] ,\xi \in \{0,1,2,3,4\}\right\}.
\end{equation*}
The optimal schedule for this simplified problem instance can be calculated in $\mathcal{O}(T)$ time, since each column contains only five states. Given an optimal schedule $\hat{X}^k = (\hat{x}^k_1, \dots, \hat{x}^k_T)$ of iteration $k$, let
\begin{equation*}
V^{k-1}_t \coloneqq \left\{\hat{x}^k_t + \xi \cdot 2^{k-1} \mid \xi \in \{-2, -1, 0, 1, 2\} \right\} \cap [m]_0
\end{equation*}
be the states used in the $t$-th column of the next iteration $k-1$. Thus, the iteration $k-1$ uses the vertex set
\begin{equation*}
V^{k-1} \coloneqq  \{v_{0,0}, v_{T+1, 0}\} \cup \left\{v_{t, j} \mid t \in [T], j \in V^{k-1}_t\right\} .
\end{equation*}

Note that the states with $\xi \in \{-2, 0, 2\}$ were already used in iteration $k$ and we just insert the intermediate states $\xi = -1$ and $\xi = 1$. 
If $\hat{x}^k_t = 0$ (or $\hat{x}^k_t = m$), then $\xi \in \{-2,-1\}$ (or $\xi \in \{1,2\}$) leads to negative states (or to states larger than $m$), thus the set $V^{k-1}_t$ is cut with $[m]_0$ to ensure that we only use valid states.

The last iteration ($k = 0$) provides an optimal schedule for the original problem instance as shown in the next section. The runtime of the algorithm is $\mathcal{O}(T \cdot \log m)$ and thus polynomial.

\subsection{Correctness}
\label{sec:poly:correct}

To prove the correctness of the algorithm described in the previous section we have to introduce some definitions:

Given the original problem instance $\mathcal{P} = (T, m, \beta, F)$, we define $\mathcal{P}_k$ (with $k \in [K]_0 \coloneqq [\log m - 2]_0$) as the data-center optimization problem where we are only allowed to use the states that are multiples of $2^k$. Let $M_{k} \coloneqq \{n \in [m]_0 \mid n \bmod 2^k = 0\}$, so $X$ is a feasible schedule for $\mathcal{P}_k$ if $x_t \in M_k$ holds for all $t \in [T]$. To express $\mathcal{P}_k$ as a tuple, we need another tuple element called $M$ which describes the allowed states, i.e., $x_t \in M$ for all $t \in [T]$. The original problem instance can be written as $\mathcal{P} = (T, m, \beta, F, [m]_0)$ and $\mathcal{P}_k = (T, m, \beta, F, M_k)$.  Note that $\mathcal{P}_0 = \mathcal{P}$. Let $\hat{X}^k = (\hat{x}^k_1, \dots, \hat{x}^k_T)$ denote an optimal schedule for $\mathcal{P}_k$. In general, for any given problem instance $Q = (T, m, \beta, F, M)$, let $\Phi_k(Q) \coloneqq (T, m, \beta, F, M \cap \{i \cdot 2^k \mid i \in \mathbb{N}\})$, so $\Phi_k(\mathcal{P}) = \mathcal{P}_k$.

Instead of using only states that are multiple of $2^k$ we can also scale a given problem instance $Q = (T,m,\beta, F, M)$ as follows. Let 
\begin{equation*}
\Psi_l(Q) \coloneqq (T, m/2^l, \beta \cdot 2^l, F', M')
\end{equation*}
with $M' \coloneqq \{x / 2^l \mid x \in M\}$, $F' = (f'_1, \dots, f'_T)$ and $f'_t(x) \coloneqq f_t(x \cdot 2^l)$. Given a schedule $X = (x_1, \dots, x_T)$ for $Q$ with cost $C^Q(X)$, the corresponding schedule $X' = (x_1 / 2^l, \dots, x_T / 2^l)$ for $\Psi_l(Q)$ has exactly the same cost, i.e., $C^Q(X) = C^{\Psi_l(Q)}(X')$. Note that the problem instance $\Psi_k(\mathcal{P}_k)$ uses all integral states less than or equal to $m / 2^k$, so there are no gaps.

Furthermore, we introduce a continuous version of any given problem instance $Q$ where fractional schedules are allowed.
Let $\bar{Q} = (T, m, \beta, \allowbreak \bar{F}, [0,m])$ with $\bar{F} = (\bar{f}_1, \dots, \bar{f}_T)$ be the continuous extension of the problem instance $Q = (T, m, \beta, F, M)$, where $x_t \in [0,m]$, $\bar{f}_t : [0,m] \rightarrow \mathbb{R}_{\geq 0}$ and 
\begin{equation}\label{eqn:poly:correct:extension}
\bar{f}_t(x) \coloneqq \begin{cases}
f_t(x) & \text{if }x \in M \\
(\lceil x \rceil - x) f_t(\lfloor x \rfloor) + (x - \lfloor x \rfloor) f_t(\lceil x \rceil) & \text{else} .
\end{cases}
\end{equation}
The operating cost of the fractional states is linearly interpolated, thus $\bar{f}_t$ is convex for all $t \in [T]$. 
Let $X^\ast = (x^\ast_1, \dots, x^\ast_T) \in [0,m]^T$ be an optimal schedule for $\bar{\mathcal{P}}$.

The set of all optimal schedules for a given problem instance $Q$ is denoted by $\Omega(Q)$.
Let $C^Q_{[a,b]}(X) \coloneqq \sum_{t=a}^{b} f_t(x_t) + \sum_{t=a+1}^{b} \beta (x_t - x_{t-1})^+ $ 
be the cost during the time interval $\{a, a+1, \dots, b\}$. We define $f_0(x) \coloneqq 0$, so $C^Q_{[0,T]}(X) = C^Q(X)$. 

Now, we are able to prove the correctness of our algorithm. We begin with a simple lemma showing the relationship between the functions $\Phi$ and $\Psi$.

\begin{lemma} \label{lemma:poly:correct:phipsi}
	The problem instances $\Phi_{k-l}(\Psi_l(\mathcal{P}_l))$ and $\Psi_l(\mathcal{P}_k)$ are equivalent.
\end{lemma}

\begin{proof}
	We begin with $\Phi_{k-l}(\Psi_l(\mathcal{P}_l))$ and simply apply the definitions of $\mathcal{P}_l$, $\Psi_l$ and $\Phi_{k-l}$. 
	\begin{align*}
		& \Phi_{k-l}(\Psi_l(\mathcal{P}_l)) \\
		={}& \Phi_{k-l}(\Psi_l(\Phi_l((T, m, \beta, F, \{n \in [m]_0\})))) \\
		={}& \Phi_{k-l}(\Psi_l((T, m, \beta, F, \{n \in [m]_0 \mid n \bmod  2^l = 0\}))) \\
		={}& \Phi_{k-l}((T, m / 2^l, \beta \cdot 2^l, F_l, \{n \in [m / 2^l]_0 \mid n \bmod 1 = 0\}) )\\
		={}& (T, m / 2^l, \beta \cdot 2^l, F_l, \{n \in [m / 2^l]_0 \mid n \bmod 2^{k-l} = 0\}). \\
		\intertext{Afterwards, we use the definitions of $\Psi_l$, $\Phi_k$ and $\mathcal{P}_k$ and get $\Psi_l(\mathcal{P}_k)$ as shown below:}
		& (T, m / 2^l, \beta \cdot 2^l, F_l, \{n \in [m / 2^l]_0 \mid n \bmod 2^{k-l} = 0\}) \\
		={}& \Psi_l((T, m, \beta, F, \{n \in [m]_0 \mid n \bmod 2^{k} = 0\}))) \\
		={}& \Psi_l(\Phi_{k}((T, m, \beta, F, \{n \in [m]_0\}))) \\
		={}& \Psi_l(\mathcal{P}_k). \qedhere
	\end{align*}
\end{proof}

The next technical lemma will be needed later. Informally, it demonstrates that optimal solutions for the reduced discrete problem instance and the 
continuous problem instance behave similarly. 

\begin{lemma} \label{lemma:poly:correct:direction}
	Let $Y \in \Omega(\mathcal{P}_k)$ be an optimal schedule for $\mathcal{P}_k$ with $k \in [K]_0$. There exists an optimal solution $X^\ast \in \Omega(\bar{\mathcal{P}})$ such that 
	\begin{equation} \label{eqn:poly:correct:border:slopeproduct}
		(y_t - y_{t-1}) \cdot (x^\ast_t - x^\ast_{t-1}) \geq 0
	\end{equation}
	holds for all $t \in [T]$ with $|y_t - x^\ast_t| \geq 2^k$ or $|y_{t-1} - x^\ast_{t-1}| \geq 2^k$.
\end{lemma}

\begin{proof}
	Let $x^{\text{min}+}_t \coloneqq \max ( \argmin_x f_t(x))$ be the greatest state that minimizes $f_t$ and let $x^{\text{min}-}_t \coloneqq \min (\argmin_x f_t(x))$ be the smallest state that minimizes $f_t$. Let $X^\ast \in \Omega(\bar{\mathcal{P}})$ be an arbitrary optimal solution. We will show that it is possible to modify $X^\ast$ such that it fulfills equation~\eqref{eqn:poly:correct:border:slopeproduct} without increasing the cost. The modified schedule is denoted by $\tilde{X}^\ast$. We differ between several cases according to the relations of $y_{t-1}, y_t, x_{t-1}$ and $x_t$:	
	\begin{myEnumerate}
		\item $x^\ast_{t-1} > x^\ast_t $
		\begin{myEnumerate}
			\item $y_{t-1} \geq y_t$ \\
			Equation~\eqref{eqn:poly:correct:border:slopeproduct} is fulfilled.
			\item $y_{t-1} < y_t$
			\begin{myEnumerate}
				\item $y_{t-1} \leq x^\ast_{t-1}$ \\
				If $x^{\text{min}+}_{t-1} < x^\ast_{t-1}$, then using 
				$\tilde{x}^\ast_{t-1} \coloneqq x^\ast_{t-1} - \epsilon$ (for a small $\epsilon > 0$)
				instead of $x^\ast_{t-1}$ would lead to a better solution, because $f_{t-1}$ is a convex function and the switching costs between the time slots $t-2$ and $t$ are not increased, so 
				\begin{equation} \label{eqn:poly:correct:direction:a}
				x^{\text{min}+}_{t-1} \geq x^\ast_{t-1}
				\end{equation}
				must be fulfilled. If $x^{\text{min}-}_{t} > x^\ast_t$, then 
				$\tilde{x}^\ast_t \coloneqq x^\ast_{t} + \epsilon$ 
				would lead to a better solution for the same reason, so 
				\begin{equation} \label{eqn:poly:correct:direction:b}
				x^{\text{min}-}_{t} \leq x^\ast_t .
				\end{equation}
				\begin{myEnumerate}
					\item $y_t \leq x^\ast_{t-1}$ \\
					If $x^{\text{min}-}_{t-1} > y_{t-1}$, then using 
					$
					\tilde{y}_{t-1} \coloneqq y_{t} 
					\alignstacktinytext{\text{(i)}}{\leq} x^\ast_{t-1} 
					\alignstacktinytext{\eqref{eqn:poly:correct:direction:a}}{\leq} x^{\text{min}+}_{t-1}
					$
					instead of $y_{t-1}$ would lead to a better solution, so 
					\begin{equation} \label{eqn:poly:correct:direction:c}
					x^{\text{min}-}_{t-1} \leq y_{t-1}
					\end{equation}
					must be fulfilled. 
					\begin{myEnumerate}
						\item $x^\ast_t \geq y_{t-1}$ \\ 
							We set $\tilde{x}^\ast_{t-1} \coloneqq x^\ast_t$, so equation~\eqref{eqn:poly:correct:border:slopeproduct} is fulfilled. Since
							$
							x^{\text{min}-}_{t-1} 
							\alignstacktinytext{\eqref{eqn:poly:correct:direction:c}}{\leq} y_{t-1} 
							\alignstacktinytext{\text{C1}}{\leq} x^\ast_t 
							\alignstacktinytext{\text{I}}{<} x^\ast_{t-1} 
							\alignstacktinytext{\eqref{eqn:poly:correct:direction:a}}{\leq} x^{\text{min}+}_{t-1} 
							$
							, the cost of $\tilde{X}^\ast$ is not increased. 
						\item $x^\ast_t < y_{t-1}$ \\
							We set $\tilde{x}^\ast_{t-1} \coloneqq y_{t-1}$ which does not increase the cost of $\tilde{X}^\ast$ because 
							$x^\ast_{t-1} 
							\alignstacktinytext{\text{(a)}}{\geq} y_{t-1} 
							\alignstacktinytext{\text{C2}}{>} x^\ast_t$ 
							and 
							$x^{\text{min}-}_{t-1} 
							\alignstacktinytext{\eqref{eqn:poly:correct:direction:c}}{\leq} y_{t-1} 
							\alignstacktinytext{\text{(a)}}{\leq} x^\ast_{t-1} 
							\alignstacktinytext{\eqref{eqn:poly:correct:direction:a}}{\leq} x^{\text{min}+}_{t-1}$.
							If $x^{\text{min}+}_t < y_t$, then $\tilde{y}_t \coloneqq y_{t-1}$ would lead to a better solution, so $x^{\text{min}+}_t \geq y_t$. 
							We set $\tilde{x}^\ast_t \coloneqq \tilde{x}^\ast_{t-1}$, so equation~\eqref{eqn:poly:correct:border:slopeproduct} is fulfilled. Since 
							$x^{\text{min}-}_{t} 
							\alignstacktinytext{\eqref{eqn:poly:correct:direction:b}}{\leq} x^\ast_t 
							\alignstacktinytext{\text{C2}}{<} \tilde{x}^\ast_{t-1} 
							\alignstacktinytext{\text{(B)}}{<} y_t 
							\leq x^{\text{min}+}_{t}$,
							the cost of $\tilde{X}^\ast$ is not increased. 
					\end{myEnumerate}
					\item $y_t > x^\ast_{t-1}$ 
					\begin{myEnumerate}
						\item $x^\ast_t \leq y_{t-1}$ \\
							We have 
							$y_t 
							\alignstacktinytext{\text{(B)}}{>} y_{t-1} 
							\alignstacktinytext{\text{C1}}{\geq} x^\ast_t 
							\alignstacktinytext{\eqref{eqn:poly:correct:direction:b}}{\geq} x^{\text{min}-}_t$.
							If $x^{\text{min}+}_t < y_t$, then $\tilde{y}_t \coloneqq y_{t-1}$ would lead to a better solution, so $x^{\text{min}+}_t \geq y_t$. 
							We set $\tilde{x}^\ast_t \coloneqq x^\ast_{t-1}$, so equation~\eqref{eqn:poly:correct:border:slopeproduct} is fulfilled. Since 
							$x^{\text{min}-}_{t} 
							\alignstacktinytext{\eqref{eqn:poly:correct:direction:b}}{\leq} x^\ast_t 
							\alignstacktinytext{\text{I}}{<} x^\ast_{t-1} 
							\alignstacktinytext{\text{(ii)}}{<} y_t 
							\leq x^{\text{min}+}_{t}$,
							the cost of $\tilde{X}^\ast$ is not increased.
						\item $x^\ast_t > y_{t-1}$ and $|y_{t-1} - x^\ast_{t-1}| \geq 2^k$ \\
							By $y_{t-1} \alignstacktinytext{\text{(a)}}{\leq} x^\ast_{t-1}$ and $|y_{t-1} - x^\ast_{t-1}| \geq 2^k$, there exists a state $\tilde{y}_{t-1} \in M_k$ with $y_{t-1} < \tilde{y}_{t-1} \leq x^\ast_{t-1}$. If $x^{\text{min}-}_{t-1} > y_{t-1}$, then using $\tilde{y}_{t-1}$ instead of $y_{t-1}$ would lead to a better solution, so $x^{\text{min}-}_{t-1} \leq y_{t-1}$ must be fulfilled.  
							We set $\tilde{x}^\ast_{t-1} \coloneqq x^\ast_t$, so equation~\eqref{eqn:poly:correct:border:slopeproduct} is fulfilled. Since  
							$x^{\text{min}-}_{t-1} 
							\leq y_{t-1} 
							\alignstacktinytext{\text{C2}}{<} x^\ast_t 
							\alignstacktinytext{\text{I}}{<} x^\ast_{t-1} 
							\alignstacktinytext{\eqref{eqn:poly:correct:direction:a}}{\leq} x^{\text{min}+}_{t-1}$,
							the cost of $\tilde{X}^\ast$ is not increased. 
						\item $x^\ast_t > y_{t-1}$ and $|y_{t} - x^\ast_{t}| \geq 2^k$ \\
							Since 
							$x^\ast_t 
							\alignstacktinytext{\text{I}}{<} x^\ast_{t-1} 
							\alignstacktinytext{\text{(ii)}}{<} y_t$
							and $|y_{t} - x^\ast_{t}| \geq 2^k$,
							there exists a state $\tilde{y}_{t} \in M_k$ with $x^\ast_{t} \leq \tilde{y}_{t} < y_t$. 
							If $x^{\text{min}+}_t < y_t$, then using $\tilde{y}_t$ instead of $y_t$ would lead to a better solution, so $x^{\text{min}+}_t \geq y_t$. 
							We set $\tilde{x}^\ast_t \coloneqq x^\ast_{t-1}$, so equation~\eqref{eqn:poly:correct:border:slopeproduct} is fulfilled. Since 
							$x^{\text{min}-}_{t} 
							\alignstacktinytext{\eqref{eqn:poly:correct:direction:b}}{\leq} x^\ast_t 
							\alignstacktinytext{\text{I}}{<} x^\ast_{t-1} 
							\alignstacktinytext{\text{(ii)}}{<}< y_t 
							\leq x^{\text{min}+}_{t}$,
							the cost of $\tilde{X}^\ast$ is not increased. 
					\end{myEnumerate}
				\end{myEnumerate}
				\item $y_{t-1} > x^\ast_{t-1}$ \\
				If $x^{\text{min}-}_{t} > x^\ast_t$, then $\tilde{x}^\ast_t \coloneqq x^\ast_t + \epsilon$ would lead to a better solution, so $x^{\text{min}-}_{t} \leq x^\ast_t$. 
				If $x^{\text{min}+}_t < y_t$, then $\tilde{y}_t \coloneqq y_{t-1}$ would lead to a better solution because 
				$x^{\text{min}-}_{t} 
				\leq x^\ast_t 
				\alignstacktinytext{\text{I}}{<} x^\ast_{t-1} 
				\alignstacktinytext{\text{(b)}}{<} y_{t-1} 
				\alignstacktinytext{\text{(B)}}{<} y_t$,
				so $x^{\text{min}+}_t \geq y_t$. 
				We set $\tilde{x}^\ast_t \coloneqq x^\ast_{t-1}$, so equation~\eqref{eqn:poly:correct:border:slopeproduct} is fulfilled. Since 
				$x^{\text{min}-}_{t} 
				\leq x^\ast_t 
				\alignstacktinytext{\text{I}}{<} x^\ast_{t-1} 
				\alignstacktinytext{\text{(b)}}{<} y_{t-1} 
				\alignstacktinytext{\text{(B)}}{<} y_t 
				\leq x^{\text{min}+}_t$,
				the cost of $\tilde{X}^\ast$ is not increased.
			\end{myEnumerate}
		\end{myEnumerate}
		\item $x^\ast_{t-1} = x^\ast_t$ \\
		Equation~\eqref{eqn:poly:correct:border:slopeproduct} is fulfilled.
		\item $x^\ast_{t-1} < x^\ast_t$  \\
		This case is symmetric to case 1. \qedhere
	\end{myEnumerate}
\end{proof}

By using Lemma~\ref{lemma:poly:correct:direction}, we can show that an optimal solution for a discrete problem instance $\mathcal{P}_k$ 
cannot be very far from an optimal solution of the continuous problem instance $\bar{\mathcal{P}}$.

\begin{lemma}\label{lemma:poly:correct:border}
	Let $\hat{X}^k \in \Omega(\mathcal{P}_k)$ be an arbitrary optimal schedule for $\mathcal{P}_k$ with $k \in [K]_0$. There exists an optimal schedule $X^\ast \in \Omega(\bar{\mathcal{P}})$ for $\bar{\mathcal{P}}$ such that $|\hat{x}^k_t - x^\ast_t| < 2^k$ holds for all $t \in [T]$. Formally,
	\begin{equation*}
	\forall k \in [K]_0 : \forall \hat{X}^k \in \Omega(\mathcal{P}_k) : \exists X^\ast \in \Omega(\bar{\mathcal{P}}) : \forall t \in [T] : |\hat{x}^k_t - x^\ast_t| < 2^k .
	\end{equation*}
\end{lemma}

\begin{proof}
	To get a contradiction, we assume 
	that there exists a $\hat{X}^k \in \Omega(\mathcal{P}_k)$ with $k \in [K]_0$ such that for all optimal schedules $X^\ast \in \Omega(\bar{\mathcal{P}})$ there is at least one $t \in [T]$ with $|\hat{x}^k_t - x^\ast_t| \geq 2^k$.
	
	Let $X^\ast \in \Omega(\bar{\mathcal{P}})$ be an arbitrary optimal schedule that fulfills Lemma~\ref{lemma:poly:correct:direction}, i.e., $(\hat{x}^k_t - \hat{x}^k_{t-1}) \cdot (x^\ast_t - x^\ast_{t-1}) \geq 0$ holds for all $t \in [T]$ with $|\hat{x}^k_t - x^\ast_t| \geq 2^k$ or $|\hat{x}^k_{t-1} - x^\ast_{t-1}| \geq 2^k$.
	
	Given an arbitrary schedule $X = (x_1, \dots, x_T)$, let $J_1, \dots, J_l \subseteq [T]$ be the inclusion maximal time intervals such that $|x_t - x^\ast_t| \geq 2^k$ or $x_t \notin M_k$ holds for all $t \in J_j$ and the sign of $x_t - x^\ast_t$ remains the same during~$J_j$. The set of all $J_j$ with $j \in [l]$ is denoted by $\mathcal{J}(X)$. If $\mathcal{J}(X)$ is empty, then the conditions $|x_t - x^\ast_t| < 2^k$ and $x_t \in M_k$ are fulfilled for all $t \in [T]$. 
	%
	The set of all time slots in $\mathcal{J}(X)$ is denoted by $\mathcal{T}(X) \coloneqq \{t \in J \mid J \in \mathcal{J}(X)\}$ and the number of time slots in $\mathcal{J}$ by $L(X) \coloneqq |\mathcal{T}(X)| =  \sum_{J \in \mathcal{J}} |J|$. 
	
	We will use a recursive transformation $\phi$ that reduces $L(X)$ at least by one for each step, while the cost of $X$ is not increased. Formally, we have to show that $L(\phi(X)) \leq L(X) - 1$ and $C^{\bar{\mathcal{P}}}(\phi(X)) \leq C^{\bar{\mathcal{P}}}(X)$ holds. The first inequality ensures that the recursive procedure will terminate. 
	The transformation described below will produce fractional schedules, however for each $t \in [T] \setminus \mathcal{T}(X)$ it is ensured that $x_t \in M_{k}$. Therefore, if $L(X) = 0$, the corresponding schedule fulfills $|x_t - x^\ast_t| < 2^k$ and $x_t \in M_{k}$ for all $t \in [T]$. 
	
	To describe the transformation, we will use the following notation: A given schedule $Y = (y_1, \dots, y_T)$ with $L(Y) > 0$ is transformed to $Z = \phi(Y) = (z_1, \dots, z_T)$. We assume that $Y$ fulfills the invariant
	\begin{equation} \label{eqn:poly:correct:border:proof:invariant}
	(y_t - y_{t-1}) \cdot (x^\ast_t - x^\ast_{t-1}) \geq 0
	\end{equation} 
	for all $t$ with $\{t-1, t\} \cap \mathcal{J}(Y) \not= \emptyset$, i.e., $t-1$ or $t$ (or both) belong to $\mathcal{J}(Y) $. 
	For $Y = \hat{X}^k$, this is the case since we chose $X^\ast$ such that the property of Lemma~\ref{lemma:poly:correct:direction} is satisfied. We will show that inequality~\eqref{eqn:poly:correct:border:proof:invariant} still holds for the transformed schedule $Z$. 
	
	%
	Let $J \coloneqq \{t_i + 1, \dots, t_{i+1} -1 \} \in \mathcal{J}(Y)$.
	We differ between two cases, in case~1 we handle the intervals with $y_t > x^\ast_t$ (for all $t \in J$) and in case 2 we handle the intervals with $y_t < x^\ast_t$. We will handle case 1 first.
	
	Let $\lceil x \rceil_{n} \coloneqq n \cdot \lceil x/n \rceil$ with $x \in \mathbb{R}$ and $n \in \mathbb{N}$ be the smallest value that is divisible by $n$ and greater than or equal to $x$.
	The schedule $Y$ is transformed to $Z$ with 
	\begin{equation} \label{eqn:poly:correct:border:proof:z}
	z_t \coloneqq \begin{cases}
	y_t & \text{if $t \notin J$} \\
	\lambda \cdot y_t + (1-\lambda ) \cdot x^\ast_t & \text{if $t \in J$}
	\end{cases}
	\end{equation}
	where $\lambda \in [0,1]$ is as small as possible such that $z_t \geq \lceil x^\ast_t \rceil_{2^k}$ holds for all $t \in J$, so at least one time slot $t_= \in J$ satisfies this condition with equality. 
	This transformation ensures that $L(Z) \leq L(Y) - 1$ holds, because there is at least one time slot ($t_=$) in $J$ that fulfills $|z_{t_=}- x^\ast_{t_=}| < 2^k$. 
	
	We still have to show that the total cost is not increased by this operation. The total cost can be written as 
	\begin{equation} \label{eqn:poly:correct:border:cost}
	C^{\bar{\mathcal{P}}}(X) ={}  C^{\bar{\mathcal{P}}}_{[0,t_i]}(X) + \beta (x_{t_i + 1} - x_{t_i})^+ + C^{\bar{\mathcal{P}}}_{[t_i + 1, t_{i+1} - 1]}(X) 
	+ \beta(x_{t_{i+1}} -x_{t_{i+1} - 1} )^+ + C^{\bar{\mathcal{P}}}_{[t_{i+1}, T]}(X) .
	\end{equation}
	We have $C^{\bar{\mathcal{P}}}_{[0,t_i]}(Y) = C^{\bar{\mathcal{P}}}_{[0,t_i]}(Z)$ and $C^{\bar{\mathcal{P}}}_{[t_{i+1}, T]}(Y) = C^{\bar{\mathcal{P}}}_{[t_{i+1}, T]}(Z)$. 
	
	Consider the time slot $t_i$. 
	The invariant (inequality~\eqref{eqn:poly:correct:border:proof:invariant}) says
	that the terms $(y_{t_i + 1}  - y_{t_i})$ and $(x^\ast_{t_i + 1}  - x^\ast_{t_i})$ are both either non-negative or non-positive, so
	in equation~\eqref{eqn:poly:correct:border:cost} the term
	 $\beta (x_{t_i + 1} - x_{t_i})^+$ can be replaced by $\beta (x_{t_i + 1} - x_{t_i})$ or zero, respectively. 
	Analogously, for the time slot $t_{i+1}$, 
	the term $\beta(x_{t_{i+1}} -x_{t_{i+1} - 1} )^+$ in equation~\eqref{eqn:poly:correct:border:cost} can be replaced by $\beta(x_{t_{i+1}} -x_{t_{i+1} - 1} )$ or zero. In the former cases, the cost function is
	\begin{equation*}
	C^{\bar{\mathcal{P}}}(X) ={}  C^{\bar{\mathcal{P}}}_{[0,t_i]}(X) + \beta x_{t_i + 1} - \beta x_{t_i} + C^{\bar{\mathcal{P}}}_{[t_i + 1, t_{i+1} - 1]}(X) \\
	+ \beta x_{t_{i+1}} - \beta x_{t_{i+1} - 1} + C^{\bar{\mathcal{P}}}_{[t_{i+1}, T]}(X) .
	\end{equation*}
	Given a schedule $X = (x_1, \dots, x_T)$, we define $X_{[a:b]} \coloneqq (x_a, \dots, x_b)$ and $X_J = X_{[t_i + 1 : t_{i+1} - 1]}$.  
	Since there is no summand that contains both $x_{t_i}$ and $x_{t_i + 1}$, the function
	\begin{equation*}
	\begin{aligned}
	D_{X^\ast}((x'_{t_i + 1}, \dots, x'_{t_{i+1} - 1}))  \coloneqq{} &  C^{\bar{\mathcal{P}}}_{[0,t_i]}(X^\ast) - \beta x^\ast_{t_i} \\
	& + \beta x'_{t_i + 1} + C^{\bar{\mathcal{P}}}_{[t_i + 1, t_{i+1} - 1]}(X') + \beta x'_{t_{i+1}} \\
	& - \beta x^\ast_{t_{i+1} - 1} + C^{\bar{\mathcal{P}}}_{[t_{i+1}, T]}(X^\ast)
	\end{aligned}
	\end{equation*}
	with $x'_{t_i + 1} \geq x^\ast_{t_i + 1}$ and $x'_{t_{i+1} - 1} \geq x^\ast_{t_{i+1} - 1}$ is convex and has a minimum at $X_J^\text{min} \coloneqq (x^\ast_{t_i + 1}, \dots, x^\ast_{t_{i+1} - 1})$.	
	
	Due to convexity, $D_{X^\ast}(Y_J) \geq D_{X^\ast}(Z_J) \geq D_{X^\ast}(X^\text{min}_J)$, because $Z_J = \lambda Y_J + (1-\lambda) X^\text{min}_J$. Therefore, $C^{\bar{\mathcal{P}}}(Z) \leq C^{\bar{\mathcal{P}}}(Y)$ holds. If $\beta (x_{t_i + 1} - x_{t_i})^+ = 0$ or $\beta(x_{t_{i+1}} -x_{t_{i+1} - 1} )^+ = 0$ we can use the same argument.

	We still have to handle the second case, i.e., $y_t < x^\ast_t$. The proof is almost analogous, the difference is that we choose $\lambda$ as small as possible such that $z_t \leq \lfloor x^\ast_t \rfloor_{2^k}$ (where $\lfloor x \rfloor_n \coloneqq n \cdot \lfloor x/n \rfloor$). Then we have a time slot $t_=$ with $z_{t_=} = \lfloor x^\ast_{t_=} \rfloor_{2^k}$, so $L(Z) \leq L(Y) - 1$. 
	The proof that shows $C^{\bar{\mathcal{P}}}(Z) \leq C^{\bar{\mathcal{P}}}(Y)$ holds for both cases.
	
	To apply the transformation several times, we have to show that the invariant still holds for $Z$. We differ between three cases: (1) $\{t-1, t\} \subseteq J$, i.e., both $z_{t-1}$ and $z_t$ are transformed, (2) $t - 1 = t_i$, so only $z_t$ is transformed, and (3) $t = t_{i+1}$, so only $z_{t-1}$ is transformed. For the first case, we get
	\begin{align*} 
		(z_t - z_{t-1}) (x^\ast_t - x^\ast_{t-1})
		&\alignstack{\eqref{eqn:poly:correct:border:proof:z}}{=} \big(\lambda y_t + (1-\lambda) x^\ast_t - \lambda y_{t-1} - (1-\lambda) x^\ast_{t-1} \big) \cdot (x^\ast_t - x^\ast_{t-1}) \\
		&= \lambda (y_t - y_{t-1})(x^\ast_t - x^\ast_{t-1}) + (1-\lambda) (x^\ast_t - x^\ast_{t-1})^2 \\
		&\alignstack{\eqref{eqn:poly:correct:border:proof:invariant}}{\geq} 0 .
	\end{align*}
	The first equation uses the definition of $Z$. The last inequality holds since $Y$ fulfills the invariant.
	
	In case 2, we have $t-1 \notin J$, so $y_{t-1} \in M_k$. 
	If $y_t \geq y_{t-1}$ and $x^\ast_t \geq x^\ast_{t-1}$, then $x^\ast_{t-1} > y_{t-1} - 2^k$, because otherwise $t-1$ would belong to $J$. 
	Therefore, we have 
	\begin{equation*}
	z_t \geq \lceil x^\ast_t \rceil_{2^k} \geq \lceil x^\ast_{t-1} \rceil_{2^k} \geq y_{t-1} .
	\end{equation*}
	By using this inequality as well as $z_{t-1} = y_{t-1}$, we get $z_t - z_{t-1} \geq 0$, so $Z$ fulfills the invariant. The case $y_t \leq y_{t-1}$ and $x^\ast_t \leq x^\ast_{t-1}$ is analogous. Since the invariant holds for $Y$, there are no other cases.
	
	Case 3 (i.e., $t \notin J$) is analogous to case 2. Therefore, $Z$ always fulfills the invariant.
	
	We use the transformation $\phi$ until $L(Z) = 0$. Then, $\mathcal{J}(Z)$ is empty, so all states of $Z$ are multiples of $2^k$, i.e., $z_t \in M_k$ for all $t \in [T]$. Since $\hat{X}^k$ was defined to be optimal, $C^{\bar{\mathcal{P}}}(\hat{X}^k) = C^{\bar{\mathcal{P}}}(Z)$ holds. By our assumption, $Z \not = \hat{X}^k$ holds (because otherwise $|\hat{x}^k_t - x^\ast_t| < 2^k$ would be fulfilled for all $t \in [T]$), so there was a transformation with $\lambda < 1$. Thus, we moved towards the optimal schedule, however by $C^{\bar{\mathcal{P}}}(\hat{X}^k) = C^{\bar{\mathcal{P}}}(Z)$, the cost does not change. As $D_{X^\ast}(X')$ is a convex function, $C^{\bar{\mathcal{P}}}(\hat{X}^k) = C^{\bar{\mathcal{P}}}(Z)$ implies that $C^{\bar{\mathcal{P}}}(Z) = C^{\bar{\mathcal{P}}}(X^\ast)$, because $X^\ast$ minimizes $C^{\bar{\mathcal{P}}}$. In this case $\hat{X}^k$ is also optimal for $\bar{\mathcal{P}}$, so by choosing $X^\ast \coloneqq \hat{X}^k \in \Omega(\bar{\mathcal{P}})$, we have an optimal schedule for $\bar{\mathcal{P}}$ that fulfills the condition $|\hat{x}^k_t - x^\ast_t| < 2^k$ for all $t \in [T]$. Therefore, our assumption was wrong and the lemma is proven.
\end{proof}

The next lemma shows how an optimal fractional schedule can be rounded to an integral schedule such that it is still optimal.

\begin{lemma} \label{lemma:poly:correct:rounding}
	Let $X^\ast \in \Omega(\bar{\mathcal{P}})$. The schedules $\lfloor X^\ast \rfloor \coloneqq (\lfloor x^\ast_1 \rfloor, \dots, \lfloor x^\ast_T \rfloor)$ and $\lceil X^\ast \rceil \coloneqq (\lceil x^\ast_1 \rceil, \dots, \lceil x^\ast_T \rceil)$ are optimal too, i.e., $\lfloor X^\ast \rfloor , \lceil X^\ast \rceil \in \Omega(\bar{\mathcal{P}})$.
\end{lemma}

\begin{proof}
	Let $X^\ast \in \Omega(\bar{\mathcal{P}})$ be arbitrary. Let $\mathcal{I}(X^\ast) = \{I_1, \dots, I_l\}$ be the set of time intervals such that for each $I_i \coloneqq \{a_i, a_i + 1, \dots, b_i\}$ with $i \in [l]$ the following conditions are fulfilled.
	\begin{enumerate}
		\item All states of $X^\ast$ have the same value during $I_i$, i.e., $x^\ast_t = v_i$ for all $t \in I_i$.
		\item The value is fractional, i.e., $v_i \notin \mathbb{N}$.
		\item Each $I_i$ is inclusion maximal, i.e., $x^\ast_{a_i - 1} \not= v_i$ and $x^\ast_{b_i + 1} \not = v_i$.
		\item The intervals are sorted, i.e., $b_i < a_{i+1}$ for all $i \in [l-1]$. 
	\end{enumerate}
	If $\mathcal{I}(X^\ast) = \emptyset$, then $X^\ast$ is an integral schedule, so $\lfloor X^\ast \rfloor = X^\ast = \lceil X^\ast \rceil $. Otherwise let $I_i \in \mathcal{I}(X^\ast)$ be an arbitrary interval. We will transform $X^\ast$ to $X'$ by changing the states at $I_i$ such that $|\mathcal{I}(X')| < |\mathcal{I}(X^\ast)|$ and $\lfloor x^\ast_t \rfloor \leq x'_t \leq \lceil x^\ast_t \rceil$ for all $t \in I_i$.
	Let $g(x) \coloneqq \sum_{t=a_i}^{b_i} \bar{f}_t(x)$. Since each $\bar{f}_t(x)$ is linear for $x \in [\lfloor v_i \rfloor, \lceil v_i \rceil ]$, the slope of $g(x)$ is constant for $x \in [\lfloor v_i \rfloor, \lceil v_i \rceil ]$ and denoted by $g'(v_i)$. 
	According to $I_i$, we have to differ between different cases:

	
	\begin{enumerate}
		\item $x^\ast_{a_i-1} < v_i < x^\ast_{b_i+1}$
		
		Let $\tilde{x}^\ast_{a_i-1} \coloneqq \max \{x^\ast_{a_i-1}, \lfloor v_i \rfloor \}$ and $\tilde{x}^\ast_{b_i+1} \coloneqq \min \{x^\ast_{b_i+1}, \lceil v_i \rceil \}$. By using any schedule with $x'_{a_i} = x'_{a_i + 1} = \dots = x'_{b_i} \in [\tilde{x}^\ast_{a_i-1}, \tilde{x}^\ast_{b_i+1}]$ (and $x'_t = x^\ast_t$ otherwise), the switching cost is unchanged. Since $I_i$ is inclusion maximal and $X^\ast$ is optimal, we can conclude that $g'(v_i) = 0$, so $C(X') = C(X^\ast)$. To show that $\lfloor X^\ast \rfloor$ is optimal, we set $x'_t = \tilde{x}^\ast_{a_i-1}$ for all $t \in I_i$. To show that $\lceil X^\ast \rceil $ is optimal, we set $x'_t = \tilde{x}^\ast_{b_i+1}$ for all $t \in I_i$.
		
		\item $x^\ast_{a_i-1} > v_i > x^\ast_{b_i+1}$
		
		This case is analogous to the first case
		
		\item $x^\ast_{a_i-1} > v_i < x^\ast_{b_i+1}$
		
		Let $v^+ = \min \{x^\ast_{a_i-1}, x^\ast_{b_i+1}, \lceil v_i \rceil \}$. Let $v'_i \in [\lfloor v_i \rfloor, v^+ ]$. By using the schedule $x'_t = v'_i$ for all $t \in I_i$, the switching cost is increased by $\beta (v_i - v'_i)$, but the operating cost is reduced by $g'(v_i) \cdot (v_i -  v'_i )$. As $X^\ast$ is optimal, we can conclude that $g'(v_i) = \beta$ because otherwise either using $v'_i = \lfloor v_i \rfloor$ or $v'_i = v^+$ (for all $t \in I_i$) would lead to a better solution, since $\lfloor v_i \rfloor < v_i < v^+$. Therefore, the total cost of $X'$ does not change  for $v'_i \in [\lfloor v_i \rfloor, v^+ ]$. To show that $\lfloor X^\ast \rfloor$ is optimal, we set $x'_t = \lfloor v_i \rfloor$ for all $t \in I_i$. To show that $\lceil X^\ast \rceil $ is optimal, we set $x'_t = v^+ $ for all $t \in I_i$.
		
		\item $x^\ast_{a_i-1} < v_i > x^\ast_{b_i+1}$
		
		This case is analogous to the third case, but $\lfloor x \rfloor$ and $\lceil x \rceil$ are swapped as well as $\min$ and $\max$. Furthermore, $g'(v_i) = -\beta$ and we replace $(v_i - v'_i)$ with $(v'_i - v_i)$. 
	\end{enumerate}
	By using the transformation described above, the number $|\mathcal{I}|$ of fractional intervals is at least reduced by 1. By applying the transformation several times until $|\mathcal{I}| = 0$, we receive $\lfloor X^\ast \rfloor$ or $\lceil X^\ast \rceil$. The total cost is not increased by the operations.
\end{proof}

So far, we have shown in Lemma~\ref{lemma:poly:correct:border} that for each optimal solution of the discrete problem instance $\mathcal{P}_k$ there is an optimal solution of the continuous problem instance $\bar{\mathcal{P}}$ that is not far away. In the following lemma we expand this statement: Given an optimal solution for $\mathcal{P}_k$, there is not only a fractional solution for $\bar{\mathcal{P}}$ that is not far away, but also an optimal solution of the discrete problem instance $\mathcal{P}_l$ for the subsequent iterations $l < k$.

\begin{lemma} \label{lemma:poly:correct:iteration}
	Let $k > l$ with $k,l \in [K]_0$. Let $\hat{X}^k \in \Omega(\mathcal{P}_k)$ be an arbitrary optimal schedule for $\mathcal{P}_k$ with $k \in [K]_0$. There exists an optimal schedule $\hat{X}^l \in \Omega(\mathcal{P}^l)$ for $\mathcal{P}^l$ such that $|\hat{x}^k_t - \hat{x}^l_t| \leq 2^k$ for all $t \in[T]$. 
	Formally, $\forall k \in [K]_0 : \forall l \in [k-1] : \forall \hat{X}^k \in \Omega(\mathcal{P}_k) : \exists \hat{X}^l \in \Omega(\mathcal{P}_l) : \forall t \in [T] : |\hat{x}^k_t - \hat{x}^l_t| \leq 2^k $.
\end{lemma}
%
\begin{proof}
	Consider the reduced problem instance $\mathcal{Q} \coloneqq \Psi_l(\mathcal{P}_l)$ as well as the instance $\mathcal{Q}_{k-l} \coloneqq \Phi_{k-l}(\mathcal{Q})$ which is equivalent to $\Psi_l(\mathcal{P}_k)$ due to Lemma~\ref{lemma:poly:correct:phipsi}. Let $\hat{X}_\mathcal{Q}^{k-l} = (\hat{x}^k_1 / 2^l, \dots, \hat{x}^k_T / 2^l)$ be an optimal schedule for $\mathcal{Q}_{k-l}$. We apply Lemma~\ref{lemma:poly:correct:border}, but we use $\hat{X}_\mathcal{Q}^{k-l}$ and $\mathcal{Q}$ instead of $\hat{X}^k$ and~$\mathcal{P}$. By Lemma~\ref{lemma:poly:correct:border}, there exists an optimal fractional schedule $X_\mathcal{Q}^\ast = (x^\ast_1, \dots, x^\ast_T)$ for $\bar{\mathcal{Q}}$ such that $|\hat{x}^k_t / 2^l - x^\ast_t| \leq 2^{k-l}$. By Lemma~\ref{lemma:poly:correct:rounding}, $\lfloor X_\mathcal{Q}^\ast \rfloor$ is also an optimal schedule for $\bar{\mathcal{Q}}$ and therefore it is also optimal for $\mathcal{Q}$. The inequality $|\hat{x}^k_t / 2^l - \lfloor x^\ast_t \rfloor| \leq 2^{k-l}$ still holds, because the terms $\hat{x}^k_t / 2^l$ and $2^{k-l}$ are integral and therefore adding a value less than 1 to the left side cannot invalidate the inequality. Let $\hat{X}^l \coloneqq (\lfloor x^\ast_1 \rfloor \cdot 2^l, \dots, \lfloor x^\ast_T \rfloor \cdot 2^l)$. As $\lfloor X_\mathcal{Q}^\ast \rfloor$ is optimal for $\mathcal{Q}$, $\hat{X}^l$ must be optimal for $\mathcal{P}_l$. Furthermore, $\lfloor x^\ast_t \rfloor = \hat{x}^l_t / 2^l$ holds, so we can insert it into the above inequality and get $|\hat{x}^k_t / 2^l - \hat{x}^l_t / 2^l| \leq 2^{k-l}$ which is equivalent to $|\hat{x}^k_t - \hat{x}^l_t| \leq 2^k$.
\end{proof}

Now, we have proven all parts to show the correctness of our polynomial-time algorithm:

\begin{theorem}
	The algorithm described in Section~\ref{sec:poly:polytimealg} is correct.
\end{theorem}

\begin{proof}
	We will show the correctness by induction. 
	In the first iteration, the algorithm finds an optimal schedule for $\mathcal{P}_K$, because all states of $M_K$ are considered.
	
	Given an optimal schedule $\hat{X}^k$, in the next iteration the algorithm only considers the states $x_t \in M_{k-1}$ with $|\hat{x}^k_t - x_t| \leq 2^k$. By Lemma~\ref{lemma:poly:correct:iteration}, there exists an optimal schedule $\hat{X}^l$ with $l = k-1$ such that $|\hat{x}^k_t - \hat{x}^l_t| \leq 2^k$ holds. Therefore, the schedule found in iteration $k-1$ must be optimal for $\mathcal{P}_{k-1}$ (although some states are ignored by the algorithm). Thus, by induction, the algorithm will find an optimal schedule for $\mathcal{P}_0 = \mathcal{P}$ in the last iteration.
\end{proof}

\section{Deterministic online algorithm} 
\label{sec:lcp}
Lin et. al.~\cite{LinWierman2011infocom,LinWierman2013} developed an algorithm called \emph{Lazy Capacity Provisioning} (LCP) that achieves a competitive ratio of 3 for the continuous 
setting (i.e., $x_t \in \mathbb{R}$). In this section, we adapt LCP to the discrete data-center optimization problem and prove that the algorithm is
3-competitive for this problem as well.

The general approach of our proof is similar to the proof of the continuous setting in \cite{LinWierman2011infocom}. Some lemmas (e.g., Lemma~\ref{lemma:lcp:cf:lu} and~\ref{lemma:lcp:cf:optrec}) were adopted, however, their proofs are completely different. Lin et. al. use the properties of the convex program, especially duality and the complementary slackness conditions. This approach cannot be adapted to the discrete setting.  

\subsection{Algorithm}
\label{sec:lcp:alg}

First, we will define lower and upper bounds for the optimal offline solution that can be calculated online. For a given time slot $\tau$, let $X^L_\tau \coloneqq (x^L_{\tau,1}, \dots , x^L_{\tau, \tau} )$ be the vector that minimizes 
\begin{equation} \label{eqn:lcp:algo:lower}
C^L_\tau(X) = \sum_{t=1}^\tau f_t(x_t) + \beta \sum_{t=1}^\tau (x_{t} - x_{t-1})^+
\end{equation}
with $X = (x_1, \dots, x_\tau)$. This term describes the cost of a workload that ends at $\tau \leq T$. For $\tau = T$ this equation is equivalent to~\eqref{eqn:model:cost} . Let $x^L_\tau \coloneqq x^L_{\tau, \tau}$ be the last state for this truncated workload. If there is more than one vector that minimizes~\eqref{eqn:lcp:algo:lower}, then $x^L_\tau$ is defined as the smallest possible value.

Similarly, let $X^U_\tau \coloneqq (x^U_{\tau,1}, \dots , x^U_{\tau, \tau} )$ be the vector that minimizes 
\begin{equation} \label{eqn:lcp:algo:upper}
C^U_\tau (X) = \sum_{t=1}^\tau f_t(x_t) + \beta \sum_{t=1}^\tau (x_{t-1} - x_{t})^+ .
\end{equation}
The difference to the equation~\eqref{eqn:lcp:algo:lower} is that we pay the switching cost for powering down. Powering up does not cost anything. The last state is denoted by $x^U_\tau \coloneqq x^U_{\tau, \tau}$. If there is more than one vector that minimizes~\eqref{eqn:lcp:algo:upper}, then  $x^U_\tau$ is the largest possible value.

Define $[x]^b_a \coloneqq \max \{a, \min\{b, x\}\}$ as the projection of $x$ into the interval $[a, b]$. The LCP algorithm is defined as follows:
\begin{equation} \label{eqn:lcp:algo:lcp}
x^\text{LCP}_\tau \coloneqq 
\begin{cases}
0, & \tau = 0 \\
[x^\text{LCP}_{\tau - 1}]^{x^U_\tau} _{x^L_\tau}, & \tau \geq 1.
\end{cases}
\end{equation}

Before we can prove that this algorithm is 3-competitive, we have to introduce some notation.

\subsection{Notation}
\label{sec:lcp:algo:nota}

Let $X^\ast = (x^\ast_1, \dots, x^\ast_T)$ be an optimal offline solution that minimizes equation~\eqref{eqn:model:cost} (i.e., the whole workload). Note that $C^L_\tau(X^\ast)$ indicates the cost of the optimal solution until $\tau$. 

Let $R_\tau(X) \coloneqq \sum_{t=1}^{\tau} f_t(x_t)$ with $X = (x_1, \dots, x_\tau )$ denote the operating cost until $\tau$, let $S^L_\tau(X) \coloneqq \beta \sum_{t=1}^{\tau} (x_t - x_{t-1})^+$ denote the switching cost in $C^L_\tau(X)$ and let $S^U_\tau(X) \coloneqq \beta \sum_{t=1}^{\tau} (x_{t-1} - x_{t})^+$ denote the switching cost in $C^U_\tau(X)$. Note that $C^L_\tau(X) = R_\tau(X) + S^L_\tau(X)$ and $C^U_\tau(X) = R_\tau(X) + S^U_\tau(X)$. Furthermore,
\begin{equation} \label{eqn:lcp:algo:nota:slurelation}
S^L_\tau(X) = S^U_\tau(X) + \beta x_\tau
\end{equation}
as well as $C^L_\tau(X) = C^U_\tau(X) + \beta x_\tau$ holds, because in $C^L_\tau$ we have to pay the missing switching cost to reach the final state $x_\tau$. Note that $\beta x_\tau$ equals the cost for powering up in $C^L_\tau$ minus the cost for powering down in $C^U_\tau$. 

Given an arbitrary function $g : [m]_0 \rightarrow \mathbb{R}$, we define 
\begin{equation*}
\Delta g(x) \coloneqq g(x) - g(x-1)
\end{equation*}
as the slope of $g$ at $x$.
Let 
\begin{equation*}
\hat{C}^B_\tau(x) \coloneqq \min_{x_1, \dots, x_{\tau-1}} C^B_\tau((x_1, \dots, x_{\tau-1}, x))
\end{equation*}
with $B \in \{L, U\}$ be the  minimal cost achievable with $x_\tau = x$.

\subsection{Competitive ratio}
\label{sec:lcp:cf}

In this section, we prove that the LCP algorithm described by equation~\eqref{eqn:lcp:algo:lcp} achieves a competitive ratio of 3. First, we show that the optimal solution is bounded by the upper and lower bounds defined in the previous section.

\begin{lemma} \label{lemma:lcp:cf:lu}
	Let $X^\ast$ be an arbitrary optimal schedule. For all $\tau$, $x^L_\tau \leq x^\ast_\tau \leq x^U_\tau$ holds.
\end{lemma}

\begin{proof}
	We prove both parts of the inequality by contradiction:
	
	\textbf{Part 1 ($x^L_{\tau} \leq x^\ast_{\tau}$):} Assume that $x^L_{\tau} > x^\ast_{\tau}$. By the definition of the lower bound, $C^L_\tau(X^L_\tau) < C^L_\tau(X^\ast)$ holds and  we can replace $(x^\ast_1, \dots,  x^\ast_\tau)$ by $(x^L_{\tau,1}, \dots x^L_{\tau, \tau})$. This reduces the total cost of $x^\ast$, because the cost up to $\tau$ is reduced and for $\tau +1$ there is no additional switching cost because $x^L_{\tau} > x^\ast_{\tau}$ holds. Thus, $X^\ast$ would not be an optimal solution which is a contradiction, so $x^L_{\tau} \leq x^\ast_{\tau}$ must be fulfilled.
	
	\textbf{Part 2 ($x^\ast_{\tau} \leq x^U_{\tau}$):\ } Assume that $x^\ast_{\tau} > x^U_{\tau}$. 
	By definition of the upper bound, $C^U_\tau(X^U_\tau) < C^U_\tau(X^\ast)$ and thus,  
	\begin{equation} \label{eqn:lcp:algo:lu:proof}
	R_\tau(X^U_\tau) + S^U_\tau(X^U_\tau) < R_\tau(X^\ast) + S^U_\tau(X^\ast)
	\end{equation}
	holds. The cost of the optimal solution until $\tau$ is 
	$R_\tau(X^\ast) + S^L_\tau(X^\ast)$.
	If the states $(x^\ast_1, \dots , x^\ast_\tau)$ are replaced by $X^U_\tau$ and afterwards $ x^\ast_\tau - x^U_{\tau, \tau}$ servers are powered up (to ensure that we end in the same state), then the cost is $R_\tau(X^U_\tau) + S^L_\tau(X^U_\tau) + \beta (x^\ast_\tau - x^U_{\tau, \tau})$.
	This cost must be greater than or equal to the cost of the optimal solution, so
	\begin{equation*}
	R_\tau(X^U_\tau) + S^L_\tau(X^U_\tau) + \beta (x^\ast_\tau - x^U_{\tau, \tau}) \geq R_\tau(X^\ast) + S^L_\tau(X^\ast) 
	\end{equation*}
	holds. By using equation~\eqref{eqn:lcp:algo:nota:slurelation}, we get
	\begin{equation*}
	R_\tau(X^U_\tau) + S^U_\tau(X^U_\tau) + \beta x^U_{\tau, \tau} + \beta (x^\ast_\tau - x^U_{\tau, \tau}) \geq R_\tau(X^\ast) + S^U_\tau(X^\ast) + \beta x^\ast_\tau .
	\end{equation*}
	We eliminate identical terms and get
	\begin{equation*}
	R_\tau(X^U_\tau) + S^U_\tau(X^U_\tau) \geq R_\tau(X^\ast) + S^U_\tau(X^\ast)
	\end{equation*}
	which is a contradiction to inequality~\eqref{eqn:lcp:algo:lu:proof}. Therefore, our assumption was wrong, so $x^\ast_\tau \leq x^U_\tau$ must be fulfilled.
\end{proof}

The following four lemmas show important properties of $\hat{C}^L_\tau(x)$. First, we prove that the relation between $C^L_\tau(X)$ and $C^U_\tau(X)$ described by equation~\eqref{eqn:lcp:algo:nota:slurelation} still holds for $\hat{C}^L_\tau(x)$ and $\hat{C}^U_\tau(x)$. 

\begin{lemma} \label{lemma:lcp:cf:hatcluconvert}
	For all $\tau$, $\hat{C}^L_\tau (x) = \hat{C}^U_\tau (x) + \beta x$ holds.
\end{lemma}

\begin{proof}
	Let $X^L$ be a corresponding solution for $\hat{C}^L_\tau (x)$ such that $C^L_\tau(X^L) = \hat{C}^L_\tau(x)$ and let $X^U$ be a corresponding solution for $\hat{C}^U_\tau(x)$ such that $C^U_\tau (X^U) = \hat{C}^U_\tau (x)$. Note that the last state of $X^L$ and $X^U$ is $x$. Since $X^U$ is optimal for $C^U_\tau$, the inequality $C^U_\tau (X^U) \leq C^U_\tau(X)$ holds for all $X = (x_1, \dots, x_{\tau-1}, x)$. By equation~\eqref{eqn:lcp:algo:nota:slurelation}, we get 
	\begin{equation*}
	C^L_\tau (X^U) - \beta x \leq C^L_\tau(X) - \beta x
	\end{equation*}
	which is equivalent to $C^L_\tau (X^U)  \leq C^L_\tau(X)$. With $X \coloneqq X^L$, we get $C^L_\tau(X^U) \leq C^L_\tau(X^L)$. Since $X^L$ is optimal for $C^L_\tau$, $X^U$ must be optimal too, so $C^L_\tau (X^U) = C^L_\tau(X^L)$ holds. All in all, we get 
	\begin{equation*}
	\hat{C}^L_\tau (x) = C^L_\tau(X^L) = C^L_\tau (X^U) = C^U_\tau (X^U) + \beta x = \hat{C}^U_\tau (x) + \beta x . \qedhere
	\end{equation*}
\end{proof}

Obviously, the cost functions $C^L_\tau(X)$ and $C^U_\tau(X)$ are convex, since convexity is closed under addition. The following lemma shows that also $\hat{C}^L_\tau(x)$ and $\hat{C}^U_\tau(x)$ are convex.

\begin{lemma} \label{lemma:lcp:cf:convex}
	For all $\tau$ and $B \in \{L, U\}$, $\hat{C}^B_\tau(x)$ is a convex function.
\end{lemma}

We will prove this lemma together with the next lemma:

\begin{lemma} \label{lemma:lcp:cf:slope}
	The slope of $\hat{C}^L_\tau(x)$ is at most $\beta$ for $x \leq x^U_\tau$ and at least $\beta$ for $x > x^U_\tau$, i.e., $\Delta \hat{C}^L_\tau(x^U_\tau) \leq \beta$ and $\Delta \hat{C}^L_\tau(x^U_\tau+1) \geq \beta$ 
\end{lemma}

\begin{proof}[Proof of Lemmas~\ref{lemma:lcp:cf:convex} and~\ref{lemma:lcp:cf:slope}]
	First, we will prove the case $B = L$ by induction. The function 
	\begin{equation*}
	\hat{C}^L_1(x) = f_1(x) + \beta x
	\end{equation*} is convex, because all $f_t$ are convex and $\beta x$ is a linear function which is also convex (note that convexity is closed under addition). For $\hat{C}^U_\tau$ there are no costs for powering up, so $x^U_1 = \argmin_{x} f_1(x)$ and therefore
	\begin{align*}
	\Delta \hat{C}^L_1(x^U_1) &= \hat{C}^L_1(x^U_1) - \hat{C}^L_1(x^U_1 - 1) \\
	&= f_1(x^U_1) - f_1(x^U_1 - 1) + \beta \\
	&\leq \beta
	\end{align*}
	and
	\begin{align*}
	\Delta \hat{C}^L_1(x^U_1+1) &= \hat{C}^L_1(x^U_1 +1) - \hat{C}^L_1(x^U_1) \\
	&= f_1(x^U_1 + 1) - f_1(x^U_1) + \beta \\
	&\geq \beta,
	\end{align*}
	so for $\tau = 1$ both lemmas are fulfilled.
	
	Assume that $\hat{C}^L_{\tau -1}$ is convex, $\Delta \hat{C}^L_{\tau-1}(x^U_{\tau-1}) \leq \beta$ and $\Delta \hat{C}^L_{\tau-1}(x^U_{\tau-1}+1) \geq \beta$. By definition, we have 
	\begin{align*}
	\hat{C}^L_\tau(x) 
	&= \min_{x'} \left\{\hat{C}^L_{\tau-1}(x') + \beta(x-x')^+\right\} + f_\tau(x) .
	\end{align*}
	Let $x'_\text{min} \coloneqq \argmin_{x'} \hat{C}^L_{\tau-1}(x')$. 
	If $x \leq x'_\text{min}$, then 
	\begin{equation*}
	\min_{x'} \left\{\hat{C}^L_{\tau-1}(x') + \beta(x-x')^+\right\} = C^L_{\tau-1}(x'_\text{min})
	\end{equation*}
	holds, so 
	\begin{equation} \label{eqn:lcp:cf:convex:la}
	\hat{C}^L_\tau(x) = \hat{C}^L_{\tau-1}(x'_\text{min}) + f_\tau(x) 
	\end{equation}
	is convex for $x \leq x'_\text{min}$.
	
	Now we consider the case $x > x'_\text{min}$. It is clear that $x'_\text{min} \leq x' \leq x$, because for $x' > x$ the term $\beta (x - x')^+$ is zero. We differ between two cases:
	
	If $x \leq x^U_{\tau - 1}$, then $\Delta \hat{C}^L_{\tau-1}(x) \leq \beta$ holds, since $\Delta \hat{C}^L_{\tau-1}(x^U_{\tau-1}) \leq \beta$ and $\hat{C}^L_{\tau-1}$ is convex. Therefore, $x' = x$ minimizes the term $\hat{C}^L_{\tau-1}(x') + \beta(x-x')^+$ because using a smaller state $\tilde{x} < x$ instead of $x' = x$ would increase the switching cost by $\beta (\tilde{x} - x')$ while the decrease of $\hat{C}^L_{\tau - 1}$ is less than or equal to $\beta (\tilde{x} - x')$. Thus,
	\begin{equation*}
	\min_{x'} \left\{\hat{C}^L_{\tau-1}(x') + \beta(x-x')^+\right\} = \hat{C}^L_{\tau-1}(x)
	\end{equation*}
	and
	\begin{equation} \label{eqn:lcp:cf:convex:lb}
	\hat{C}^L_\tau (x) = \hat{C}^L_{\tau - 1}(x) + f_\tau(x)
	\end{equation}
	is convex for $x'_\text{min} < x \leq x^U_{\tau - 1}$.
	
	If $x > x^U_{\tau - 1}$, then $x' = x^U_{\tau - 1}$, because using a greater state $\tilde{x} > x^U_{\tau - 1}$ would increase the value of $C^L_{\tau - 1}(x')$ by at least $\beta (\tilde{x} - x^U_{\tau - 1})$ while the switching cost is decreased by $\beta (\tilde{x} - x^U_{\tau - 1})$. Analogously, using a smaller state $\tilde{x} < x^U_{\tau - 1}$ would decrease the value of $C^L_{\tau - 1}(x')$ by at most $\beta ( x^U_{\tau - 1} - \tilde{x})$ while the switching cost is increased by $\beta ( x^U_{\tau - 1} - \tilde{x})$. Thus,
	\begin{equation*}
	\min_{x'} \left\{\hat{C}^L_{\tau-1}(x') + \beta(x-x')^+\right\} = \hat{C}^L_{\tau-1}(x^U_{\tau - 1}) + \beta(x-x^U_{\tau - 1})
	\end{equation*}
	and
	\begin{equation} \label{eqn:lcp:cf:convex:lc}
	\hat{C}^L_\tau (x) = \hat{C}^L_{\tau - 1}(x^U_{\tau -1}) - \beta x^U_{\tau -1} + \beta x + f_\tau(x)
	\end{equation}
	is convex for $x > x^U_{\tau -1}$.
	
	To show that $\hat{C}^L_\tau(x)$ is convex for all $x$, we have to compare the slopes of the edge cases. Note that equation~\eqref{eqn:lcp:cf:convex:la} and~\eqref{eqn:lcp:cf:convex:lb} as well as~\eqref{eqn:lcp:cf:convex:lb} and~\eqref{eqn:lcp:cf:convex:lc} have the same values for $x = x'_\text{min}$ and $x = x^U_{\tau-1}$, respectively. We have to show that
	\begin{align}
	\Delta \hat{C}^L_\tau (x'_\text{min}) &\leq \Delta \hat{C}^L_\tau (x'_\text{min} +1) 
	\label{eqn:lcp:cf:convex:delta:ab} \\
	\Delta \hat{C}^L_\tau (x^U_{\tau - 1}) &\leq \Delta \hat{C}^L_\tau (x^U_{\tau - 1} + 1).
	\label{eqn:lcp:cf:convex:delta:bc}
	\end{align}
	First, we will prove~\eqref{eqn:lcp:cf:convex:delta:ab} by using equations~\eqref{eqn:lcp:cf:convex:lb} and~\eqref{eqn:lcp:cf:convex:la}:
	\begin{align*}
	 \Delta \hat{C}^L_\tau (x'_\text{min} +1) 
	\alignstack{\eqref{eqn:lcp:cf:convex:lb}}{=}{}& \hat{C}^L_{\tau-1}(x'_\text{min} +1) - \hat{C}^L_{\tau-1}(x'_\text{min}) + f_\tau(x'_\text{min}+1) - f_\tau(x'_\text{min} ) \\
	\geq{}& f_\tau(x'_\text{min}+1) - f_\tau(x'_\text{min} ) \\
	\geq{}&  f_\tau(x'_\text{min}) - f_\tau(x'_\text{min} - 1) \\
	\alignstack{\eqref{eqn:lcp:cf:convex:la}}{=}{}& \Delta \hat{C}^L_\tau (x'_\text{min}) .
	\end{align*}
	The first inequality holds since $x'_\text{min}$ minimizes $\hat{C}^L_{\tau-1}$, the second inequality uses the convexity of~$f_\tau$.
	
	Inequality~\eqref{eqn:lcp:cf:convex:delta:bc} can be shown as follows:
	\begin{align*}
	 \Delta \hat{C}^L_\tau (x^U_{\tau - 1}) 
	\alignstack{\eqref{eqn:lcp:cf:convex:lb}}{=}{}& \hat{C}^L_{\tau-1}(x^U_{\tau - 1}) - \hat{C}^L_{\tau-1}(x^U_{\tau - 1} - 1) + f_\tau(x^U_{\tau - 1}) - f_\tau(x^U_{\tau - 1} - 1 ) \\
	={}& \Delta \hat{C}^L_{\tau-1}(x^U_{\tau - 1}) + f_\tau(x^U_{\tau - 1}) - f_\tau(x^U_{\tau - 1} - 1 ) \\
	\alignstack{\text{IH}}{\leq}{}& \beta + f_\tau(x^U_{\tau - 1}) - f_\tau(x^U_{\tau - 1} - 1 ) \\
	\leq{}& \beta + f_\tau(x^U_{\tau - 1} + 1) - f_\tau(x^U_{\tau - 1} ) \\
	\alignstack{\eqref{eqn:lcp:cf:convex:lc}}{=}{}& \Delta \hat{C}^L_\tau (x^U_{\tau - 1} + 1) .
	\end{align*}
	The first inequality uses the induction hypothesis ($\Delta \hat{C}^L_{\tau-1}(x^U_{\tau-1}) \leq \beta$), the second inequality holds due to the convexity of $f_\tau$.
	
	Now, we know that $\hat{C}^L_\tau$ is convex. We still have to show the slope property of Lemma~\ref{lemma:lcp:cf:slope}. We begin showing $\Delta \hat{C}^L_\tau (x^U_{\tau}) \leq \beta$ by using Lemma~\ref{lemma:lcp:cf:hatcluconvert}.
	\begin{align*}
	\Delta \hat{C}^L_\tau (x^U_{\tau}) &\;=\; \hat{C}^L_\tau (x^U_{\tau}) - \hat{C}^L_\tau (x^U_{\tau} - 1) \\
	&\alignstack{\text{\tiny L\ref{lemma:lcp:cf:hatcluconvert}}}{\;=\;} \hat{C}^U_\tau (x^U_{\tau}) + \beta x^U_\tau - \hat{C}^U_\tau (x^U_{\tau} - 1) - \beta (x^U_\tau - 1)  \\
	&\;\leq\; \beta .
	\end{align*}
	The inequality holds, because $x^U_\tau$ minimizes $\hat{C}^U_\tau$, so $\hat{C}^U_\tau (x^U_{\tau}) - \hat{C}^L_\tau (x^U_{\tau} - 1) \leq 0$. 
	
	The same arguments can be used to show that $\Delta \hat{C}^L_\tau (x^U_{\tau} + 1) \geq \beta$.
	\begin{align*}    
	\Delta \hat{C}^L_\tau (x^U_{\tau} + 1) &\;=\; \hat{C}^L_\tau (x^U_{\tau} + 1) - \hat{C}^L_\tau (x^U_{\tau}) \\
	&\alignstack{\text{\tiny L\ref{lemma:lcp:cf:hatcluconvert}}}{\;=\;} \hat{C}^U_\tau (x^U_{\tau}+1) + \beta (x^U_\tau+1) - \hat{C}^U_\tau (x^U_{\tau}) - \beta (x^U_\tau) \\
	&\;\geq\; \beta .
	\end{align*}
	The inequality holds because  $\hat{C}^U_\tau (x^U_{\tau}+1) - \hat{C}^L_\tau (x^U_{\tau}) \geq 0$.
	
	Since $\hat{C}^L_\tau(x)$ is convex, by Lemma~\ref{lemma:lcp:cf:hatcluconvert}, $\hat{C}^U_\tau(x) = \hat{C}^L_\tau(x) + \beta x$ is convex too, because convexity is closed under addition. 
\end{proof}

\begin{lemma} \label{lemma:lcp:cf:slope:allx}
	For $x \leq x^U_\tau$, the slope of $\hat{C}^L_\tau(x)$ is at most $\beta$, i.e., $\Delta\hat{C}^L_\tau(x) \leq \beta$ holds. 
\end{lemma}

\begin{proof}
	By Lemma~\ref{lemma:lcp:cf:slope}, $\Delta\hat{C}^L_\tau(x^U_\tau) \leq \beta$ holds and by Lemma~\ref{lemma:lcp:cf:convex}, $\hat{C}^L_\tau$ is convex, so $\Delta\hat{C}^L_\tau(x) \leq \beta$ holds for $x \leq x^U_\tau$.
\end{proof}

The next lemma characterizes the behavior of the optimal solution backwards in time.

\begin{lemma} \label{lemma:lcp:cf:optrec}
	A solution vector $(\hat{x}_1, \dots,  \hat{x}_T)$ that fulfills the following recursive equality for all $t \in [T]$ is optimal:
	\begin{equation*}
	\hat{x}_t \coloneqq 
	\begin{cases}
	0, & t = T+1 \\
	[\hat{x}_{t + 1}]^{x^U_t} _{x^L_t}, & t \leq T .
	\end{cases}
	\end{equation*}
\end{lemma}

\begin{proof}
	We will prove the lemma by induction in reverse time. Powering down does not cost anything, so setting $\hat{x}_{T+1} = 0$ does not produce any additional costs. Assume that $(\hat{x}_{\tau+1}, \dots, \hat{x}_T)$ can lead to an optimal solution, i.e., there exists an optimal solution $X^\ast$ with $x^\ast_t = \hat{x}_t$ for $t \geq \tau + 1$. We will show that the vector $(\hat{x}_{\tau}, \dots, \hat{x}_T)$ can still lead to an optimal solution. 
	
	We have to examine three cases:
	
	\textbf{Case 1:\ } If $\hat{x}_{\tau+1} < x^L_\tau$, then $\hat{x}_\tau = x^L_\tau$. By Lemma~\ref{lemma:lcp:cf:lu}, $x^\ast_\tau \geq x^L_\tau$ holds. Since $X^L_\tau$ minimizes $C^L_\tau$, we know that $C^L_\tau(X^L_\tau) \leq C^L_\tau(X)$ for all $X = (x_1,\dots, x_\tau)$. Thus, there is no benefit to use a state $x' \geq x^L_\tau$, because afterwards we have to power down some servers to reach $\hat{x}_{\tau+1}$. Therefore, $\hat{x}_\tau = x^L_\tau$ can still lead to an optimal solution.
	
	\textbf{Case 2:\ } If $\hat{x}_{\tau +1} > x^U_\tau$, then $\hat{x}_\tau = x^U_\tau$. By Lemma~\ref{lemma:lcp:cf:lu}, $x^\ast_\tau \leq x^U_\tau$ holds. Since $X^U_\tau$ minimizes $C^U_\tau$, we know that $C^U_\tau(X^U_\tau) \leq C^U_\tau(X)$ for all $X$. By using the solution $X^U_\tau$ and then switching to state $\hat{x}_{\tau + 1}$, the resulting cost is
	\begin{alignat*}{6}
	&C^L_\tau(X^U_\tau)  &{}+{}&\beta (\hat{x}_{\tau + 1} - x^U_\tau) &{}+{}& f_{\tau + 1}(\hat{x}_{\tau + 1}) \\
	={}&C^U_\tau(X^U_\tau) &{}+{}&  \beta \hat{x}_{\tau + 1}  &{}+{}& f_{\tau + 1}(\hat{x}_{\tau + 1}) \\
	\leq{}& C^U_\tau(X) &{}+{}&  \beta \hat{x}_{\tau + 1} &{}+{}& f_{\tau + 1}(\hat{x}_{\tau + 1}) \\
	={}& C^L_\tau(X)  &{}+{}& \beta (\hat{x}_{\tau + 1} -  x_\tau) &{}+{}& f_{\tau + 1}(\hat{x}_{\tau + 1}) .
	\end{alignat*}	
	The last line describes the cost until $\tau + 1$ by using the schedule $X = (x_1, \dots, x_\tau)$ with $x_\tau \leq \hat{x}_{\tau + 1}$ instead of $X^U_\tau$. The cost is not reduced by using $X$, so $\hat{x}_\tau = x^U_\tau$ can still lead to an optimal solution.
	
	\textbf{Case 3:\ } If $x^L_\tau \leq \hat{x}_{\tau +1} \leq x^U_\tau$, then $\hat{x}_\tau = \hat{x}_{\tau + 1}$. Assume that there is a better state $\hat{x}^-_\tau < \hat{x}_\tau$ 
	such that 
	\begin{equation} \label{eqn:lcp:cf:optrec:a}
	\hat{C}^L_\tau(\hat{x}^-_\tau) + \beta (\hat{x}_{\tau + 1} - \hat{x}^-_\tau) < \hat{C}^L_\tau(\hat{x}_ \tau) .
	\end{equation}
	In other words, we assume that using $\hat{x}^-_\tau$ servers at time $\tau$ and then powering up the missing $\hat{x}_{\tau + 1} - \hat{x}^-_\tau$ servers is cheaper than using $\hat{x}_\tau$ servers. 
	By Lemma~\ref{lemma:lcp:cf:slope:allx}, we know that the slope of $\hat{C}^L_\tau(x)$ is at most $\beta$ for $x \leq x^U_\tau$. This leads to the contradiction
	\begin{equation*}
	\hat{C}^L_\tau( \hat{x}_\tau) - \hat{C}^L_\tau(\hat{x}^-_ \tau) 
	\leq \beta (\hat{x}_\tau - \hat{x}^-_\tau) 
	= \beta (\hat{x}_{\tau+1} - \hat{x}^-_\tau) 
	\alignstack{\eqref{eqn:lcp:cf:optrec:a}}{<} \hat{C}^L_\tau( \hat{x}_\tau) - \hat{C}^L_\tau(\hat{x}^-_ \tau)   .
	\end{equation*}
	Therefore, there is no $\hat{x}^-_\tau$ with the desired properties.
	
	The other case is more simple: Assume that there is a better state $\hat{x}^+_\tau > \hat{x}_\tau$ with $\hat{C}^L_\tau(\hat{x}^+_\tau) < \hat{C}^L_\tau(\hat{x}_ \tau)$, then $x^L_\tau$ (which minimizes $\hat{C}^L_\tau$) must be greater than $\hat{x}_\tau$, because, by Lemma~\ref{lemma:lcp:cf:convex}, $\hat{C}^L_\tau$ is a convex function. However, this is a contradiction to $x^L_\tau \leq \hat{x}_{\tau +1} = \hat{x}_\tau$. 
\end{proof}

In the following $X^\ast = (x^\ast_1, \dots, x^\ast_T)$ denotes an optimal solution that fulfills the recursive equality of Lemma~\ref{lemma:lcp:cf:optrec}. The next lemma describes time slots where $X^\text{LCP}$ and $X^\ast$ are in same state. Informally, the lemma says that if the LCP curve cuts the optimal solution, then there is one time slot $\tau$ where both solutions are in the same state.

\begin{lemma} \label{lemma:lcp:cf:cut}
	If $x^\text{LCP}_{\tau - 1} < x^\ast_{\tau - 1}$ and $x^\text{LCP}_{\tau} \geq x^\ast_{\tau}$, then $x^\text{LCP}_{\tau} = x^\ast_{\tau}$. 
	
	\noindent $\phantom{\textbf{Lemma 12. }}$\hspace{-5pt}
	If $x^\text{LCP}_{\tau - 1}  > x^\ast_{\tau - 1} $ and $x^\text{LCP}_{\tau} \leq x^\ast_{\tau}$, then $x^\text{LCP}_{\tau} = x^\ast_{\tau}$.
\end{lemma}

\begin{proof}
	We will only show the first statement of the lemma, since the other one works exactly analogously.
	Assume that $x^\text{LCP}_{\tau - 1}  < x^\ast_{\tau - 1} $ and $x^\text{LCP}_{\tau} \geq x^\ast_{\tau}$ holds. We differ between two cases. 
	
	\textbf{Case 1:\ } If $x^\text{LCP}_{\tau - 1}  < x^\text{LCP}_{\tau}$, then $x^\text{LCP}_{\tau} = x^L_{\tau}$ (by the definition of the LCP algorithm). By $x^\text{LCP}_{\tau} \geq x^\ast_{\tau}$ and Lemma~\ref{lemma:lcp:cf:lu} (which says that $x^L_{\tau} \leq x^\ast_\tau$), we get $x^\text{LCP}_{\tau} = x^\ast_{\tau}$.
	
	\textbf{Case 2:\ } If $x^\text{LCP}_{\tau - 1} \geq x^\text{LCP}_{\tau}$, then $x^\ast_{\tau - 1} > x^\ast_\tau$ (since $x^\ast_{\tau - 1} > x^\text{LCP}_{\tau - 1}  \geq x^\text{LCP}_{\tau} \geq x^\ast_\tau$). By Lemma~\ref{lemma:lcp:cf:optrec}, $x^\ast_{\tau - 1} = x^L_{\tau - 1}$ holds which is a contradiction to $ x^\ast_{\tau - 1} > x^\text{LCP}_{\tau - 1} \geq x^L_{\tau - 1}$, so this case never occurs. \qedhere
	
%
%
\end{proof}

The time slots where $x^\text{LCP}_t = x^\ast_t$ are denoted by $0 = t_0 < t_1 < \dots < t_\kappa$. Between these time slots it is not possible that $X^\text{LCP}$ powers one or more servers down and $X^\ast$ powers servers up or vice versa. 
At the end of the time horizon, $X^\text{LCP}$ and $X^\ast$ can be in different states. To get rid of this special case, let $f_{T+1}(x) \coloneqq (\beta + \epsilon) x$ for any $\epsilon > 0$ such that $x^U_{T+1} = x^L_{T+1} = 0$. Therefore, $x^\text{LCP}_{T+1} = x^\ast_{T+1} = 0$ holds, so $t_{\kappa + 1} \coloneqq T + 1$. 

In the following $[a:b]$ with $a,b \in \mathbb{N}$ denotes the set $\{a, a+1, \dots, b\}$. Analogously, we define $[a:b[ {} \coloneqq \{a, a+1 \dots, b-1\}$, $]a:b] {}\coloneqq \{a+1, a+2, \dots, b\}$ and $]a:b[ {}\coloneqq \{a+1, a+2, \dots b-1\}$. 

\begin{lemma} \label{lemma:lcp:cf:monotony}
	\renewcommand{\labelenumi}{(\roman{enumi})}
	For all time intervals $]t_{i} : t_{i+1}[$ with $i \in [\kappa]_0$, either
	\begin{enumerate}
		\item $x^\text{LCP}_\tau > x^\ast_ \tau$ and both $x^\text{LCP}_\tau$ and $x^\ast_ \tau$ are non-increasing for all $\tau \in ]t_{i} : t_{i+1}[$, or
		\item $x^\text{LCP}_\tau < x^\ast_ \tau$ and both $x^\text{LCP}_\tau$ and $x^\ast_ \tau$ are non-decreasing for all $\tau \in ]t_{i} : t_{i+1}[$.
	\end{enumerate}
\end{lemma}

\begin{proof}
	First, we consider case (i). Let $x^\text{LCP}_{\tau} > x^\ast_{\tau}$ for any $\tau \in ]t_{i} : t_{i+1}[$. By Lemma~\ref{lemma:lcp:cf:cut}, this inequality holds for all $\tau \in ]t_{i} : t_{i+1}[$. 
	
	Assume that $x^\text{LCP}_{\tau+1} > x^\text{LCP}_{\tau}$. Then $x^L_{\tau + 1} = x^\text{LCP}_{\tau+1}$ by the LCP algorithm and $x^\ast_{\tau+ 1} \geq x^L_{\tau+ 1}$ by Lemma~\ref{lemma:lcp:cf:lu}. By Lemma~\ref{lemma:lcp:cf:optrec}, we get $x^U_\tau = x^\ast_\tau$ which leads to the contradiction $x^U_\tau = x^\ast_\tau < x^\text{LCP}_\tau \leq x^U_\tau$ (the last inequality uses the definition of the LCP algorithm). Thus, $x^\text{LCP}_{\tau}$ is non-increasing for all $\tau \in ]t_{i} : t_{i+1}[$. 
	
	Assume that $x^\ast_{\tau+1} > x^\ast_{\tau}$. Then  $x^\ast_\tau = x^U_\tau$ by Lemma~\ref{lemma:lcp:cf:optrec} which is a contradiction to $x^U_\tau \geq x^\text{LCP}_\tau > x^\ast_{t}$, so $x^\ast_\tau$ is also non-increasing for all $\tau \in ]t_{i} : t_{i+1}[$.
	
	Case (ii) works analogously.
%
%
\end{proof}

Now, we can calculate the switching cost of the LCP algorithm.

\begin{lemma} \label{lemma:lcp:cf:switching}
	$S^L_T(X^\text{LCP}) \leq	S^L_T(X^\ast)$
\end{lemma}

\begin{proof}
	By Lemma~\ref{lemma:lcp:cf:monotony}, both $x^\text{LCP}_\tau$ and $x^\ast_ \tau$ are either non-in\-creas\-ing or non-decreasing until there is a time slot $t$ with $x^\text{LCP}_t = x^\ast_ t$. Therefore, the switching cost during each time interval $[t_{i} : t_{i+1}]$ with $i \in [\kappa]_0$ is $\beta(x^\ast_{t_i} - x^\ast_{t_{i-1}})^+$ for both $X^\text{LCP}$ and $X^\ast$. By adding the switching costs of all intervals, we get $S^L_T(X^\text{LCP}) \leq S^L_T(X^\ast)$. 
\end{proof}

Lemma~\ref{lemma:lcp:cf:monotony} divides the intervals $[t_i : t_{i+1}[$ into two sets: Intervals of case (i) are called \emph{decreasing} intervals, the set of those intervals is denoted by $\mathcal{T}^-$. Intervals of case (ii) are called \emph{increasing} intervals and the set is denoted by $\mathcal{T}^+$. The following lemma is needed to estimate the operating cost of the LCP algorithm.

\begin{lemma}\label{lemma:lcp:cf:lcpcost}
	For all $\tau \in [t_i : t_{i+1}[ \in \mathcal{T}^+$,   
	\begin{equation}\label{eqn:lcp:cf:lcpcost:l}
	\hat{C}^L_\tau(x^\text{LCP}_\tau) + f_{\tau + 1}(x^\text{LCP}_{\tau +1}) \leq \hat{C}^L_{\tau + 1}(x^\text{LCP}_{\tau +1}).
	\end{equation}
	Analogously, for all $\tau \in [t_i : t_{i+1}[ \in \mathcal{T}^-$,   
	\begin{equation}\label{eqn:lcp:cf:lcpcost:u}
	\hat{C}^U_\tau(x^\text{LCP}_\tau) + f_{\tau + 1}(x^\text{LCP}_{\tau +1}) \leq \hat{C}^U_{\tau + 1}(x^\text{LCP}_{\tau +1}).
	\end{equation}
\end{lemma}

\begin{proof}
	First, we will prove equation~\eqref{eqn:lcp:cf:lcpcost:l}. 
	We differ between $x^\text{LCP}_{\tau} < x^\text{LCP}_{\tau +1}$ (case~1) and $x^\text{LCP}_{\tau} = x^\text{LCP}_{\tau +1}$ (case~2). Note that the case $x^\text{LCP}_{\tau} > x^\text{LCP}_{\tau +1}$ never occurs because $x^\text{LCP}_{t}$ is non-decreasing (Lemma~\ref{lemma:lcp:cf:monotony}).
	
	\textbf{Case 1:\ } If $x^\text{LCP}_{\tau} < x^\text{LCP}_{\tau +1}$, then $x^\text{LCP}_{\tau +1} = x^L_{\tau + 1}$ by the definition of the LCP algorithm. Furthermore,
	\begin{equation} \label{eqn:lcp:cf:lcpcost:c1}
	C^L_{\tau + 1}(X^L_{\tau +1}) = \hat{C}^L_{\tau}(x^L_{\tau +1, \tau}) + f_{\tau + 1}(x^L_{\tau + 1}) + \beta (x^L_{\tau + 1} - x^L_{\tau+1, \tau})^+ 
	\end{equation}
	holds by the definition of the upper bound. 
	If $x^L_{\tau + 1, \tau} \geq x^\text{LCP}_\tau$, then $\hat{C}^L_\tau(x^\text{LCP}_\tau) \leq \hat{C}^L_\tau(x^L_{\tau + 1, \tau})$ holds because $\hat{C}^L_\tau(x)$ is convex (Lemma~\ref{lemma:lcp:cf:convex}) with a minimum at $x^L_\tau \leq x^\text{LCP}_\tau$.
	If $x^L_{\tau + 1, \tau} < x^\text{LCP}_\tau$, then by using Lemma~\ref{lemma:lcp:cf:hatcluconvert} we get
	\begin{align*}
	\hat{C}^L_\tau(x^\text{LCP}_\tau) 
	&\alignstack{\text{\tiny L\ref{lemma:lcp:cf:hatcluconvert}}}{\,=\,} \hat{C}^U_\tau(x^\text{LCP}_\tau) +\beta x^\text{LCP}_\tau \\
	&\alignstack{\text{\tiny L\ref{lemma:lcp:cf:convex}}}{\,\leq\,} \hat{C}^U_\tau(x^L_{\tau + 1, \tau})  +\beta x^\text{LCP}_\tau \\
	&\alignstack{\text{\tiny L\ref{lemma:lcp:cf:hatcluconvert}}}{\,=\,} \hat{C}^L_\tau(x^L_{\tau + 1, \tau}) - \beta x^L_{\tau+1, \tau} + \beta x^\text{LCP}_\tau \\
	&\,\leq\, \hat{C}^L_\tau(x^L_{\tau + 1, \tau}) + \beta (x^L_{\tau + 1} - x^L_{\tau+1, \tau} )^+.
	\end{align*}
	The first inequality holds because $x^U_\tau$ minimizes $\hat{C}^U_\tau$ and $x^U_\tau \geq x^\text{LCP}_\tau > x^L_{\tau + 1, \tau}$. The last inequality uses $x^\text{LCP}_\tau < x^\text{LCP}_{\tau + 1} = x^L_{\tau + 1}$.
	
	By using this result in equation~\eqref{eqn:lcp:cf:lcpcost:c1}, we get
	\begin{equation*}
	C^L_{\tau + 1}(X^L_{\tau +1}) \geq \hat{C}^L_{\tau}(x^\text{LCP}_{\tau}) + f_{\tau + 1}(x^L_{\tau + 1}).
	\end{equation*}
	With $\hat{C}^L_{\tau + 1}(x^\text{LCP}_{\tau +1}) \geq  C^L_{\tau + 1}(X^L_{\tau +1})$ and $x^\text{LCP}_{\tau +1} = x^L_{\tau + 1}$, we get equation~\eqref{eqn:lcp:cf:lcpcost:l}.
	
	\textbf{Case 2:\ } $x^\text{LCP}_{\tau} = x^\text{LCP}_{\tau +1}$. 
	Let $\hat{X} = (\hat{x}_1, \dots, \hat{x}_\tau, x^\LCP_{\tau+1})$ be an optimal solution for $C^L_{\tau +1}$ that ends in the state $x^\LCP_{\tau+1}$, so $C^L_{\tau+1}(\hat{X}) = \hat{C}^L_{\tau+1}(x^\LCP_{\tau+1})$. It holds:
	\begin{equation} \label{eqn:lcp:cf:lcpcost:c3}
	\hat{C}^L_{\tau+1}(x^\LCP_{\tau +1}) = \hat{C}^L_\tau(\hat{x}_\tau) + f_{\tau+1}(x^\LCP_{\tau +1}) + \beta (x^\LCP_{\tau + 1} - \hat{x}_\tau)^+.
	\end{equation}
	If $\hat{x}_\tau \geq x^\LCP_\tau$, then $\hat{C}^L_\tau(\hat{x}_\tau) \geq \hat{C}^L_\tau(x^\LCP_\tau) $ holds because $\hat{C}^L_\tau(x)$ is convex (Lemma~\ref{lemma:lcp:cf:convex}) with a minimum at $x^L_\tau \leq x^\text{LCP}_\tau$. If $\hat{x}_\tau < x^\LCP_\tau$, then similar to case 1, we get
	\begin{align*}
	\hat{C}^L_\tau(x^\LCP_\tau) 
	&\alignstack{\text{\tiny L\ref{lemma:lcp:cf:hatcluconvert}}}{\,=\,} \hat{C}^U_\tau(x^\LCP_\tau) + \beta x^\LCP_\tau \\
	&\alignstack{\text{\tiny L\ref{lemma:lcp:cf:convex}}}{\,\leq\,} \hat{C}^U_\tau(\hat{x}_\tau) + \beta x^\LCP_\tau \\
	&\alignstack{\text{\tiny L\ref{lemma:lcp:cf:hatcluconvert}}}{\,=\,} \hat{C}^L_\tau(\hat{x}_\tau) - \beta \hat{x}_\tau + \beta x^\LCP_\tau \\
	&\,=\, \hat{C}^L_\tau(x^\LCP_\tau) + \beta (x^\LCP_{\tau+1} - \hat{x}_\tau)^+ .
	\end{align*}
	By using this in equation~\eqref{eqn:lcp:cf:lcpcost:c3}, we get
	\begin{equation*}
	\hat{C}^L_{\tau + 1}(x^\text{LCP}_{\tau +1}) \geq \hat{C}^L_\tau(x^\text{LCP}_\tau) + f_{\tau + 1}(x^\LCP_{\tau + 1})
	\end{equation*}
	which is exactly equation~\eqref{eqn:lcp:cf:lcpcost:l}.

	The proof of equation~\eqref{eqn:lcp:cf:lcpcost:u} works analogously by using the upper bound cost $\hat{C}^U_\tau$ and reversing the inequality signs. 

\end{proof}

We can use Lemma~\ref{lemma:lcp:cf:lcpcost} to estimate the operating cost of the LCP algorithm.

\begin{lemma} \label{lemma:lcp:cf:running}
	$R_T(X^\text{LCP}) \leq R_T(X^\ast) + \beta \sum_{t=1}^{T+1} |x^\ast_t - x^\ast_{t-1}|$
\end{lemma}

\begin{proof} 
	Consider the time interval $[t_i : t_{i+1}[ \in \mathcal{T}^+$. By adding the inequalities of Lemma~\ref{lemma:lcp:cf:lcpcost} for $\tau \in [t_i : t_{i+1}[$, we get
	\begin{equation*}
	\sum_{t = t_i}^{t_{i+1}-1} \hat{C}^L_t(x^\text{LCP}_t) + \sum_{t = t_i}^{t_{i+1}-1} f_{t + 1}(x^\text{LCP}_{t +1}) \leq \sum_{t = t_i}^{t_{i+1}-1} \hat{C}^L_{t + 1}(x^\text{LCP}_{t +1}).
	\end{equation*}
	Subtracting the first sum gives
	\begin{align*}
	\sum_{t = t_i}^{t_{i+1}-1} f_{t + 1}(x^\text{LCP}_{t +1}) &\leq \hat{C}^L_{t_{i + 1}}(x^\text{LCP}_{t_{i + 1}}) - \hat{C}^L_{t_{i}}(x^\text{LCP}_{t_{i}}) \\ &
	=  \hat{C}^L_{t_{i + 1}}(x^\ast_{t_{i + 1}}) - \hat{C}^L_{t_{i}}(x^\ast_{t_{i}}) \\
	&= \sum_{t = t_i}^{t_{i+1}-1} f_{t + 1}(x^\ast_{t +1}) + \beta(x^\ast_{t_{i + 1}} - x^\ast_{t_{i}}) . \numberthis	\label{eqn:lcp:cf:running:a}
	\end{align*}
	The first equality holds because $x^\text{LCP}_{t_{i}} = x^\ast_{t_{i}}$ and $x^\text{LCP}_{t_{i+1}} = x^\ast_{t_{i+1}}$. 
	
	Considering the time interval $[t_i : t_{i+1}[ \in \mathcal{T}^-$ yields to the following inequality:
	\begin{align*}
	\sum_{t = t_i}^{t_{i+1}-1} f_{t + 1}(x^\text{LCP}_{t +1}) &\leq \hat{C}^U_{t_{i + 1}}(x^\text{LCP}_{t_{i + 1}}) - \hat{C}^U_{t_{i}}(x^\text{LCP}_{t_{i}}) \\ &
	=  \hat{C}^U_{t_{i + 1}}(x^\ast_{t_{i + 1}}) - \hat{C}^U_{t_{i}}(x^\ast_{t_{i}}) \\
	&= \sum_{t = t_i}^{t_{i+1}-1} f_{t + 1}(x^\ast_{t +1}) + \beta(x^\ast_{t_{i}} - x^\ast_{t_{i + 1}}). \numberthis
	\label{eqn:lcp:cf:running:b}
	\end{align*}
	In both~\eqref{eqn:lcp:cf:running:a} and~\eqref{eqn:lcp:cf:running:b} the factor after $\beta$ is positive, so we can write
	\begin{align*}
	\sum_{t = t_i}^{t_{i+1}-1} f_{t + 1}(x^\text{LCP}_{t +1}) &\leq \sum_{t = t_i}^{t_{i+1}-1} f_{t + 1}(x^\ast_{t +1}) + \beta|x^\ast_{t_{i+1}} - x^\ast_{t_{i}}| \\
	&= \sum_{t = t_i}^{t_{i+1}-1} f_{t + 1}(x^\ast_{t +1}) + \beta \sum_{t=t_i}^{t_{i+1} - 1}|x^\ast_{t+1} - x^\ast_{t}|.
	\end{align*}
	By adding all intervals in $\mathcal{T}^+ \cup \mathcal{T}^-$, we get
	\begin{equation*}
	\sum_{t = 1}^{T+1} f_{t}(x^\text{LCP}_{t}) \leq \sum_{t = 1}^{T+1} f_{t}(x^\ast_{t}) + \beta\sum_{t=1}^{T+1}|x^\ast_{t} - x^\ast_{t-1}|.
	\end{equation*} 
	This is equivalent to $R_T(X^\text{LCP}) \leq R_T(X^\ast) + \beta \sum_{t=1}^{T+1} |x^\ast_t - x^\ast_{t-1}|$ since $f_{T+1}(x^\text{LCP}_{T+1}) = f_{T+1}(x^\ast_{T+1}) = 0$. 
\end{proof}

The term $\beta \sum_{t=1}^{T+1} |x^\ast_t - x^\ast_{t-1}|$ in Lemma~\ref{lemma:lcp:cf:running} is equal to twice the switching cost of the optimal schedule:

\begin{lemma} \label{lemma:lcp:cf:running:formula}
	$\beta \sum_{t=1}^{T+1} |x^\ast_t - x^\ast_{t-1}| = 2 \cdot S^L_T(X^\ast)$.
\end{lemma}

\begin{proof}
	Since we start at $x_0 = 0$ and end at $x_{T+1} = 0$, the number of servers that are powered up is equal to the number of servers that are powered down, i.e.,
	\begin{equation*}
	\sum_{t=1}^{T+1} (x^\ast_t - x^\ast_{t-1})^+ = \sum_{t=1}^{T+1} (x^\ast_{t-1} - x^\ast_{t})^+ .
	\end{equation*}
	Thus,
	\begin{align*}
	\sum_{t=1}^{T+1} |x^\ast_t - x^\ast_{t-1}| &= \sum_{t=1}^{T+1} (x^\ast_t - x^\ast_{t-1})^+ + \sum_{t=1}^{T+1} (x^\ast_{t-1} - x^\ast_{t})^+ \\&
	= 2 \cdot \sum_{t=1}^{T+1} (x^\ast_t - x^\ast_{t-1})^+  .
	\end{align*}
	Since $(x^\ast_{T+1} - x^\ast_T)^+ = 0$, we get
	\begin{equation*}
	2 \cdot S^L_T(X^\ast) 
	= 2 \cdot \beta \sum_{t=1}^{T} (x^\ast_t - x^\ast_{t-1})^+ 
	= \beta \sum_{t=1}^{T+1} |x^\ast_t - x^\ast_{t-1}|. \qedhere
	\end{equation*}
\end{proof}

Now, we are able to show that LCP is 3-competitive.

\begin{theorem}
	The LCP algorithm is 3-competitive.
\end{theorem}

\begin{proof}
	By using Lemmas~\ref{lemma:lcp:cf:switching}, \ref{lemma:lcp:cf:running} and~\ref{lemma:lcp:cf:running:formula}, we get
	\begin{align*}
	C^L_T(X^\text{LCP}) &\;=\; R_T(X^\text{LCP}) + S^L_T(X^\text{LCP}) \\
	&\alignstack{\substack{\text{\tiny L\ref{lemma:lcp:cf:switching}}\\ \text{\tiny L\ref{lemma:lcp:cf:running}}}}{\;\leq\;} R_T(X^\ast) + \beta \sum_{t=1}^{T+1} |x^\ast_t - x^\ast_{t-1}| + S^L_T(X^\ast) \\
	&\alignstack{\text{\tiny L\ref{lemma:lcp:cf:running:formula}}}{\;=\;} R_T(X^\ast) + 3 \cdot S^L_T(X^\ast) \\
	&\;\leq\; 3 \cdot C^L_T(X^\ast) . \qedhere
	\end{align*} 
\end{proof}

\section{A randomized offline algorithm}
\label{sec:random}
    \newcommand{\ceilstrong}[1]{\lceil #1 \rceil^\ast}
    \newcommand{\ceilstronglr}[1]{\left\lceil #1 \right\rceil^\ast}
  
  \newcommand{\barx}{\bar{x}}
  \newcommand{\barxt}{\bar{x}_t}
  \newcommand{\barxtLower}{\left\lfloor \bar{x}_t \right\rfloor}
  \newcommand{\barxtUpper}{\ceilstronglr{\bar{x}_t}}
  \newcommand{\barxPrev}{\bar{x}_{t-1}}
  \newcommand{\barxPrevLower}{\left\lfloor \bar{x}_{t-1} \right\rfloor}
  \newcommand{\barxPrevUpper}{\ceilstronglr{\bar{x}_{t-1}}}
  \newcommand{\barxBegin}{\bar{x}_{1}}
  \newcommand{\barxBeginLower}{\left\lfloor \bar{x}_{1} \right\rfloor}
  \newcommand{\barxBeginUpper}{\ceilstronglr{\bar{x}_{1}}}

In the last section, we presented a deterministic online algorithm for the dynamic data-center optimization problem that  achieves a competitive ratio of 3. This result can be improved by using randomization. In this section, we present a randomized online algorithm that is 2-competitive against an oblivious adversary. The basic idea is to use the algorithm of Bansal et al.~\cite{Bansal2015} to get a 2-competitive schedule for the continuous extension of the given problem instance. Then, we round the particular states of the schedule randomly to achieve an integral schedule. The expected total cost of the resulting schedule is at most twice as much as the cost of an oblivious adversary.

\subsection{Algorithm}
\label{sec:random:algo}

We consider the continuous extension $\bar{\mathcal{P}}$ of the original problem instance $\mathcal{P}$ as introduced in Section~\ref{sec:poly:correct} (see equation~\eqref{eqn:poly:correct:extension}). 
For this continuous problem, Bansal et al. \cite{Bansal2015} give a 2-competitive (deterministic) online algorithm. 
Let $\bar{X} = (\bar{x}_1, \dots, \bar{x}_T)$ be the schedule calculated by the algorithm of Bansal et al. We will convert this solution to an integral schedule $X = (x_1, \dots x_T)$.

To describe our algorithm we use the following notation.
In contrast to the usual definition of~$\lceil \cdot \rceil$, we define $\ceilstrong{x} \coloneqq \min \{n \in \mathbb{Z} \mid n > x\}$, i.e., for an integer $n' \in \mathbb{Z}$, we have $\ceilstrong{n'} = n' + 1$. Note that the definition remains the same for non-integral arguments.
The definition of $\lfloor \cdot \rfloor$
does not change, so the identity $\ceilstrong{x} = \lfloor x \rfloor + 1$ is always fulfilled.
Let $[x]_a^b \coloneqq \max\{a, \min\{b, x\}\}$ be the projection of $x$ into the interval $[a,b]$, let $\text{frac}(x) \coloneqq x - \lfloor x \rfloor$ denote the fractional part of~$x$ and let $\bar{x}'_{t-1} \coloneqq [\bar{x}_{t-1}]_{\lfloor \bar{x}_t \rfloor}^{\ceilstrong{\bar{x}_t}}$ be the projection of the previous state into the interval of the current state.

We distinguish between time slots where the  number of active servers increases and those where the number of active servers decreases. In the first case, we have $\bar{x}_{t-1} \leq \bar{x}_t$. If $x_{t-1}$ is already in the upper state $\ceilstrong{\bar{x}_t}$, we keep this state, so $x_t = \ceilstrong{\bar{x}_t}$. Otherwise, with probability $p^\uparrow_t \coloneqq \frac{\bar{x}_t - \bar{x}'_{t-1}}{1 - \text{frac}(\bar{x}'_{t-1})}$, we set $x_t$ to the upper state $\barxtUpper$ and with probability $1 - p^\uparrow_t$, we keep the lower state $\barxtLower$. 
The other case (i.e., $\bar{x}_{t-1} > \bar{x}_t$) is handled symmetrically. If $x_{t-1} = \lfloor \bar{x}_t \rfloor$, then we keep the state, i.e., $x_t = \lfloor \bar{x}_t \rfloor$, and otherwise with probability $p^\downarrow_t \coloneqq \frac{\bar{x}'_{t-1} - \bar{x}_t }{\text{frac}(\bar{x}'_{t-1})}$, we set $x_t$ to the lower state $\barxtLower$  and with probability $1 - p^\downarrow_t$, we keep the upper state $\barxtUpper$.
Obviously, $X$ is an integral schedule.

\subsection{Analysis}
\label{sec:random:analysis}

To show that the algorithm described above is 2-competitive against an oblivious adversary, we have to prove that the expected cost of our algorithm is at most twice the cost of an optimal offline solution. Let $C^Q(Y)$ denote the total cost of the schedule $Y$ for the problem instance $Q$, so we want to prove that 
\begin{equation} \label{eqn:random:analysis:competitive}
\mathbb{E}[C^\mathcal{P}(X)] \leq 2 \cdot C^\mathcal{P}(X^\ast) .
\end{equation}

Let $\bar{X}^\ast$ be an optimal offline solution for $\bar{P}$. By Lemma~\ref{lemma:poly:correct:rounding}, we know that this solution can be easily rounded to an integral solution $X^\ast$ without increasing the cost, i.e.,
\begin{equation} \label{eqn:random:analysis:cxopt}
C^{\bar{\mathcal{P}}}(\bar{X}^\ast) = C^\mathcal{P}(X^\ast) .
\end{equation}
Furthermore, we know that the algorithm of Bansal et al.\ is 2-competitive for the continuous setting, so we have
\begin{equation} \label{eqn:random:analysis:bansal}
C^{\bar{\mathcal{P}}}(\bar{X}) \leq 2 \cdot C^{\bar{\mathcal{P}}}(\bar{X}^\ast) .
\end{equation}
Thus, it is sufficient to show that $\mathbb{E}[C^\mathcal{P}(X)] = C^{\bar{\mathcal{P}}}(\bar{X})$.

The following lemma describes the probability that a value $\bar{x}_t$ is rounded up.

\begin{lemma} \label{lemma:random:analysis:prob}
	The probability that $x_t$ equals the upper state  $\ceilstrong{\bar{x}_t}$ of the fractional schedule is $\fpart(\barx_t)$. Formally, $\Pr[x_t = \ceilstrong{\bar{x}_t}] = \fpart(\barxt)$. 
\end{lemma}

\begin{proof}
	We prove the lemma by induction. It is clear that $x_t$ is either $\barxtLower$ or $\barxtUpper$. For $t=1$ the probability for $x_1 = \barxBeginUpper$ is $ p^\uparrow_1 = \frac{\barx_1 - \barx'_0}{1 - \fpart(\barx'_0)} = \barx_1 - \barxBeginLower = \fpart(\barx_1)$, because $\bar{x}_0 = 0$ and therefore $\bar{x}'_{0} = \barxBeginLower$.
	
	Assume that the claim of Lemma~\ref{lemma:random:analysis:prob} holds for $t-1$, so $\Pr[x_{t-1} = \ceilstrong{\bar{x}_{t-1}}] = \fpart(\barxPrev)$. We differ between increasing time slots where $\bar{x}_{t-1} \leq \bar{x}_t$ holds (case 1) and decreasing time slots where $\bar{x}_{t-1} > \bar{x}_t$ (case 2). In case 1, the probability $\Pr[x_t = \ceilstrong{\bar{x}_t}]$ can be written as
	\begin{equation} \label{eqn:random:analysis:prob:pr}
	\begin{aligned}
	\Pr[x_t = \barxtUpper] 
	&= \Pr[x_t = \barxtUpper \mid x_{t-1} = \barxtUpper]
	\cdot \Pr[x_{t-1} = \barxtUpper] \\
	&+ \Pr[x_t = \barxtUpper \mid x_{t-1} \leq \barxtLower] 
	\cdot \Pr[x_{t-1} \leq \barxtLower] .
	\end{aligned}
	\end{equation}
	Note that $x_{t-1}$ is integral and cannot be greater than $\barxtUpper$. If $\barxPrev \leq \barxtLower$, then $\Pr[x_{t-1} = \barxtUpper] = 0$ and $\Pr[x_{t-1} \leq \barxtLower] = 1$, so similar to the base case we get
	\begin{align*}
	\Pr[x_t = \barxtUpper]  &= \Pr[x_t = \barxtUpper \mid x_{t-1} \leq \barxtLower]  \\
	&= p^\uparrow_t \\
	&= \barxt - \barxtLower \\
	&= \fpart (\barxt) .
	\end{align*}
	If $\barxPrev > \barxtLower$, then by our induction hypothesis $\Pr[x_{t-1} = \barxtUpper] = \fpart(\barxPrev)$ and $\Pr[x_{t-1} \leq \barxtLower]  =  \Pr[x_{t-1} = \barxtLower] = 1 - \fpart(\barxPrev)$. By the definition of our algorithm, $\Pr[x_t = \barxtUpper \mid x_{t-1} = \barxtUpper] = 1$, because we keep the state if we are already in the upper state. Furthermore, we get 
	\begin{align*}
	\Pr[x_t = \barxtUpper \mid x_{t-1} \leq \barxtLower] &= \Pr[x_t = \barxtUpper \mid x_{t-1} = \barxtLower] \\
	&= p^\uparrow_t \\
	&= \frac{\barxt - \barx'_{t-1}}{1 - \fpart(\barx'_{t-1})} \\
	&= \frac{\barxt - \barxPrev}{1 - \fpart(\barxPrev)} .
	\end{align*}
	
	By inserting this results into equation~\eqref{eqn:random:analysis:prob:pr}, we get
	\begin{align*}
	\Pr[x_t = \barxtUpper] &= 1 \cdot \fpart(\barxPrev) + \frac{\barxt - \barxPrev}{1 - \fpart(\barxPrev)} \cdot ( 1 - \fpart(\barxPrev)) \\
	&= \barxPrev - \barxPrevLower + \barxt - \barxPrev \\
	&= \barxt - \barxtLower \\
	&= \fpart (\barxt) .
	\end{align*}
	The third equation holds, because $\barxtLower < \barxPrev < \barxt$, so $\barxtLower = \barxPrevLower$.
	
	The second case $\barxPrev > \barxt$ works analogously. \qedhere

\end{proof}

The proof of $\mathbb{E}[C^\mathcal{P}(X)] = C^{\bar{\mathcal{P}}}(\bar{X})$ is divided into two parts. First, in the following lemma, we will show that the expected operating cost of our algorithm is equal to the operating cost of the algorithm of Bansal et al.\ for the continuous version of the problem instance. Then, in Lemma~\ref{lemma:random:analysis:switching}, we will show the same for the switching cost. Let $R^Q(Y)$ and $S^Q(Y)$ denote the operating and switching cost of the schedule $Y$ for the problem instance $Q$, respectively.

\begin{lemma}\label{lemma:random:analysis:operating}
	The expected operating cost of our algorithm is equal to the operating cost of the algorithm of Bansal et al.\ for the continuous extension of the problem instance, i.e., $\mathbb{E}[R^\mathcal{P}(X)] = R^{\bar{\mathcal{P}}}(\bar{X})$. 
\end{lemma}

\begin{proof}
	The expected operating cost of our algorithm can be written as
	\begin{align*}
	\mathbb{E}[R^\mathcal{P}(X)] = \sum_{t=1}^{T}  \big( \Pr[x_t = \barxtLower] \cdot f_t(\barxtLower) 
	+ \Pr[x_t = \barxtUpper] \cdot f_t(\barxtUpper) \big) .
	\end{align*}
	By using Lemma~\ref{lemma:random:analysis:prob}, we get
	\begin{align*}
	\mathbb{E}[R^\mathcal{P}(X)] 
	&\alignstack{\text{\tiny L\ref{lemma:random:analysis:prob}}}{\,=\,} \sum_{t=1}^{T} \Big( \big(1 - \fpart (\barxt)\big) \cdot f_t(\barxtLower) + \fpart(\barxt)\cdot f_t(\barxtUpper) \Big) \\
	&\alignstack{\eqref{eqn:poly:correct:extension}}{\,=\,} \sum_{t=1}^{T} \bar{f}_t(\barxt) \\
	&\,=\, R^{\bar{\mathcal{P}}}(\bar{X}) .
	\end{align*}
	The second equality follows from the definition of the continuous extension of the operating cost functions (since $1 - \fpart (x) = \ceilstrong{x} - x$).
\end{proof}

Now, we will determine the expected switching cost of our algorithm for each time slot. 

\begin{lemma} \label{lemma:random:analysis:switching}
	The expected switching cost of our algorithm is equal to the switching cost of the continuous schedule, i.e., $\mathbb{E}[S^\mathcal{P}(X)] = S^{\bar{\mathcal{P}}}(\bar{X})$.
\end{lemma}

\begin{proof}
	We calculate the switching cost for each time slot separately.	We distinguish between the cases (1) $\barxPrev < \barxtLower$, (2) $\barxPrev \in [\barxtLower, \barxt]$ and (3) $\barxPrev > \barxt$. The last case is trivial, because no servers are powered up, so there is no switching cost.
		
	In case 1, we can separate the expected switching cost into three parts: The expected cost for powering up from $\barxPrev$ to $\barxPrevUpper$, the cost from $\barxPrevUpper$ to $\barxtLower$ (can be zero) and the expected cost from $\barxtLower$ to $\barxt$. The expected number of servers powered up is
	\begin{align*}
	\phantom{{}={}} \mathbb{E}[(x_t - x_{t-1})^+] 
	&\,=\, \Pr[x_{t-1} = \barxPrevLower] + \left(\barxtLower - \barxPrevUpper \right) + \Pr[x_t = \barxtUpper] \\
	&\alignstack{\text{\tiny L\ref{lemma:random:analysis:prob}}}{\,=\,} 1 - \fpart (\barxPrev) + \barxtLower - \barxPrevUpper + \fpart(\barxt)  \\
	&\,=\, (\barxt - \barxPrev)^+ .
	\end{align*}
	The second equation uses Lemma~\ref{lemma:random:analysis:prob} and the third equation follows from the definition of $\fpart$ and the identity $\ceilstrong{x} = \lfloor x \rfloor + 1$. 
	
	For case 2, let $l \coloneqq \barxtLower$ be the lower and $u \coloneqq \barxtUpper$ the upper state of the fractional state $\barxt$. Since $\barxPrev \in [\barxtLower, \barxt]$ holds, we only switch the state, if we are in the lower state during time slot $t-1$ and in the upper state during time slot $t$. Thus, the expected number of servers powered up is
	\begin{align*}
	\mathbb{E}[(x_t - x_{t-1})^+] &= \Pr [x_{t-1} = l] \cdot \Pr[x_t = u \mid x_{t-1} = l]  .
	\end{align*}
	By Lemma~\ref{lemma:random:analysis:prob}, we know $\Pr [x_{t-1} = l] = 1 - \fpart (\barxPrev)$. Furthermore, by the definition of our algorithm, we have $\Pr[x_t = u \mid x_{t-1} = l] = p^\uparrow_t$, so we get
	\begin{align*}
	\mathbb{E}[(x_t - x_{t-1})^+] &= (1 - \fpart(\barxPrev)) \cdot \frac{\barxt - \barxPrev}{1 - \fpart(\barxPrev)} \\
	&=  (\barxt - \barxPrev)^+ .
	\end{align*}
	
	So for all cases, $\mathbb{E}[\beta (x_t - x_{t-1})^+] =  \beta (\barxt - \barxPrev)^+$ holds. By summing over all time slots, we get $\mathbb{E}[S^\mathcal{P}(X)] = S^{\bar{\mathcal{P}}}(\bar{X})$.
\end{proof}

\begin{theorem}
	The algorithm described in Section~\ref{sec:random:algo} is 2-competitive against an oblivious adversary.
\end{theorem}

\begin{proof}
	We have to show that $\mathbb{E}[C^\mathcal{P}(X)] \leq 2 \cdot C^{\mathcal{P}}(X^\ast)$. By using Lemmas~\ref{lemma:random:analysis:operating} and~\ref{lemma:random:analysis:switching} as well as equations~\eqref{eqn:random:analysis:cxopt} and~\eqref {eqn:random:analysis:bansal}, we get
	\begin{align*}
	\mathbb{E}[C^\mathcal{P}(X)] &= \mathbb{E}[R^\mathcal{P}(X)] + \mathbb{E}[S^\mathcal{P}(X)] \\
	&\alignstack{\substack{\text{\tiny L\ref{lemma:random:analysis:operating}}\\ \text{\tiny L\ref{lemma:random:analysis:switching}}}}{=} R^{\bar{\mathcal{P}}}(\bar{X}) + S^{\bar{\mathcal{P}}}(\bar{X}) \\
	&= C^{\bar{\mathcal{P}}}(\bar{X}) \\
	&\alignstack{\eqref{eqn:random:analysis:bansal}}{\leq} 2 \cdot C^{\bar{\mathcal{P}}}(\bar{X}^\ast) \\
	&\alignstack{\eqref{eqn:random:analysis:cxopt}}{=} 2 \cdot C^{\mathcal{P}}(X^\ast) . \qedhere
	\end{align*}
\end{proof}

\section{Lower bounds}
\label{sec:lower}
In this section, we will show lower bounds for both the discrete and continuous data-center optimization problem. First, in Section~\ref{sec:lower:deterministic} we prove that there is no deterministic online algorithm that achieves a competitive ratio better than 3 for the discrete setting. This lower bound demonstrates that the LCP algorithm analyzed in Section~\ref{sec:lcp} is optimal. Afterwards, we show that this lower bound also holds for the restricted model introduced by Lin et al.~\cite{LinWierman2011infocom} where the operating cost functions are more restricted than in the general model investigated in Section~\ref{sec:lcp}. 
A formal definition of the restricted model is given in Section~\ref{sec:lower:deterministic:lin}. Moreover, we give a lower bound for the continuous setting and show that this lower bound holds again for the restricted model (see Section~\ref{sec:lower:cont}). 
A lower bound of~2 for the general continuous setting was independently shown by Antoniadis~et~al.~\cite{Antoniadis2017}. 
Based on our result for the continuous setting, we show in Section~\ref{sec:lower:random} that there is no randomized algorithm that achieves a competitive ratio better than 2 in the discrete setting. Again, this lower bound still holds for the restricted model.
Finally, in Section~\ref{sec:lower:prediction} we extend our lower bounds to the scenario that an online algorithm has a finite prediction window.

To simplify the analysis, the switching costs are paid for both powering up and powering down. At the end of the workload all servers have to be powered down. This ensures that the total cost remains the same. We will set $\beta = 2$, so changing a server's state will cost $\beta / 2 = 1$. Thus, the cost of a schedule is defined by
\begin{equation*}
C(X) \coloneqq \sum_{t=1}^{T} f_t(x_t) + \sum_{t=1}^{T+1} |x_{t}  - x_{t-1}|
\end{equation*}
with $x_0 \coloneqq x_{T+1} \coloneqq 0$. 

\subsection{Discrete setting, deterministic algorithms}
\label{sec:lower:deterministic}

First, we analyze the discrete setting for deterministic online algorithms. We begin with the general model and afterwards show in Section~\ref{sec:lower:deterministic:lin} how our construction can be adapted to the restricted model.

\subsubsection{General model}
\label{sec:lower:deterministic:general}

\begin{theorem} \label{theo:lower:disc:three}
	There is no deterministic online algorithm that achieves a competitive ratio of $c < 3$ for the discrete data-center optimization problem. 
\end{theorem}

\begin{proof}
	Assume that there is a deterministic algorithm $\mathcal{A}$ that is $(3-\delta)$-competitive with $\delta > 0$. The adversary will use the functions $\varphi_0(x) = \epsilon|x|$ and $\varphi_1(x) = \epsilon|x-1|$ with $\epsilon \rightarrow 0$, so we only need the states $0$ and $1$, there is no benefit to use other states. If $\mathcal{A}$ is in state $0$ or $1$, the adversary will send $\varphi_1$ or $\varphi_0$, respectively.
	
	Let $S$ be the number of time slots where algorithm $\mathcal{A}$ changes the state of a server, i.e., $S$ is the switching cost of $\mathcal{A}$. Let $T$ be length of the whole workload (we will define $T$ later), so for $T-S$ time slots the operating costs of $\mathcal{A}$ are $\epsilon$. Thus, the total cost of $\mathcal{A}$ is
	\begin{equation*}
	C(\mathcal{A}) = (T-S)\epsilon + S .
	\end{equation*}
	
	The cost of the optimal offline schedule can be bounded by the minimum of the following two strategies. The first strategy is to stay at one state for the whole workload. If $\varphi_0$ is sent more often than $\varphi_1$, then this is state 0, else it is state 1. The operating cost is at most $T\epsilon /2$, the switching cost is at most $2$, because if we use state 1, we have to switch the state at the beginning and end of the workload. The second strategy is to always switch the state, such that there is no operating cost. In this case the switching cost is at most $S + 2$, because we switch the state after each time $\mathcal{A}$ switches its state as well as possibly at the beginning and the end of the workload.
	Thus, the cost of the optimal offline schedule is 
	\begin{equation} \label{eqn:lower:disc:three:opt}
	C(X^\ast) \leq \min (T\epsilon/2 + 2, S + 2) .
	\end{equation}
	
	We want to find a lower bound for the competitive ratio $\frac{C(\mathcal{A})}{C(X^\ast)}$. We distinguish between $S \geq T\epsilon / 2$ (case 1) and $S < T\epsilon / 2$ (case 2).
	
	In \textbf{case 1} the competitive ratio of $\mathcal{A}$ is 
	\begin{alignat*}{2}
	&\frac{C(\mathcal{A})}{C(X^\ast)} 
	&{}\alignstack{\eqref{eqn:lower:disc:three:opt}}{\,\geq\,}{}& \frac{(T-S)\epsilon + S}{T\epsilon/2 + 2} 
	{}={} 2 + \frac{S(1-\epsilon) - 4}{T\epsilon / 2 + 2} \\
	&&{}\,\geq\,{}& 2 + \frac{(T\epsilon / 2)(1 - \epsilon) - 4}{T\epsilon / 2 + 2} 
	{}={} 2 + (1-\epsilon) - \frac{2(1-\epsilon) + 4}{T\epsilon / 2 + 2} .
	\end{alignat*}
	The last inequality uses $S \geq T\epsilon / 2$ that holds for case~1. By setting $T \geq \frac{1}{\epsilon^2}$, we get $\lim_{\epsilon \rightarrow 0} T \epsilon = \infty$ and thus
	$\lim_{\epsilon \rightarrow 0} \frac{C(\mathcal{A})}{C(X^\ast)} = 3$.
	
	In \textbf{case 2}, we get
	\begin{alignat*}{2}
	&\frac{C(\mathcal{A})}{C(X^\ast)} 
	&{}\alignstack{\eqref{eqn:lower:disc:three:opt}}{\geq}{}& \frac{(T-S)\epsilon + S}{S + 2}
	{}={} (1-\epsilon) + \frac{T\epsilon - 2(1 - \epsilon)}{S+2} \\
	&&{}\geq{}& (1-\epsilon) + \frac{T\epsilon - 2(1-\epsilon)}{T\epsilon/2+2} 
	{}={} 3 - \epsilon - \frac{2(1- \epsilon) + 4}{T\epsilon/2+2} .
	\end{alignat*}
	
	The last inequality uses $S > T\epsilon / 2$. Again, we set $T \geq \frac{1}{\epsilon^2}$ and get
	$\lim_{\epsilon \rightarrow 0} \frac{C(\mathcal{A})}{C(X^\ast)} = 3$.
	
	Therefore, there is no algorithm with a competitive ratio that is less than 3. We can set $T$ to an arbitrarily large value, so the total cost of $\mathcal{A}$ converges to infinity. 
\end{proof}

\subsubsection{Restricted model}
\label{sec:lower:deterministic:lin}

Lin et. al. \cite{LinWierman2011infocom} introduced a more restricted setting as described by equation~\eqref{eqn:model:lin}. In this section, we show that the lower bound of 3 still holds for this model.
The essential differences of the restricted model to the general model are: (1) There is only one convex function for the whole problem instance and (2) there is the additional condition that $x_t \geq \lambda_t$. The different definition of the switching cost does not influence the total cost as already mentioned at the 
beginning of Section~\ref{sec:lower}.

\begin{theorem} \label{theo:lower:disc:lin}
	There is no deterministic online algorithm for the discrete setting of the restricted model with a competitive ratio of $c < 3$. 
\end{theorem}

\begin{proof}
	The general model (examined in the previous sections) is denoted by $\mathcal{G}$ and the restricted model by Lin et al.\ is denoted by $\mathcal{L}$.  The states of the model $\mathcal{X} \in \{\mathcal{G}, \mathcal{L}\}$ are indicated by $x^\mathcal{X}_t$. We will use the same idea as in the proof of Theorem~\ref{theo:lower:disc:three}, but we have to modify it such that it fits for the restricted model.
	
	We use 2 servers, so the states are $x^\mathcal{L}_t \in \{0, 1, 2\}$. Instead of switching between the states $0$ and $1$ in $\mathcal{G}$, we will switch between $1$ and $2$ in $\mathcal{L}$, so for $t \in [T]$ we have $x^\mathcal{L}_t = x^\mathcal{G}_t + 1$. In $\mathcal{L}$ the state $0$ is only used at the beginning ($t=0$) of the workload. This leads to additional switching costs of 1 for both the optimal offline solution and the online algorithm. However, for a sufficiently long workload the total cost converges to infinity, so the constant extra cost does not influence the competitive ratio.
	
	We will apply the same adversary strategy used in the proof of Theorem~\ref{theo:lower:disc:three}. Let $f(z) \coloneqq  \epsilon|1 - 2z|$ with $\epsilon \rightarrow 0$, let $\beta = 2$. If the adversary in $\mathcal{G}$ sends $\varphi_0(x) = \epsilon|x|$ as function, then we will use $\lambda_t = l_0 \coloneqq 0.5$ which leads to operating cost of
	\begin{equation*}
	x^\mathcal{L}_t f \left(l_0/x^\mathcal{L}_t \right) =x^\mathcal{L}_t \cdot \epsilon \left|1 - \frac{1}{x^\mathcal{L}_t}\right| = \epsilon \left|x^\mathcal{L}_t - 1 \right| = \epsilon \left|x^\mathcal{G}_t \right| .
	\end{equation*}
	If the adversary sends $\varphi_1(x) = \epsilon |1 - x|$, then we will use $\lambda_t = l_1 \coloneqq 1$ which leads to operating cost of
	\begin{equation*}
	x^\mathcal{L}_t f \left(l_1/x^\mathcal{L}_t \right) = x^\mathcal{L}_t \cdot \epsilon \left|1  - \frac{2}{x^\mathcal{L}_t} \right| =\epsilon \left|x^\mathcal{L}_t - 2\right| = \epsilon \left| x^\mathcal{G}_t - 1\right| = \epsilon \left|1 - x^\mathcal{G}_t \right| .
	\end{equation*}
	Thus, the difference (1) between both models is solved.
	
	For $t \geq 1$ it is not allowed to use the state $x^\mathcal{L}_t = 0$, because both $l_0$ and $l_1$ are greater than 0. For $x^\mathcal{L}_t \in \{1,2\}$ the inequality $x_t \geq \lambda_t$ is always fulfilled, so the difference (2) is solved too.
\end{proof}

\subsection{Continuous setting}
\label{sec:lower:cont}

In this section, we determine a lower bound for the continuous data-center optimization problem. Again, we begin with the general model and analyze the restricted model afterwards in Section~\ref{sec:lower:cont:lin}.

\subsubsection{General model}
\label{sec:lower:cont:general}

\begin{theorem} \label{theo:lower:cont:two}
	There is no deterministic online algorithm for the continuous data-center optimization problem that achieves a competitive ratio that is less than 2.
\end{theorem}

The proof consists of two parts. 
First we will construct an algorithm $\mathcal{B}$ whose competitive ratio is greater than $2-\delta$ for an arbitrary small $\delta$.  
Then we will show that the competitive ratio of any deterministic algorithm that differs from $\mathcal{B}$ is  greater than 2. 

For the first part we use an adversary that uses $\varphi_0(x) = \epsilon|x|$ and $\varphi_1(x) = \epsilon|1-x|$ as functions where $\epsilon \rightarrow 0$.  Let $b_t$ be the state of $\mathcal{B}$ at time $t$. If the function $\varphi_0$ arrives, then the next state $b_{t+1}$ is $\max \{b_t - \epsilon/2, 0\}$. If $\varphi_1$ arrives, the next state is $b_{t+1} \coloneqq \min\{b_t + \epsilon / 2, 1\}$, so formally
\begin{equation*}
b_{t+1} \coloneqq \begin{cases}
\max \{b_t - \epsilon / 2, 0\} & \text{if $f_t = \varphi_0$} \\
\min \{b_t + \epsilon / 2, 1\} & \text{if $f_t = \varphi_1$} .
\end{cases}
\end{equation*}
The algorithm starts at $b_0 = 0$, so $b_t \in [0,1]$ is fulfilled for all $t$. Note that algorithm $\mathcal{B}$ is equivalent to the algorithm of Bansal et al.~\cite{Bansal2015} for the special case of $\varphi_0$ and $\varphi_1$ functions. 
To simplify the calculations we assume that $\epsilon^{-1}$ is an integer, so the algorithm $\mathcal{B}$ is able to use $2\epsilon + 1$ different states. Note that $\epsilon$ can be chosen arbitrarily, so this is not a restriction.

\begin{lemma} \label{lemma:lower:cont:blowerbound}
	The competitive ratio of $\mathcal{B}$ is at least $2-\delta$ for an arbitrary small $\delta > 0$, so $C(\mathcal{B}) \geq (2-\delta) \cdot C(X^\ast)$. 
\end{lemma}

\begin{proof}
	Let $N_0(t)$ be the number of time slots $t' \leq t$ where $f_{t'} = \varphi_0$ and let $N_1(t)$ be the number of time slots where $f_{t'} = \varphi_1$. Note that $N_0(t) + N_1(t) = t$ for all~$t$. 
	
	Let $T > 0$ denote the first time slot, when $b_t$ reaches $0$ (case 1) or $1$ (case 2). Case 3 handles the case that there is no such time slot. 
	
	\textbf{Case 1:} If $b_T = 0$, then $N_0(T) = N_1(T)$. In each time step the algorithm $\mathcal{B}$ either increases or decreases its state by $\epsilon/2$, so the switching cost during the whole workload is $T\epsilon/2$. For each time slot $t$ with $f_t = \varphi_1$ there is exactly one unique corresponding time slot $t'$ with $f_{t'} = \varphi_0$ and $b_{t'} = b_{t} - \epsilon/2$. The operating costs for both time slots are
	\begin{align*}
	f_t(b_t) + f_{t'}(b_{t'})
	&= \epsilon |1-b_t| + \epsilon |b_{t'}|\\
	&= \epsilon (1- \epsilon/2)	.
	\end{align*}
	As $T$ must be even, the operating cost is $T/2 \cdot \epsilon (1- \epsilon/2)$. 
	
	To estimate the cost of an optimal solution, we consider the schedule that stays at $x = 0$ for the whole time. For this schedule, there is no switching cost and the operating cost is $\epsilon N_1(T) = \epsilon T / 2$. 
	Therefore, the cost of an optimal schedule is at most $C(X^\ast) \leq \epsilon T/2$. 
	
	
	The competitive ratio of $\mathcal{B}$ is 
	\begin{equation*}
	\frac{C(\mathcal{B})}{C(X^\ast)} \geq \frac{T\epsilon/2 + T/2 \cdot \epsilon (1- \epsilon/2)}{\epsilon T/2} = 2 - \epsilon/2 .
	\end{equation*}
	
	\textbf{Case 2:} If $b_T = 1$, then $N_1(T) = N_0(T) + 2/\epsilon$. The switching cost during the time interval is again $T\epsilon/2 = N_0(T)\epsilon+1$. For each time slot $t$ with  $f_t = \varphi_1$ there exists either one corresponding time slot~$t'$ with $f_{t'} = \varphi_0$ and $b_{t'} = b_{t} - \epsilon/2$ or for all $t' > t$ we have $b_{t'} \geq b_t$ . 
	Analogously to case 1, the operating costs of the corresponding pairs are $N_0(T) \cdot \epsilon (1-\epsilon/2)$. For each state $x \in \{\epsilon/2, 2\epsilon/2, \dots , 1\}$ there is exactly one time slot where $b_t$ has no corresponding time slot $t'$. This leads to operating costs of 
	\begin{align*}
	\sum_{i = 1}^{2/\epsilon} \varphi_1(i \cdot \epsilon/2) 
	&= \sum_{i = 1}^{2/\epsilon} \epsilon \cdot |1 - i \cdot \epsilon/2| \\
	&= \sum_{i' = 0}^{2/\epsilon - 1} \epsilon \cdot |i' \cdot \epsilon/2| \\
	&= \epsilon^2 / 2 \sum_{i' = 0}^{2/\epsilon - 1} i' \\
	&= \epsilon^2 / 2 \frac{(2/\epsilon-1) \cdot (2/\epsilon)}{2} \\
	&= 1 - \epsilon/2 .
	\end{align*}
	
	A schedule that switches directly to $x = 1$ at the beginning of the workload and stays there has a switching cost of $1$ and a operating cost of $\epsilon N_0(T)$. Therefore, the total cost of an optimal solution is at most $C(X^\ast) \leq 1 + \epsilon N_0(T)$, 
	%
	so the competitive ratio of $\mathcal{B}$ is:
	\begin{equation*}
	\frac{C(\mathcal{B})}{C(X^\ast)} \geq \frac{(\epsilon N_0(T)+1) + N_0(T) \cdot \epsilon (1-\epsilon/2) +  1 - \epsilon/2}{1 + \epsilon N_0(T)} = 2 - \epsilon/2 .
	\end{equation*}
	
	\textbf{Case 3:} It is possible that $b_t$ never reaches $0$ or $1$, for example if the adversary sends $\varphi_0$ and $\varphi_1$ alternately. Let $T$ be an arbitrary time slot.  The state of $\mathcal{B}$ is $b_T$, so $N_1(T) = N_0(T) + 2b_T/\epsilon$ holds. The switching cost of $\mathcal{B}$ is again $T\epsilon/2 = \epsilon N_1(T) - b_T$. Similar to case 2, there are corresponding pairs with operating costs of  $N_0(T) \cdot \epsilon (1-\epsilon/2) = (\epsilon N_1(T) - 2b_T)(1-\epsilon/2)$. For a lower bound, it is not necessary to consider the operating cost of the time slots without a corresponding partner, so $C(\mathcal{B}) \geq \epsilon N_1(T) - b_T + (\epsilon N_1(T) - 2b_T) (1-\epsilon/2)$. 
	
	A schedule that stays at $x = 0$ for the whole time has a total cost of $\epsilon N_1(T)$, so $C(X^\ast) \leq \epsilon N_1(T)$. 
	Thus, the competitive ratio is
	\begin{align*}
	\frac{C(\mathcal{B})}{C(X^\ast)} &\geq \frac{\epsilon N_1(T) - b_T + (\epsilon N_1(T) - 2b_T) (1-\epsilon/2)}{\epsilon N_1(T)} \\
	&= 2 - \epsilon/2 - \frac{b_T(3 - \epsilon/2)}{\epsilon N_1(T)} \\
	&\geq 2 - \epsilon/2 - \frac{6}{T} .
	\end{align*}
	The last inequality holds because $b_T < 1$ and $N_1(T) > T/2$. We set $T \geq 12 / \epsilon$ and get $\frac{C(\mathcal{B})}{C(X^\ast)} \geq 2 - \epsilon$. 
	
	We set $\epsilon \coloneqq \delta$, so the inequality $C(\mathcal{B}) \geq (2-\delta) \cdot C(X^\ast)$ is satisfied in all cases.
\end{proof}

Instead of ending at the states $0$ or $1$, we can extend the workload such that the competitive ratio is still at least 2, but the total cost of $\mathcal{B}$ converges to infinity. This leads to the following lemma which contains a stronger definition of the competitive ratio: 

\begin{lemma} \label{lemma:lower:cont:blowerbound:plus}
	For all $\delta > 0$ and $\alpha \geq 0$, there exists a workload such that 
	\begin{equation*}
	C(\mathcal{B}) \geq (2-\delta) \cdot C(X^\ast) + \alpha
	\end{equation*}  
	is fulfilled.
\end{lemma}

\begin{proof}
	We prove the lemma by extending the construction used in the proof of Lemma~\ref{lemma:lower:cont:blowerbound}.
	If $\mathcal{B}$ reaches the state $0$ (case 1), the situation is the same as at the beginning (i.e., $t=0$). We can repeat the argumentation of the proof by sending $\varphi_1$ as next function, which leads to a competitive ratio of $2-\delta$ for the new interval, so the overall competitive ratio is not reduced. If $\mathcal{B}$ reaches the state~$1$ (case 2), then we can use the same construction but the states and functions are switched, i.e., the next function is $\varphi_0$. This is possible, since both the algorithm $\mathcal{B}$ and the adversary strategy are symmetrical to $x = 0.5$. 
	
	Each workload extension (case 1 and 2) increases the total cost of $\mathcal{B}$ by at least $\epsilon$, because the adversary sends at least one $\varphi_0$ and one $\varphi_1$ function, so $\mathcal{B}$ switches its state at least two times incurring a switching cost of $2 \cdot \epsilon/2$. By repeating case~1 or~2, the total cost converges to infinity. 
	
	Case 3 already contains an arbitrarily long workload. Algorithm~$\mathcal{B}$ does not reach $0$ or $1$ in case~3 by definition, so the total cost of $\mathcal{B}$ converges to infinity.
	
	Therefore, for all $\alpha \geq 0$ there exists a workload such that $C(\mathcal{B}) \geq (2-\delta) \cdot C(X^\ast) + \alpha$ holds. 
\end{proof}

So far, we have shown that the competitive ratio of algorithm $\mathcal{B}$ is at least $2 - \delta$ for an arbitrary small $\delta > 0$. Now, in the second part of the proof of Theorem~\ref{theo:lower:cont:two}, we will show that any deterministic online algorithm that differs from $\mathcal{B}$ causes more cost than $\mathcal{B}$. Thus, $2$ is a lower bound for the competitive ratio of the continuous data-center optimization problem.

\begin{lemma} \label{lemma:lower:cont:nobetter}
	Any deterministic online algorithm $\mathcal{A}$ that  differs from the states of $\mathcal{B}$ produces more cost than $\mathcal{B}$, so $C(\mathcal{A}) \geq C(\mathcal{B})$.
\end{lemma}

\begin{proof}
	Let $\mathcal{A}$ be an arbitrary online algorithm.  
	The states of $\mathcal{A}$ are denoted by $a_t$. The adversary will use the following strategy: It sends $\varphi_1$ functions as long as $a_t \leq b_t$ and $a_t < 1$. If $a_t > b_t$, the adversary will send $\varphi_0$. 
	If $a_t$ reaches $1$, the adversary will send $\varphi_0$.
	
	We divide the resulting function sequence $F$ into time intervals $I_1, I_2, \dots$ of maximal size such that the adversary sends the same function for each time slot in the interval. The set $\mathcal{J} \coloneqq \{I_1, I_2, \dots\}$ is a partition of $[T]$. 
	Let $\mathcal{T}_1 \coloneqq \{I_1, I_3, I_5, \dots\}$ be the set of odd intervals and let $\mathcal{T}_0 \coloneqq \{I_2, I_4, \dots\}$ be the set of even intervals. 
	The first function $f_1$ is always $\varphi_1$, since $a_0 = b_0 = 0$, so for each odd interval $I \in \mathcal{T}_1$, we have $f_t = \varphi_1$ for all $t \in I$; for even intervals $I \in \mathcal{T}_0$, we have $f_t = \varphi_0$ for all $t \in I$. 
	
	Let $S_I(X) \coloneqq \sum_{t \in I} |x_t - x_{t-1}|$ denote the switching cost for algorithm $X$ during the time interval $I = [u:v]$ including the switching cost from $u-1$ to $u$. 
	Let $C_I(X) \coloneqq S_I(X) + \sum_{t \in I} f_t(x_t)$ denote the total cost of $X$ during $I$. 
	Note that contrary to Section~\ref{sec:poly}, this definition includes the switching cost from $u-1$ to $u$. 
	
	We will show that $C_I(\mathcal{A}) \geq C_I(\mathcal{B})$ holds for all intervals $I \in \mathcal{J}$. We differ between finite and infinite intervals. The last interval is infinite, if $\mathcal{A}$ permanently stays below or above $\mathcal{B}$. 
	
	For a finite even interval $I = [u:v] \in \mathcal{T}_0$, we have $a_{t-1} \geq b_{t-1}$ for all $t \in I$, because otherwise the adversary had not used the function $\varphi_0$. Furthermore, for the last time slot $v$, we have $a_v \leq b_v$. Let $\delta \coloneqq b_v - a_v$ be the difference between $\mathcal{A}$ and $\mathcal{B}$ at the end of the interval. 
	By the definition of algorithm $\mathcal{B}$, we have $b_t \leq b_{t-1}$ because $b_t = \max \{b_{t-1} - \epsilon/2, 0\}$, so the switching cost of $\mathcal{B}$ during~$I$ is exactly $S_I(\mathcal{B}) = b_u - b_v$.  
	The switching cost of $\mathcal{A}$ during $I$ is at least 
	\begin{equation*} 
	S_I(\mathcal{A}) \geq a_u - a_v \geq b_u - b_v + \delta = S_I(\mathcal{B}) + \delta
	\end{equation*}
	since $a_u \geq b_u$ and $a_v = b_v - \delta$. The operating cost of $\mathcal{A}$ for all time slots $t \in I \setminus \{v\}$ is $f_t(a_t) =\epsilon a_t \geq \epsilon b_t = f_t(b_t)$. For the last time slot $v$, we get $f_v(a_v) = \epsilon a_v = \epsilon (b_v - \delta) = f_v(b_v) - \epsilon \delta$. Therefore, the total cost of $\mathcal{A}$ during the time interval $I$ is 
	\begin{equation*} 
	C_I(\mathcal{A}) = S_I(\mathcal{A}) + \sum_{t=u}^{v} f_t(a_t) \geq S_I(\mathcal{B}) + \delta + \sum_{t=u}^{v} f_t(a_t) - \epsilon \delta \geq C_I(\mathcal{B}) .
	\end{equation*}
	The last inequality holds, because we can choose $\epsilon < 1$.
	
	If $\mathcal{A}$ permanently stays above $\mathcal{B}$, i.e., $a_t > b_t$ for all $t \geq u$, then the interval $I$ does not end, so there is no last time slot $v$.  If there is a constant $c > 0$ such that $a_{t} \geq c$ for all $t \in I$, then the operating cost of $\mathcal{A}$ goes towards infinity since $f_t(a_t) \geq \epsilon c > 0$ for all $t \in I$. If there is no such constant, the difference of the switching costs of $\mathcal{A}$ and $\mathcal{B}$ goes towards zero, while the operating cost of $\mathcal{A}$ is greater than the operating cost of $\mathcal{B}$. Thus, in both cases we get $C_I(\mathcal{A}) \geq C_I(\mathcal{B})$.
	
	The proof for an odd interval $I \in \mathcal{T}_1$ is analogous. 
	We have $a_{t-1} \leq b_{t-1}$ for all $t \in I$ and $a_v \geq b_v$ if $I$ is finite. Let $\delta \coloneqq a_v - b_v$. Since $S_I(\mathcal{B}) = b_v - b_u$, we get $S_I(\mathcal{A}) \geq a_v - a_u \geq b_v + \delta - b_u  = S_I(\mathcal{B}) + \delta$. Furthermore, for all $t \in I \setminus \{v\}$, we have $f_t(a_t) = \epsilon (1 - a_t) \geq \epsilon (1 - b_t) = f_t(b_t)$. For the last time slot $v$, we get $f_v(a_v) = \epsilon (1- a_v) = \epsilon (1 - b_v - \delta) = f_v(b_v) - \epsilon \delta$. Therefore, $C_I(\mathcal{A}) = S_I(\mathcal{A}) + \sum_{t=u}^{v} f_t(a_t) \geq S_I(\mathcal{B}) + \delta + \sum_{t=u}^{v} f_t(a_t) - \epsilon \delta \geq C_I(\mathcal{B})$.

	If $I \in \mathcal{T}_1$ is infinite, then either there is a constant $c < 1$ with $a_t \leq c$, so the operating cost of $\mathcal{A}$ converges to infinity since $f_t(a_t) \geq \epsilon (1 - c) > 0$, or there is no such constant, so the difference of the switching costs of $\mathcal{A}$ and $\mathcal{B}$ goes towards zero, while the operating cost of $\mathcal{A}$ is greater than the operating cost of $\mathcal{B}$. Thus, $C_I(\mathcal{A}) \geq C_I(\mathcal{B})$ is always fulfilled.
	
	By adding the total cost of all intervals, we get
	$C(\mathcal{A}) = \sum_{I \in \mathcal{J}} C_I(\mathcal{A}) \geq \sum_{I \in \mathcal{J}} C_I(\mathcal{B}) = C(\mathcal{B})$.
\end{proof}

\begin{proof}[\textbf{Proof of Theorem~\ref{theo:lower:cont:two}}]
	Let $\mathcal{A}$ be an arbitrary deterministic online algorithm. By using Lemmas~\ref{lemma:lower:cont:blowerbound} and~\ref{lemma:lower:cont:nobetter}, we get 
	\begin{equation*}
	C(\mathcal{A}) 
	\alignstack{\text{\tiny L\ref{lemma:lower:cont:nobetter}}}{\,\geq\,} C(\mathcal{B}) 
	\alignstack{\text{\tiny L\ref{lemma:lower:cont:blowerbound}}}{\,\geq\,} (2-\delta)  \cdot C(X^\ast) + \alpha
	\end{equation*} for all $\delta > 0$ and $\alpha \geq 0$. 
\end{proof}

\subsubsection{Restricted model}
\label{sec:lower:cont:lin}

Analogously to the discrete setting, in this section we want to show that the lower bound of 2 for the continuous data-center optimization problem still holds for the restricted model described in Section~\ref{sec:lower:deterministic:lin}.

\begin{theorem} \label{theo:lower:cont:lin}
	There is no deterministic online algorithm for the continuous setting of the restricted model with a competitive ratio of $c < 2$. 
\end{theorem}

\begin{proof}
	The restricted model is denoted by $\mathcal{L}$, the general model is denoted by $\mathcal{G}$.
	Let $f(z) \coloneqq \epsilon|1-kz|$ with $\epsilon \rightarrow 0$ and $k \rightarrow \infty$, let $\beta = 2$. If the adversary in $\mathcal{G}$ sends $\varphi_0(x) = \epsilon|x|$ as function, then we will use $\lambda_t = l_0 \coloneqq 0$ which leads to operating costs of
	\begin{equation*}
	x_t f(l_0 / x_t) = x_t \cdot \epsilon |1| = \epsilon|x_t| .
	\end{equation*}
	The last equality holds, because $x_t \geq l_0 = 0$. If the adversary sends $\varphi_1(x) = \epsilon|1-x|$, then we will use $\lambda_t = l_1 \coloneqq 1/k$ which leads to operating costs of 
	\begin{equation*}
	x_t f(l_1 / x_t) = x_t \cdot \epsilon \left|1 - \frac{1}{x_t}\right| = \epsilon|x_t - 1| = \epsilon|1-x_t| .
	\end{equation*}
	Hence, difference (1) between both models is solved. 
	
	The additional condition that $x_t \geq \lambda_t$ does not change anything, because both $l_0$ and $l_1$ are arbitrary close to 0 as $k \rightarrow \infty$, so difference (2) is solved too.
\end{proof}

\subsection{Discrete setting, randomized algorithms}
\label{sec:lower:random}

In this section, we determine a lower bound for randomized online algorithms in the discrete setting.
We begin with the analysis of the general model and afterwards show how our construction can be adapted to the restricted model (see Section~\ref{sec:lower:random:lin}).

\subsubsection{General model}
\label{sec:lower:random:general}

In this section, we show that there is no randomized online algorithm that achieves a competitive ratio that is smaller than 2 in the discrete setting against an oblivious adversary. The construction is similar to the continuous setting (Section~\ref{sec:lower:cont:general}). We have only one single server and the adversary will use the functions $\varphi_0(x) = \epsilon |x|$ and $\varphi_1(x) = \epsilon |1-x|$ with $\epsilon > 0$ and $\epsilon^{-1} \in \mathbb{N}$. 

The lower bound is proven as follows: First, we will construct an algorithm $\mathcal{B}$ that solves the continuous setting with a competitive ratio of at least $2-\delta$ for an arbitrary small $\delta > 0$. Then, we consider an arbitrary randomized online algorithm $\mathcal{A}$ for the discrete setting and show how to convert its probabilistic discrete schedule $X^\mathcal{A}$ to a deterministic continuous schedule $\bar{X}^\mathcal{A}$ without increasing the cost. Finally, we show how the adversary constructs the problem instance in dependence on the current state of $\bar{X}^\mathcal{A}$ and $\bar{X}^\mathcal{B}$.

Consider algorithm $\mathcal{B}$ described in Section~\ref{sec:lower:cont:general}. By Lemma~\ref{lemma:lower:cont:blowerbound:plus}, the competitive ratio of $\mathcal{B}$ for the continuous setting is at least $2 - \delta$ for an arbitrary small $\delta > 0$. Formally, 
\begin{equation} \label{eqn:random:lower:bopt} 
C^{\bar{\mathcal{P}}}(\bar{X}^\mathcal{B}) \geq (2-\delta) \cdot C^{\bar{\mathcal{P}}}(\bar{X}^\ast) + \alpha
\end{equation}
for all $\delta > 0$ and $\alpha \geq 0$. 

Let $\mathcal{A}$ be an arbitrary randomized online algorithm and let $\mathcal{P}$ be the problem instance created by the adversary (we will define later, how this problem instance is determined). For each time slot~$t$, the oblivious adversary knows the probability $\bar{x}^\mathcal{A}_t$ that $\mathcal{A}$ is in state 1. Note that there is only one server, so the probability that $\mathcal{A}$ is in state 0 is given by $1 - \bar{x}^\mathcal{A}_t$. Now, consider the fractional schedule $\bar{X}^\mathcal{A} = (\bar{x}^\mathcal{A}_1, \dots, \bar{x}^\mathcal{A}_T)$. The following lemma shows that the cost of $\bar{X}^\mathcal{A}$ for the continuous problem instance $\bar{\mathcal{P}}$ is smaller than or equal to the expected cost of $\mathcal{A}$ for the discrete problem instance $\mathcal{P}$.

\begin{lemma} \label{lemma:random:lower:expectedA}
	$\mathbb{E}[C^\mathcal{P}(X^\mathcal{A})] \geq C^{\bar{\mathcal{P}}}(\bar{X}^\mathcal{A})$.
\end{lemma}

\begin{proof}
	First, we will analyze the operating costs. The expected operating cost of $X^\mathcal{A}$ for time slot $t$ is 
	\begin{equation*}
	\mathbb{E}[f_t(x^\mathcal{A}_t)] = (1 - \bar{x}^\mathcal{A}_t)  f_t(0) + \bar{x}^\mathcal{A}_t  f_t(1) = \bar{f}_t(\bar{x}^\mathcal{A}_t) .
	\end{equation*}
	The last term describes the operating cost of $\bar{X}^\mathcal{A}$ in the continuous setting for time slot $t$. Thus, $\mathbb{E}[R^\mathcal{P}(X^\mathcal{A})] = R^{\bar{\mathcal{P}}}(\bar{X}^\mathcal{A})$.
	
	The switching cost of $\bar{X}^\mathcal{A}$ for time slot $t$ is $|\bar{x}^\mathcal{A}_t - \bar{x}^\mathcal{A}_{t-1}|$. The probability that $X^\mathcal{A}$ switches its state from 0 to 1 is at least $(\bar{x}^\mathcal{A}_t - \bar{x}^\mathcal{A}_{t-1})^+$. Analogously, the probability for switching the state from  1 to 0 is at least $(\bar{x}^\mathcal{A}_{t-1} - \bar{x}^\mathcal{A}_t)^+$. The actual probability can be greater, because we do not know the exact behavior of $\mathcal{A}$. All in all, the probability that $X^\mathcal{A}$ switches  its state is at least $|\bar{x}^\mathcal{A}_t - \bar{x}^\mathcal{A}_{t-1}|$, so over all time slots we get $\mathbb{E}[S^\mathcal{P}(X^\mathcal{A})] \geq S^{\bar{\mathcal{P}}}(\bar{X}^\mathcal{A})$ and therefore $\mathbb{E}[C^\mathcal{P}(X^\mathcal{A})] \geq C^{\bar{\mathcal{P}}}(\bar{X}^\mathcal{A})$.
\end{proof}

Now we have constructed a continuous schedule $\bar{X}^\mathcal{A}$ from the probabilities of $X^\mathcal{A}$.
The adversary behaves like in Section~\ref{sec:lower:cont:general}, that is, if $\bar{x}^\mathcal{A}_t$ equals 1 or 0, it will send $\varphi_0$ or $\varphi_1$, respectively, and otherwise if $\bar{x}^\mathcal{A}_t$ is greater than or smaller than  $\bar{x}^\mathcal{B}_t$, it will send $\varphi_0$ or $\varphi_1$. If $\bar{x}^\mathcal{A}_t = \bar{x}^\mathcal{B}_t$, then the adversary can choose an arbitrary state. By Lemma~\ref{lemma:lower:cont:nobetter},
\begin{equation} \label{eqn:random:lower:ab}
C^{\bar{\mathcal{P}}}(\bar{X}^\mathcal{A}) \geq C^{\bar{\mathcal{P}}}(\bar{X}^\mathcal{B})
\end{equation}
holds. 
Now, we are able to prove that 2 is a lower bound for randomized online algorithms.

\begin{theorem} \label{theo:lower:random:two}
	There is no randomized online algorithm for the discrete data-center optimization problem that achieves a competitive ratio that is less than 2 against an oblivious adversary.
\end{theorem}

\newcommand{\Ccont}{C^{\bar{\mathcal{P}}}}
\newcommand{\barX}{\bar{X}}
\begin{proof}
	Let $\mathcal{A}$ be an arbitrary randomized online algorithm. By using Lemma~\ref{lemma:random:lower:expectedA} as well as equations~\eqref{eqn:random:lower:bopt}, \eqref{eqn:random:lower:ab} and~\eqref{eqn:random:analysis:cxopt}, we get
	\begin{align*} 
	\mathbb{E}[C^\mathcal{P}(X^\mathcal{A})] 
	&\alignstack{\text{\tiny L\ref{lemma:random:lower:expectedA}}}{\,\geq\,} \Ccont (\bar{X}^\mathcal{A}) \\
	&\alignstack{\eqref{eqn:random:lower:ab}}{\,\geq\,} \Ccont (\barX^\mathcal{B}) \\
	&\alignstack{\eqref{eqn:random:lower:bopt}}{\,\geq\,} (2-\delta) \cdot \Ccont (\barX^\ast) + \alpha \\
	&\alignstack{\eqref{eqn:random:analysis:cxopt}}{\,=\,} (2-\delta) \cdot C^\mathcal{P}(X^\ast) + \alpha
	\end{align*}
	where $\delta > 0$ and $\alpha \geq 0$ can be chosen arbitrarily. 
\end{proof}

The theorem shows that the randomized algorithm given in Section~\ref{sec:random:algo} is optimal.

\subsubsection{Restricted model}
\label{sec:lower:random:lin}

In this section, we show that the lower bound of~2 presented above still holds for the restricted model. The basic idea is very similar to the proof of Theorem~\ref{theo:lower:disc:lin} in Section~\ref{sec:lower:deterministic:lin}.

\begin{theorem} \label{theo:lower:random:lin}
	There is no randomized online algorithm for the discrete setting of the restricted model with a competitive ratio of $c < 2$.
\end{theorem}

\begin{proof}
	The general model is denoted by $\mathcal{G}$ and the restricted model is denoted by $\mathcal{L}$. The states of the model $\mathcal{X} \in \{\mathcal{G}, \mathcal{L}\}$ are indicated by $x_t^\mathcal{X}$. In the restricted model we use 2 servers, the operating cost function $f(z) \coloneqq \epsilon|1 - 2z|$ with $\epsilon \rightarrow 0$ and $\beta = 2$. Instead of switching between the states $0$ and $1$ in $\mathcal{G}$, we will switch between $1$ and $2$ in $\mathcal{L}$, so for $t \in [T]$ we have $x^\mathcal{L}_t = x^\mathcal{G}_t + 1$.
	
	If the adversary in $\mathcal{G}$ sends $\varphi_0(x)$ as function, then we will use $\lambda_t = l_0 \coloneqq 0.5$, and if he sends $\varphi_1(x)$, then we will use $\lambda_t = l_1 \coloneqq 1$. As already shown in the proof of Theorem~\ref{theo:lower:disc:lin}, the operating cost $x_t^\mathcal{L} f(l_k / x_t^\mathcal{L})$ in $\mathcal{L}$ (with $k \in \{0,1\}$) is equal to the operating cost $\varphi_k(x_t^\mathcal{G})$ in $\mathcal{G}$. 
	
	In the continuous extension of $\mathcal{L}$ we are allowed to use the states $\bar{x}^\mathcal{L}_t \geq 0.5$, if $\lambda_t = l_0 = 0.5$. Since $x_0 = 0$, the first function the adversary sends in $\mathcal{G}$ is $\varphi_1$, so we have $\lambda_1 = 1$ and thus even in the continuous extension $\bar{x}_t^\mathcal{L} \geq 1$ must be fulfilled. For $t \geq 2$, there is no benefit to use states smaller than 1 in $\mathcal{L}$, since $\bar{x}_t^\mathcal{L} f(l_0 / \bar{x}_t^\mathcal{L}) = \epsilon |\bar{x}_t^\mathcal{L} - 1|$ which is minimal for $\bar{x}_t^\mathcal{L} = 1$. Moving to states below $1$ always increases the operating and switching costs. Therefore, the inequality $x_t \geq \lambda_t$ is  always fulfilled. 
\end{proof}

\subsection{Online algorithms with prediction window}
\label{sec:lower:prediction}
So far, we have considered online algorithms that at time $t$ only know the arriving function $f_t$ in determining the next state. In contrast, an offline algorithm knows the whole function sequence~$F$. There are models between these edge cases. An online algorithm with a \emph{prediction window} of length~$w$,
at any time $t$, can not only use the function $f_t$ but the function set $\{f_t, \dots, f_{t + w}\}$ to choose the state $x_t$. This problem
extension was also defined by Lin et al.~\cite{LinWierman2011infocom,LinWierman2013}.
If $w$ has a constant size (i.e., $w$ is independent of $T$), then the lower bounds developed in the previous sections
still hold as the following theorem shows. We will prove the lower bounds for the restricted model, thus they hold for 
the general model as well.

\begin{theorem} \label{theo:lower:prediction}
	Let $w \in \mathbb{N}$ and $\delta > 0$ be arbitrary constants. There is no deterministic online algorithm with a prediction window of length $w$ that achieves a competitive ratio of $3-\delta$ in the discrete setting or $2-\delta$ in the continuous setting for the restricted model.
\end{theorem}

\begin{proof}
	Let $c$ be the lower bound for the competitive ratio without prediction window, i.e., we have $c=2$ for the continuous and the randomized discrete setting and $c=3$ for the deterministic discrete setting. By Theorem~\ref{theo:lower:cont:lin} and~\ref{theo:lower:disc:lin}, there exists a function sequence $F$ such that there is no online algorithm that achieves a competitive ratio of $c-\delta/2$ for an arbitrary small $\delta > 0$. Let $\mathcal{A}$ be an optimal online algorithm without prediction window and let $\mathcal{B}_w$ be an online algorithm with a prediction window of length $w \geq 1$. We will construct a function sequence $F'$ such that the competitive ratio of $\mathcal{B}_w$ is at least $c-\delta$. 
	
	Let $n \in \mathbb{N}$. Each function $f_t$ in $F$ is replaced by the function sequence $(f'_{t, 1}, \dots, f'_{t, n\cdot w})$ with $f'_{t,u}(z)  \coloneqq \frac{1}{nw} f_t(z)$ where $u \in [n \cdot w]$. So we have
	\begin{equation*}
	F' = (f'_{1,1}, \dots, f'_{1, nw}, f'_{2,1}, \dots, f'_{2, nw}, \dots, f'_{T,1}, \dots, f'_{T, nw}) .
	\end{equation*}
	Since the functions in the subsequence $(f'_{t, 1}, \dots,  f'_{t, n\cdot w})$ are equal and since 
	\begin{equation*}
	\sum_{u=1}^{n\cdot w} f'_{t,u}(x) = f_t(x)
	\end{equation*}
	holds for all $t \in [T]$ and $x \in \mathbb{R}$, the cost of an optimal online algorithm without prediction window are equal for both function sequences, i.e., $C^{F}(\mathcal{A}) = C^{F'}(\mathcal{A})$. Furthermore, the inequality $C^{F}(X^\ast) \geq C^{F'}(X^\ast)$ holds, because in $F'$ we have more possibilities to choose from.

	Only for the last $w$ functions in the sequence $(f'_{t, 1}, \dots, f'_{t, n w})$ the algorithm~$\mathcal{B}_w$ has an extra knowledge in comparison to $\mathcal{A}$. The operating cost of~$\mathcal{B}_w$ is at least zero for these functions, so we can bound the cost of $\mathcal{B}_w$ by
	\begin{align*}
	C^{F'}(\mathcal{B}_w) 
	&\geq \frac{(n-1)\cdot w}{nw} \cdot C^{F'}(\mathcal{A}) \\
	&=  \left(1 - 1/n \right) \cdot C^{F}(\mathcal{A}) \\
	&> \left(1 - 1/n\right)\left(c- \delta/2 \right) \cdot C^{F}(X^\ast) \\
	&=\left(c - \delta/2 - \frac{c - \delta/2}{n} \right) \cdot C^{F}(X^\ast) \\
	&>\left(c - \delta/2 - c/n\right) \cdot C^{F}(X^\ast) \\
	&\geq \left(c - \delta/2 - c/n\right) \cdot C^{F'}(X^\ast) .
	\end{align*}
	By using $n \coloneqq \lceil 2c/\delta \rceil$, we get
	\begin{align*}
	C^{F'}(\mathcal{B}_w) > (c - \delta) \cdot C^{F'}(X^\ast) .
	\end{align*}
	Thus, there is no online algorithm with a prediction window of length $w$ that achieves a competitive ratio of $c-\delta$. 
\end{proof}

\section{Summary}
\label{sec:summary}

This paper examined the data-center optimization problem with homogeneous servers. In contrast to the publications of Lin et al. \cite{LinWierman2011infocom,LinWierman2013} and Bansal et al. \cite{Bansal2015}, we studied the setting where only integral solutions are allowed, i.e., the number of active servers must be an integer. 
We developed an $\mathcal{O}(T \cdot \log m)$ time algorithm for the offline version of the problem. Furthermore, we showed how to adapt Lin et al.'s deterministic online algorithm for the discrete setting and proved that it is still 3-competitive. In addition, we presented a randomized algorithm with a competitive ratio of~2 against an oblivious adversary. At the end, we showed that both the deterministic and the randomized algorithm are optimal for the discrete setting. Independently of \cite{Antoniadis2017}, we gave a lower bound of 2 for the continuous data-center optimization problem. We proved that all lower bounds still hold for the more restricted model introduced by Lin e. al. \cite{LinWierman2011infocom} and also for online algorithms with a finite prediction window. 

\appendix

\section{Variables and notation}
\label{sec:appendix:variables}


Let $k, l \in \mathbb{N}_0$, let $x, a, b \in \mathbb{R}$ and let $g : \mathbb{N} \rightarrow \mathbb{R}$ be an arbitrary function. 
\begin{small}
\begin{align*}
[k] &\coloneqq \{1, 2, \dots k\} &
[k:l] &\coloneqq \{k, k+1, \dots, l\} &
[k:l[ &\coloneqq \{k, k+1, \dots, l-1\} \\
[k]_0 &\coloneqq \{0, 1, \dots k\} &
]k:l] &\coloneqq \{k+1, k+2, \dots, l\} &
]k:l[ &\coloneqq \{k+1, k+2, \dots, l-1\} \\
[x]^a_b &\coloneqq \max\{a, \min\{b, x\}\} &
\text{frac}(x) &\coloneqq x - \lfloor x \rfloor &
\Delta g(x) &\coloneqq g(x) - g(x-1) \\
\ceilstrong{x} &\coloneqq \min \{n \in \mathbb{Z} \mid n > x\}
\end{align*}%
\end{small}%

\vspace{10pt} \noindent
The following table gives an overview of the variables defined in this paper.

\begingroup
	\setlength{\tabcolsep}{4pt}
	\renewcommand{\arraystretch}{1.3}
	\scriptsize
	
	\centering
	\begin{longtable}{|l|p{11.6cm}|} 
		\hline
		\textbf{Variable} & \textbf{Description} \\
		\hline
		\endhead
		\hline
		\endfoot
		$a_t$ & Number of active servers of algorithm $\mathcal{A}$ at time $t$. \\
		$\mathcal{A}$ & An arbitrary online algorithm. \\
		$b_t$ & Number of active servers of algorithm $\mathcal{B}$ at time $t$. Note that $b_t$ can be fractional. \\
		$\mathcal{B}$ & Online algorithm that achieves a competitive ratio of $2-\delta$ in the continuous setting if the adversary only uses $\varphi_0$ and $\varphi_1$ as functions. \\
		$\mathcal{B}_w$ & An arbitrary online algorithm with a prediction window of length $w$. \\
		$\beta$ & Switching cost. \\
		$C(X), C(\mathcal{A})$ & Total cost of the schedule $X$ or algorithm $\mathcal{A}$, respectively. \\
		$C^F(X), C^F(\mathcal{A})$ & Total cost of the schedule $X$ or algorithm $\mathcal{A}$ for the function sequence $F$. \\
		$C_{I}(X), C_{I}(\mathcal{A})$ & Total cost of the schedule $X$ or algorithm $\mathcal{A}$ during the time interval $I = \{a, a+1, \dots, b\}$ \emph{including} the switching cost from $a-1$ to $a$. \\
		$C^L_\tau(X), C^U_\tau(X)$ & Cost of the schedule $X$ up to time $\tau$ if the switching cost is paid for powering up ($L$) or down ($U$), respectively. \\
		$\hat{C}^L_\tau(x)$, $\hat{C}^U_\tau(x)$ & Minimal cost up to time $\tau$ that can be achieved if the last state $x_\tau$ is $x$ and if the switching cost is paid for powering up~($L$) or down~($U$), respectively. \\
		$C^Q(X)$ & Total cost of the schedule $X$ in the problem instance $Q$. \\
		$C^Q_{[a,b]}(X)$ & Total cost of the schedule $X$ in the problem instance $Q$ during the time interval $\{a, a+1, \dots, b\}$ \emph{excluding} the switching cost from $a-1$ to $a$. \\
		$f(z)$ & Operating cost of a single server running with load $z \in [0,1]$ in the \emph{restricted} model. \\
		$f_t(x_t)$ & Operating cost at time slot $t$ for $x_t$ active servers. \\
		$\bar{f}_t(x_t)$ & Continuous extension of $f_t$, see equation~\eqref{eqn:poly:correct:extension}. \\
		$F$ & Sequence of operating cost functions, $F = (f_1, \dots, f_t)$. \\
		$\mathcal{G}$ & General model described by equation~\eqref{eqn:model:cost}. \\
		$K$ & First iteration of the polynomial offline algorithm. All in all, the algorithm performs $K+1$ iterations. \\ 
		$\mathcal{L}$ & Restricted model described by equation~\eqref{eqn:model:lin}. \\
		$\lambda_t$ & Incoming workload at time $t$ in the \emph{restricted} model. \\
		$m$ & Total number of servers. \\
		$M_k$ & States used in the problem instance $\mathcal{P}_k$. Formally, $M_k \coloneqq \{n \in [m]_0 \mid n \mod 2^k = 0\}$. \\ 
		$N_0(t), N_1(t)$ & Number of time slots up to time slot $t$ where the adversary sent the function $\varphi_0$ or $\varphi_1$, respectively. \\
		$p^\uparrow_t$ & Probability that the randomized offline algorithm uses the upper state, if the number of active servers in $\bar{X}$ increases.  \\ 
		$p^\downarrow_t$ & Probability that the randomized offline algorithm uses the lower state, if the number of active servers in $\bar{X}$ decreases.\\
		$\mathcal{P}$ & Original problem instance, $\mathcal{P} = (T, m, \beta, F)$ or $\mathcal{P} = (T, m, \beta, F, M)$ (in Section~\ref{sec:poly}). \\
		$\bar{\mathcal{P}}$ & Continuous extension of the problem instance $\mathcal{P}$. \\
		$\mathcal{P}_k$ & Problem instance where only states are allowed that are multiples of $2^k$. Formally, $\mathcal{P}_k = \Phi_k(\mathcal{P})$. \\
		$\varphi_0, \varphi_1$ & Functions for constructing lower bounds. Formally, $\varphi_0(x) \coloneqq \epsilon |x|$ and $\varphi_1(x) \coloneqq \epsilon |1-x|$. \\
		$\Phi_k(Q)$ & Modified version of the problem instance $Q$, where the states must be multiples of $2^k$. \\
		$R_\tau(X)$ & Operating cost of the schedule $X$ up to time $\tau$. \\
		$R^Q(X)$ & Operating cost of the schedule $X$ in the problem instance $Q$. \\
		$S$ & Total switching cost of algorithm $\mathcal{A}$. \\
		$S_I(X)$, $S_I(\mathcal{A})$ & Switching cost of the schedule $X$ or algorithm $\mathcal{A}$ during the time interval $I = \{a, a+1, \dots, b\}$ \emph{including} the switching cost from $a-1$ to $a$. \\
		$S^L_\tau(X), S^U_\tau(X)$ & Switching cost of $X$ up to time $\tau$ if the switching cost is paid for powering up ($L$) or down ($U$), respectively. \\
		$S^Q(X)$ & Switching cost of the schedule $X$ in the problem instance $Q$. \\
		$T$ & Total number of time slots. \\
		$\mathcal{T}^-$, $\mathcal{T}^+$ & Decreasing ($\mathcal{T}^-$) or increasing ($\mathcal{T}^+$) time intervals of the LCP algorithm. During a time interval in $\mathcal{T}^-$ or $\mathcal{T}^+$, the number of active servers in both $X^\text{LCP}$ and $X^\ast$ never increases or decreases, respectively. \\
		$v_{t,j}$ & Vertex in $G$ representing $j$ active servers at time slot $t$. \\
		$V^k$ & Vertex set used in iteration $k$ of the polynomial offline algorithm. \\
		$w$ & Length of the prediction window. \\
		$x_t$ & Number of active servers in the schedule $X$ at time $t$. \\
		$x^\ast_t$ & Number of active servers in the optimal schedule $X^\ast$ at time $t$. \\
		$\hat{x}^k_t$ & Number of active servers at time slot $t$ in iteration $k$ of the polynomial offline algorithm. \\
		$x^\text{LCP}_t$ & Number of active servers of the LCP algorithm at time $t$. \\
		$x^L_\tau$, $x^U_\tau$ & Last state of $X^L_\tau$ or $X^U_\tau$, respectively. Formally, $x^B_\tau \coloneqq x^B_{\tau, \tau}$ for $B \in {L, U}$ where $X^B_\tau = (x^B_{\tau, 1}, \dots, x^B_{\tau, \tau})$. \\
		$x^L_{\tau,t}$, $x^U_{\tau,t}$ & Number of active server at time $t$ in the schedule $X^L_\tau = (x^L_{\tau, 1}, \dots, x^L_{\tau, \tau})$ or $X^U_\tau = (x^U_{\tau, 1}, \dots, x^U_{\tau, \tau})$, respectively. \\
		$X$ & Schedule, $X = (x_1, \dots, x_T)$. \\
		$X^\ast$ & Optimal schedule, $X^\ast = (x^\ast_1, \dots, x^\ast_T)$. \\
		$\hat{X}^k$ & Schedule that is calculated in iteration $k$ of the polynomial offline algorithm, $\hat{X}^k = (\hat{x}^k_1, \dots, \hat{x}^k_T)$. \\
		$X^\text{LCP}$ & Schedule of the LCP algorithm, $X^\text{LCP} = (x^\text{LCP}_1, \dots, x^\text{LCP}_T)$. \\
		$X^L_\tau, X^U_\tau$ & Schedule that minimizes the cost $C^L_\tau$ or $C^U_\tau$ up to time $\tau$ if the switching cost is paid for powering up~($L$) or down ~($U$), respectively. \\
		$\Psi_l(Q)$ & Scaled version of the problem instance $Q$. State $x$ in $Q$ corresponds to $x/2^l$ in $\Psi_l(Q)$. \\
		$\Omega(Q)$ & Set of optimal schedules for the problem instance $Q$. \\
	\end{longtable}
\endgroup

\bibliographystyle{plainurl}
\bibliography{literature}

\begin{thebibliography}{10}

\bibitem{Andrew2013}
Lachlan Andrew, Siddharth Barman, Katrina Ligett, Minghong Lin, Adam Meyerson,
  Alan Roytman, and Adam Wierman.
\newblock A tale of two metrics: Simultaneous bounds on competitiveness and
  regret.
\newblock In {\em Proc. 26th Annual Conference on Learning Theory (COLT'13)},
  pages 741--763, 2013.

\bibitem{Antoniadis2016}
Antonios Antoniadis, Neal Barcelo, Michael Nugent, Kirk Pruhs, Kevin Schewior,
  and Michele Scquizzato.
\newblock Chasing convex bodies and functions.
\newblock In {\em Proc. 12th Latin American Symposium on Theoretical
  Informatics (LATIN'16)}, pages 68--81. Springer, 2016.

\bibitem{Antoniadis2020}
Antonios Antoniadis, Naveen Garg, Gunjan Kumar, and Nikhil Kumar.
\newblock Parallel machine scheduling to minimize energy consumption.
\newblock In {\em Proceedings of the Fourteenth Annual ACM-SIAM Symposium on
  Discrete Algorithms}, pages 2758--2769. SIAM, 2020.

\bibitem{Antoniadis2017}
Antonios Antoniadis and Kevin Schewior.
\newblock A tight lower bound for online convex optimization with switching
  costs.
\newblock In {\em International Workshop on Approximation and Online
  Algorithms}, pages 164--175. Springer, 2017.

\bibitem{Argue2020}
CJ~Argue, Anupam Gupta, Guru Guruganesh, and Ziye Tang.
\newblock Chasing convex bodies with linear competitive ratio.
\newblock In {\em Proceedings of the Fourteenth Annual ACM-SIAM Symposium on
  Discrete Algorithms}, pages 1519--1524. SIAM, 2020.

\bibitem{Armbrust2009}
Michael Armbrust, Armando Fox, Rean Griffith, Anthony~D Joseph, Randy Katz,
  Andy Konwinski, Gunho Lee, David Patterson, Ariel Rabkin, Ion Stoica, and
  Matei Zaharia.
\newblock Above the clouds: A berkeley view of cloud computing.
\newblock Technical Report No. UCB/EECS-2009-282, EECS Department, University
  of California, Berkeley, 2009.

\bibitem{Bansal2015}
Nikhil Bansal, Anupam Gupta, Ravishankar Krishnaswamy, Kirk Pruhs, Kevin
  Schewior, and Cliff Stein.
\newblock A 2-competitive algorithm for online convex optimization with
  switching costs.
\newblock In {\em Approximation, Randomization, and Combinatorial Optimization.
  Algorithms and Techniques (APPROX/RANDOM 2015)}, volume~40 of {\em LIPIcs},
  pages 96--109. Schloss Dagstuhl - Leibniz-Zentrum f{\"{u}}r Informatik, 2015.

\bibitem{Barroso2007}
Luiz~Andr{\'e} Barroso and Urs H{\"o}lzle.
\newblock The case for energy-proportional computing.
\newblock {\em IEEE Computer}, 40(12):33--37, 2007.

\bibitem{Bawden2016}
Tom Bawden.
\newblock Global warming: Data centres to consume three times as much energy in
  next decade, experts warn, 2016.
\newblock URL:
  \url{http://www.independent.co.uk/environment/global-warming-data-centres-to-consume-three-times-as-much-energy-in-next-decade-experts-warn-a6830086.html}.

\bibitem{Brill2007}
Kenneth~G Brill.
\newblock The invisible crisis in the data center: The economic meltdown of
  moore's law.
\newblock {\em white paper, Uptime Institute}, pages 2--5, 2007.

\bibitem{BubeckSellke2020nested}
S{\'e}bastien Bubeck, Bo'az Klartag, Yin~Tat Lee, Yuanzhi Li, and Mark Sellke.
\newblock Chasing nested convex bodies nearly optimally.
\newblock In {\em Proceedings of the Fourteenth Annual ACM-SIAM Symposium on
  Discrete Algorithms}, pages 1496--1508. SIAM, 2020.

\bibitem{ChenGoelWierman2018}
Niangjun Chen, Gautam Goel, and Adam Wierman.
\newblock Smoothed online convex optimization in high dimensions via online
  balanced descent.
\newblock {\em Proceedings of Machine Learning Research}, 75:1574--1594, 2018.

\bibitem{Dayarathna2016}
Miyuru Dayarathna, Yonggang Wen, and Rui Fan.
\newblock Data center energy consumption modeling: A survey.
\newblock {\em IEEE Communications Surveys \& Tutorials}, 18(1):732--794, 2016.

\bibitem{Delforge2014}
Pierre Delforge and et~al.
\newblock Data center efficiency assessment, 2014.
\newblock URL:
  \url{https://www.nrdc.org/sites/default/files/data-center-efficiency-assessment-IP.pdf}.

\bibitem{Gandhi2011}
Anshul Gandhi and Mor Harchol-Balter.
\newblock How data center size impacts the effectiveness of dynamic power
  management.
\newblock In {\em 49th Annual Allerton Conference on Communication, Control,
  and Computing (Allerton)}, pages 1164--1169. IEEE, 2011.

\bibitem{Gandhi2010}
Anshul Gandhi, Mor Harchol-Balter, and Ivo Adan.
\newblock Server farms with setup costs.
\newblock {\em Performance Evaluation}, 67(11):1123--1138, 2010.

\bibitem{GoelWierman2018}
Gautam Goel and Adam Wierman.
\newblock An online algorithm for smoothed regression and lqr control.
\newblock {\em Proceedings of Machine Learning Research}, 89:2504--2513, 2019.

\bibitem{Haas2015}
Zygmunt~J Haas and Shuyang Gu.
\newblock On power management policies for data centers.
\newblock In {\em IEEE International Conference on Data Science and Data
  Intensive Systems}, pages 404--411. IEEE, 2015.

\bibitem{Hamilton2008}
James Hamilton.
\newblock Cost of power in large-scale data centers., 2008.
\newblock URL:
  \url{http://perspectives.mvdirona.com/2008/11/cost-of-power-in-large-scale-data-centers/}.

\bibitem{Khuller2010}
Samir Khuller, Jian Li, and Barna Saha.
\newblock Energy efficient scheduling via partial shutdown.
\newblock In {\em Proceedings of the 21st Annual ACM-SIAM Symposium on Discrete
  Algorithms}, pages 1360--1372. SIAM, 2010.

\bibitem{LiKhuller2011}
Jian Li and Samir Khuller.
\newblock Generalized machine activation problems.
\newblock In {\em Proceedings of the 22nd Annual ACM-SIAM Symposium on Discrete
  Algorithms}, pages 80--94. SIAM, 2011.

\bibitem{LinWierman2011infocom}
Minghong Lin, Adam Wierman, Lachlan~LH Andrew, and Eno Thereska.
\newblock Dynamic right-sizing for power-proportional data centers.
\newblock In {\em 30th {IEEE} International Conference on Computer
  Communications (INFOCOM'11)}, pages 1098--1106. IEEE, 2011.

\bibitem{LinWierman2011}
Minghong Lin, Adam Wierman, Lachlan~LH Andrew, and Eno Thereska.
\newblock Online dynamic capacity provisioning in data centers.
\newblock In {\em 49th Annual Allerton Conference on Communication, Control,
  and Computing (Allerton)}, pages 1159--1163. IEEE, 2011.

\bibitem{LinWierman2013}
Minghong Lin, Adam Wierman, Lachlan~LH Andrew, and Eno Thereska.
\newblock Dynamic right-sizing for power-proportional data centers.
\newblock {\em IEEE/ACM Transactions on Networking}, 21(5):1378--1391, 2013.

\bibitem{LiuLinWierman2015}
Zhenhua Liu, Minghong Lin, Adam Wierman, Steven Low, and Lachlan~LH Andrew.
\newblock Greening geographical load balancing.
\newblock {\em IEEE/ACM Transactions on Networking}, 23(2):657--671, 2015.

\bibitem{Schmid2009power}
Patrick Schmid and Achim Roos.
\newblock Overclocking core i7: Power versus performance, 2009.
\newblock URL:
  \url{http://www.tomshardware.com/reviews/overclock-core-i7,2268-10.html}.

\bibitem{Sellke2020}
Mark Sellke.
\newblock Chasing convex bodies optimally.
\newblock In {\em Proceedings of the 14th Annual ACM-SIAM Symposium on Discrete
  Algorithms}, pages 1509--1518. SIAM, 2020.

\bibitem{Shehabi2016}
Arman Shehabi, Sarah Smith, Dale Sartor, Richard~E Brown, Magnus Herrlin,
  Jonathan~G Koomey, Eric~R Masanet, Nathaniel Horner, Ines~Lima Azevedo, and
  William Lintner.
\newblock United states data center energy usage report.
\newblock Technical Report LBNL-1005775, Lawrence Berkeley National Laboratory,
  California, 2016.

\bibitem{Wang2015}
Kai Wang, Minghong Lin, Florin Ciucu, Adam Wierman, and Chuang Lin.
\newblock Characterizing the impact of the workload on the value of dynamic
  resizing in data centers.
\newblock {\em Performance Evaluation}, 85:1--18, 2015.

\bibitem{Zhang2018}
Ming Zhang, Zizhan Zheng, and Ness~B Shroff.
\newblock An online algorithm for power-proportional data centers with
  switching cost.
\newblock In {\em IEEE Conference on Decision and Control (CDC)}, pages
  6025--6032. IEEE, 2018.

\end{thebibliography}

\end{document}